\documentclass{egpubl}
\usepackage{eg2027}
 
\ConferenceSubmission   % uncomment for Conference submission
\usepackage[T1]{fontenc}
\usepackage{dfadobe}  
\usepackage{microtype}  

\usepackage{cite}  % comment out for biblatex with backend=biber
\BibtexOrBiblatex
\electronicVersion
\PrintedOrElectronic
\ifpdf \usepackage[pdftex]{graphicx} \pdfcompresslevel=9
\else \usepackage[dvips]{graphicx} \fi

\usepackage{egweblnk} 
\usepackage{adjustbox}
\usepackage{amsmath}
\usepackage{amsfonts}
\usepackage{algpseudocode}
\usepackage{algorithm}
\usepackage{pifont}
\usepackage{amsthm}
\usepackage{todonotes}
\usepackage{booktabs}
\usepackage[table]{xcolor}
\usepackage{longtable}
\usepackage{multirow}
\usepackage{subcaption}
\theoremstyle{definition}

\newtheorem{prop}{Proposition}
\newtheorem{theorem}{Theorem}

\usepackage{xcolor}
\definecolor{dark_green}{RGB}{0,100,0}

\title[VoroDuce]%
      {VoroDuce: Massively Parallel 3D Voronoi Simplification with Volume-Aware Edge Collapse}

\author[F. Windisch et al.]
{\parbox{\textwidth}{\centering
  F. Windisch$^{1}$\orcid{0009-0001-8174-9907},
  L. Radl$^{1}$\orcid{0009-0008-4075-5877},
  T. K\"{o}hler$^{1}$\orcid{0000-0002-4654-3241},
  S. Govindarajan$^{2}$\orcid{0000-0002-3546-8223},
  M. Pur$^{1}$,
  A. Tagliasacchi$^{2,3,4}$\orcid{0000-0002-2209-7187}
  and M. Steinberger$^{1,5}$\orcid{0000-0001-5977-8536}
  }
  \\
  {\parbox{\textwidth}{\centering
    $^1$Graz University of Technology, Austria \quad
    $^2$Simon Fraser University, Canada \quad
    $^3$University of Toronto, Canada \\
    $^4$Wayve Labs, Canada \quad
    $^5$Huawei Technologies, Austria
  }
}}
\usepackage{thmtools, thm-restate}
\begin{document}

% uncomment for using teaser
% \teaser{
%  \includegraphics[width=0.9\linewidth]{eg_new}
%  \centering
%   \caption{New EG Logo}
% \label{fig:teaser}
%}

\maketitle
%-------------------------------------------------------------------------
\begin{abstract}
3D Voronoi diagrams provide a powerful mathematical framework for representing volumetric functions and geometry with applications to Materials Science, Biology, Simulation and Astrophysics. 
However, accurately resolving fine structures often requires millions of cells, imposing severe computational and memory constraints. 
Traditional mesh decimation or point cloud simplification techniques fail to account for the effects re-meshing has on neighbouring cells.
In this work, we formalize the Voronoi reduction problem and propose a greedy edge-collapsing algorithm on the dual Delaunay tetrahedralisation that accounts for changes in the local geometry. 
%To address these shortcomings, we introduce VoroDuce, the first GPU-accelerated reduction algorithm tailored for dense 3D Voronoi representations. 
Further, we demonstrate a novel volumetric error approximation technique that models convex cell intersections as truncated prismatoids, circumventing computationally prohibitive exact volume calculations and enabling massive GPU parallelization. 
Evaluations across Synthetic, Medicine and Material Science data demonstrate that VoroDuce decimates models with up to a million primitives, successfully compressing the representation while preserving the underlying data better than na\"ive decimation approaches.
%-------------------------------------------------------------------------
%  ACM CCS 1998
%  (see https://www.acm.org/publications/computing-classification-system/1998)
% \begin{classification} % according to https://www.acm.org/publications/computing-classification-system/1998
% \CCScat{Computer Graphics}{I.3.3}{Picture/Image Generation}{Line and curve generation}
% \end{classification}
%-------------------------------------------------------------------------
%  ACM CCS 2012
%   (see https://www.acm.org/publications/class-2012)
%The tool at \url{http://dl.acm.org/ccs.cfm} can be used to generate
% CCS codes.
%Example:
%\begin{CCSXML}
%<ccs2012>
%<concept>
%<concept_id>10010147.10010371.10010352.10010381</concept_id>
%<concept_desc>Computing methodologies~Collision detection</concept_desc>
%<concept_significance>300</concept_significance>
%</concept>
%<concept>
%<concept_id>10010583.10010588.10010559</concept_id>
%<concept_desc>Hardware~Sensors and actuators</concept_desc>
%<concept_significance>300</concept_significance>
%</concept>
%<concept>
%<concept_id>10010583.10010584.10010587</concept_id>
%<concept_desc>Hardware~PCB design and layout</concept_desc>
%<concept_significance>100</concept_significance>
%</concept>
%</ccs2012>
%\end{CCSXML}

%\ccsdesc[300]{Computing methodologies~Collision detection}
%\ccsdesc[300]{Hardware~Sensors and actuators}
%\ccsdesc[100]{Hardware~PCB design and layout}

%\printccsdesc   
\end{abstract}  
%-------------------------------------------------------------------------
\section{Introduction}
In 3D geometry representation and processing, Voronoi diagrams serve as a fundamental tool for modelling volumetric structures. By mathematically partitioning three-dimensional space into convex polyhedral cells based on distance to a discrete set of seed points, 3D Voronoi diagrams naturally capture the underlying topology and medial axis of complex shapes. This unique characteristic makes them indispensable across a variety of crucial computational tasks, including high-quality surface meshing~\cite{vorocrust}, controller synthesis~\cite{Blackbox}, multi-agent collision avoidance~\cite{CollisionAvoidance}, and robust shape analysis~\cite{ShapeAnalysis}. %Ultimately, they provide a powerful framework that translates raw point clouds or surface data into a meaningful, mathematically sound volume.

 \begin{figure}
     \centering
     \includegraphics[width=0.55\linewidth, angle=180]{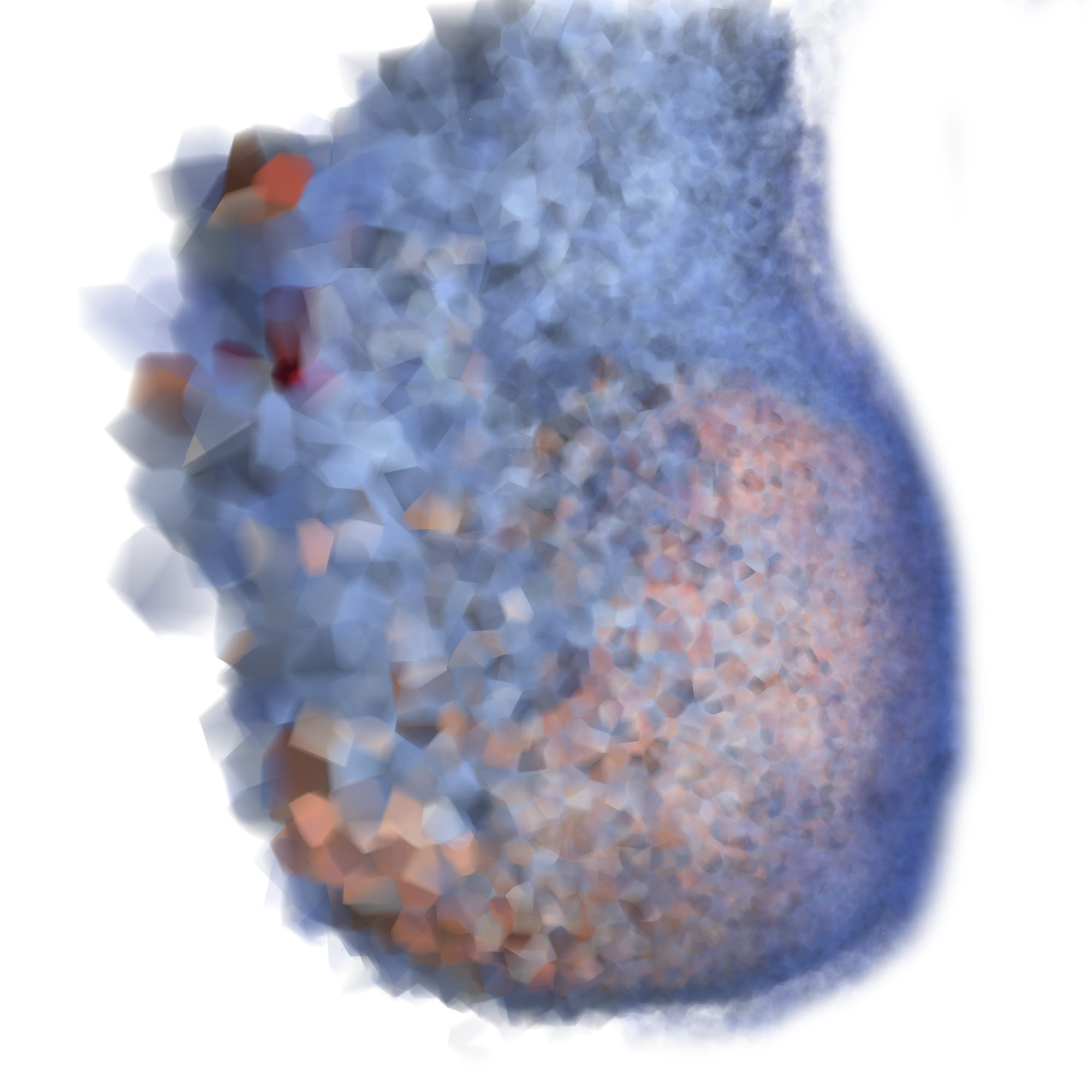}
     \caption{Volumetric function of a brain scan \cite{stanford_volume} represented as a 3D Voronoi diagram that is reduced in detail from left to right.}
     \label{fig:teaser}
 \end{figure}

Beyond purely geometric representation, Voronoi diagrams are useful for approximating and evaluating continuous functions over three-dimensional domains. By associating functional data with the generating seed points, the diagram establishes a piecewise constant volumetric field. When compared to standard 3D triangulations, which are notoriously susceptible to poorly conditioned "sliver" tetrahedra, Voronoi diagrams possess distinctly desirable properties: Every Voronoi cell is mathematically guaranteed to be a convex polytope, providing a well-defined volume ideal for finite volume methods and precise gradient estimation. This intrinsic resilience to highly irregular or anisotropic point distributions allows Voronoi-based function representations to maintain computational stability. While less popular than the dual Delaunay triangulation, Voronoi cells have found several applications for directly representing geometry in Materials Science \cite{Material1, Material2}, Biomedical Engineering~\cite{BioMed}, Architecture~\cite{Architecture}, Astrophysics~\cite{AstroPhysics} and inverse rendering~\cite{radfoam, powerfoam}. While these representations offer high fidelity and flexibility, they do so at a severe computational cost: High-quality Voronoi models routinely require millions of primitives to faithfully resolve fine internal structures. This massive primitive count severely limits their utility in real-time rendering, interactive design, and memory-constrained environments, highlighting a critical need for effective simplification techniques. 

Standard geometry processing algorithms for simplification and decimation are ill-equipped to handle 3D Voronoi representations: Traditional mesh decimation relies on surface connectivity, which is usually absent or prohibitively complex for these models. Similarly, point cloud downsampling techniques typically focus on preserving boundaries or uniform density, largely ignoring the topological interconnectedness that defines the structure of 3D Voronoi representations.

To bridge this gap, we present a novel geometry processing algorithm specifically designed to reduce the site counts of 3D Voronoi representations that approximate some underlying volumetric function while preserving their macroscopic characteristics. Our approach employs a greedy edge-collapsing algorithm on the dual Delaunay triangulation to minimize integrated squared error. Our employed cost function takes into account how point removal affects neighbouring cells and overall topology, while remaining efficient in time and memory up to millions of input sites. Figure~\ref{fig:teaser} visualizes the process of Voronoi reduction.

Our \emph{contributions} can be summarized as follows:
\begin{itemize}
\item We formalize the Voronoi reduction problem to decimate volume representations using the integrated squared error metric.
\item We formulate an efficient edge-collapsing algorithm on the dual Delaunay triangulation.
\item We provide a fully featured GPU-accelerated implementation of our algorithm that scales to millions of primitives.
\item We generate synthetic datasets and collect real-world volumetric functions to establish a benchmark that demonstrates that our algorithm preserves overall structure and details more faithfully than baseline approaches. 
\end{itemize}
%-------------------------------------------------------------------------
\section{Related Work}
In this section, we review related work on geometry reduction and convex cells.

%\subsection{3D Voronoi Applications}
%Voronoi diagrams are a central geometric datastructure that finds application in diverse fields. 
%Maderbacher et al.~\cite{Blackbox} generate control functions for continuous blackbox systems using .

\subsection{Geometry Reduction}
Like Voronoi reduction, point cloud downsampling aims to reduce the number of
points while preserving the underlying geometry. Classic techniques rely on geometric heuristics: Farthest point sampling~\cite{FPS}, Poisson disk sampling~\cite{PoissonDisk} and Centroidal Voronoi Tessellations~\cite{CVT, Lloyd} aim for well-spaced, approximately uniform density in the output point cloud, while others focus on preserving high-frequency details and sharp edges by considering surface normals and curvature~\cite{PointCloud1, PointCloud2, PointCloud3, PointSimplification}.  
More recently, learning-based approaches train sampling
networks jointly with a downstream task such as classification or
segmentation~\cite{SNet, SampleNet, APES}.
None of these methods transfers
directly to 3D Voronoi models. They assume that the points sample a surface,
whereas Voronoi sites represent volumetric cells, and they ignore the
neighborhood structure defined by the Voronoi diagram. Gaussian Mixture reduction is a related problem~\cite{GaussianMixture1, GaussianMixture2, GaussianMixture3}, which also features discrete optimization of volume error, but lacks the topology defined by Voronoi diagrams.

Volumetric mesh simplification typically operates on tetrahedral
meshes~\cite{tetmesh1, tetmesh2, tetmesh3, tetmesh4} rather than on the dual
Voronoi structure. Even when the input is a Delaunay tetrahedralization,
collapse operations generally do not preserve the Delaunay property, so the
result no longer corresponds to a valid Voronoi diagram.
Adaptive thinning
methods~\cite{AdaptiveThinning1, AdaptiveThinning2} greedily remove points
from a Delaunay triangulation according to the approximation error of the
resulting piecewise linear interpolant, but they were developed primarily
for bivariate scattered data.

%Mesh simplification operates on 2-manifolds instead of volumes, but also considers vertex topology similarly 
Numerous important works for mesh decimation employ a sequence of edge collapses ordered by a
priority queue, with costs ranging from global energy
minimization~\cite{Hoppe1996} to local volume and boundary
preservation~\cite{Lindstrom1998} and quadric
error metric~\cite{QEM}.

%Our approach is inspired by the seminal work of Garland and
%Heckbert~\cite{QEM}, who simplify triangle meshes by greedily contracting
%vertex pairs in order of a quadric error metric. While their method operates
%on 2-manifold surfaces, we directly approximate the impact of 3D Voronoi cell merges on volumetric error.

\subsection{Convex Cell Intersection}
Computing the intersection of two convex cells is the core operation of our
algorithm, and several methods for it exist in the literature.
The Sutherland-Hodgman algorithm~\cite{SutherlandHodgman} finds the exact intersection of two convex cells by iteratively slicing the volume of one cell using the faces of the other. In each iteration, one face of the clipping cell is selected and all edges of the other cell are clipped against it.
The intersection of two convex cells can alternatively be computed by finding a point inside the intersection using Linear Programming, dualizing all of the half-spaces and computing their convex hull, before mapping the faces back to primal space, yielding the vertices of the intersection, which can be used to compute the volume~\cite{DualSpace}. 
Notably, while asymptotically efficient, both of these approaches face great difficulties when parallelizing the operation on the GPU due to sequential dependencies and path divergence. In order to scale Voronoi reduction to millions of cells, a massively parallelizable way to approximate the intersection volume of two convex cells is necessary.
Monte Carlo methods have high parallelization potential and provide an efficient way to compute volume of convex cells in high dimensions \cite{MonteCarlo1, MonteCarlo2}, but offer little performance benefit over exact methods in estimating the volume for 3D convex cells and struggle in particular with thin cell intersections.
\section{Preliminaries}
In this section, we summarize important properties of Voronoi diagrams necessary for the algorithm construction.
\subsection{Voronoi Diagrams and Delaunay Triangulations}
A 3D Voronoi diagram is a partition of space based on a discrete set of seed points. For each seed $\mathbf{x}_i$, the corresponding Voronoi cell $V_i$ consists of all points in $\mathbb{R}^d$ with smaller or equal distance to $\mathbf{x}_i$ than to any other seed in the set. Mathematically, each cell is defined as:
\begin{equation}
    V_i :=  \left\{ \mathbf{p} \in \mathbb{R}^d \mid \|\mathbf{p} - \mathbf{x}_i\| \leq \|\mathbf{p} - \mathbf{x}_j\|\  \forall j \neq i  \right\}.
\end{equation}
These cells are convex polyhedra, and their boundaries represent the equidistant interfaces between neighboring seeds. 
 In the dual graph of the Voronoi diagram, known as the Delaunay triangulation $\mathcal{T}= \{\{\mathbf{x}_i\}_{i=1}^N, \mathcal{T}_\mathcal{E}\}$, two vertices are connected by an edge $e \in \mathcal{T}_\mathcal{E}$ iff their cells share a face in the Voronoi diagram. The Delaunay edge condition states that two cells share a face in the Voronoi diagram iff there exists a supporting sphere that intersects only those two points:
\begin{equation}\mathcal{T}_\mathcal{E} = \left\{ (i, j) |  \exists \mathbf{c} \in \mathbb{R}^d \text{ s.t. } \|\mathbf{x}_i - \mathbf{c}\| = \|\mathbf{x}_j - \mathbf{c}\| < \|\mathbf{x}_k - \mathbf{c}\| \ \forall k \not \in \{i,j\}\right\}. \end{equation}
We define the neighbour function $\mathcal{N}(i)$ as the set of neighbours of a point in the Delaunay triangulation
\begin{equation}\mathcal{N}(i):= \{j|(i, j) \in \mathcal{T}_\mathcal{E}\}.\end{equation}

and denote the perpendicular bisector $\Pi_{\mathbf{a}, \mathbf{b}}$ between points $a, b$ as:
\begin{equation}
    \Pi_{\mathbf{a}, \mathbf{b}} = \left\{ \mathbf{x} \in \mathbb{R}^3 \;\middle|\; \left(\mathbf{x} - \frac{\mathbf{a} + \mathbf{b}}{2}\right) \cdot (\mathbf{b} - \mathbf{a}) = 0 \right\}
\end{equation}

Voronoi cells can alternatively be defined as the intersection of half-spaces $\Pi^+_{\mathbf{x}_i, \mathbf{x}_n}$ for all points $\mathbf{x}_n \in \mathcal{N}(i)$:   
\begin{equation}
    V_i = \bigcap_{n \in \mathcal{N}(i)} \Pi_{\mathbf{x}_i, \mathbf{x}_n}^+,
\end{equation}
where $\Pi^+_{\mathbf{x}_i, \mathbf{x}_n}$ is the half-space bounded by $\Pi_{\mathbf{x}_i, \mathbf{x}_n}$ that contains $\mathbf{x}_i$.
A Voronoi face $\mathcal{F}_{i, j}$ shared by cells $V_i, V_j$ with $j \in \mathcal{N}(i)$ is then a convex polytope defined by:
\begin{align}
    \mathcal{F}_{i, j} &=  V_i \cap V_j 
    \\&= V_i \cap \bigcap_{n \in  \mathcal{N}(j)} \Pi^+_{\mathbf{x}_j, \mathbf{x}_n}
    \\&=\Pi_{\mathbf{x}_i, \mathbf{x}_j} \cap \bigcap_{n \in \mathcal{N}(i) \cap \mathcal{N}(j)} \Pi^+_{\mathbf{x}_j, \mathbf{x}_n}
    \label{Eq:Face}
\end{align}

%\begin{figure}[h]
%    \centering
%    \includegraphics[width=0.7\linewidth]{Figures/voronoi_delaunay.png}
%    \caption{Voronoi cells (blue) for a set of points in 2D space with dual Delaunay Triangulation (gray).}
%    \label{fig:VoronoiDelaunay}
%\end{figure}

Throughout, we will assume the sites are in
general position, i.e.\ no four sites are coplanar and no five are co-spherical.
In particular, two Voronoi cells intersect if and only if their sites share a
Delaunay edge.

\subsection{The Voronoi Reduction Problem}
To our knowledge, the problem of reducing the site counts of Voronoi diagrams representing volumetric functions has not been formalized in previous work. In this section, we provide a definition.

Let $g$ be a function $\mathbb{R}^{d_i} \rightarrow \mathbb{R}^{d_o}$ and $\mathcal{S}$ be a set of $N$ samples at locations $\mathbf{x}_i$:
$\mathcal{S} = \left\{\big(\mathbf{x}_i, \mathbf{y}_i\big)\right\}_{i=1}^N$ with $\mathbf{y}_i = g(\mathbf{x}_i)$. In the following, we will assume input dimension $d_i = 3$. 

The original approximated function $f(\mathbf{x})$ assigns the value of the closest sample to any point $\mathbf{x}$ (commonly referred to as ``nearest-neighbour interpolation''). Using the indicator function $\mathbb{I}_{V_i}(\mathbf{x})$ (which is 1 if $\mathbf{x} \in V_i$ and 0 otherwise), we can write this piecewise constant function as:
\begin{equation}
f(\mathbf{x}) = \sum_{i=1}^N \mathbf{y}_i \cdot \mathbb{I}_{V_i}(\mathbf{x}).
\end{equation}
The goal is to construct a smaller set of $M<N$ samples. Let this reduced set be denoted as $\hat{\mathcal{S}} = \{(\hat{\mathbf{x}}_k, \hat{\mathbf{y}}_k)\}_{k=1}^M.$ 
The reduced approximated function, $\hat{f}(\mathbf{x})$, is defined analogously:
\begin{equation}\hat{f}(\mathbf{x}) = \sum_{k=1}^M \hat{\mathbf{y}}_k \cdot \mathbb{I}_{\hat{V}_k}(\mathbf{x}),\end{equation}
with $V$ denoting the Voronoi cells of the original point set $\mathcal{S}$ and $\hat{V}$ those of the reduced point set $\hat{\mathcal{S}}$.
The reduced point set $\hat{\mathcal{S}}$ should minimize the \textit{integrated squared error} $ISE(\hat{\mathcal{S}}, \mathcal{S}) $ between the original and reduced approximations:
\begin{align}
  \hat{\mathcal{S}}^* &= \underset{%\hat{\mathcal{S}}
  }{\mathrm{argmin}} \ ISE(\hat{\mathcal{S}}, \mathcal{S}) \\  
   &= \underset{%\hat{\mathcal{S}
   }{\mathrm{argmin}
  } \int_{\Omega} \left( f(\mathbf{x}) - \hat{f}(\mathbf{x})\right)^2 \, d\mathbf{x}\\
   &= \underset{}{\mathrm{argmin}
   } \int_{\Omega} \left( \sum_{i=1}^N \mathbf{y}_i \mathbb{I}_{V_i}(\mathbf{x}) - \sum_{k=1}^M \hat{\mathbf{y}}_k \mathbb{I}_{\hat{V}_k}(\mathbf{x}) \right)^2 \, d\mathbf{x},
\end{align}
subject to $|\hat{\mathcal{S}}| = M < N$ on the domain $\Omega$ (e.g. bounding box of samples). Figure \ref{fig:2D_reduction} illustrates the problem using the example of downsampling a 2D image. 

Because both $f$ and $\hat{f}$ are piecewise constant, this integral can be broken down into a sum of integrals over the intersections of the original and new Voronoi cells ($V_i \cap \hat{V}_k$). Inside each intersection, the error is simply the constant $(\mathbf{y}_i - \hat{\mathbf{y}}_k)^2$. Thus, the error can also be written as:
\begin{align}
ISE(\hat{\mathcal{S}}, \mathcal{S}) &= \sum_{i=1}^N \sum_{k=1}^M \int_{V_i \cap \hat{V}_k} (\mathbf{y}_i - \hat{\mathbf{y}}_k)^2 \, d\mathbf{x} \\
&= \sum_{i=1}^N \sum_{k=1}^M (\mathbf{y}_i - \hat{\mathbf{y}}_k)^2 \cdot \text{Vol}(V_i \cap \hat{V}_k),    
\end{align}

where $\text{Vol}(\cdot)$ denotes the Lebesgue measure (i.e. area in 2D, volume in 3D) of the overlapping region.

\begin{figure}
    \centering
    \begin{subfigure}[b]{0.24\linewidth}
        \centering
        \includegraphics[width=0.95\linewidth]{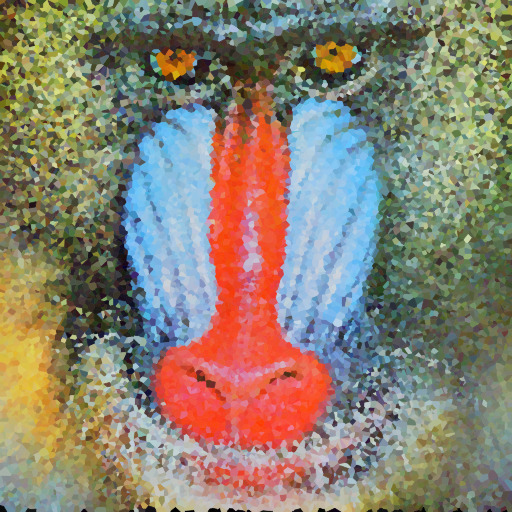}
        \caption{15k cells}
    \end{subfigure}%
    \begin{subfigure}[b]{0.24\linewidth}
        \centering
        \includegraphics[width=0.95\linewidth]{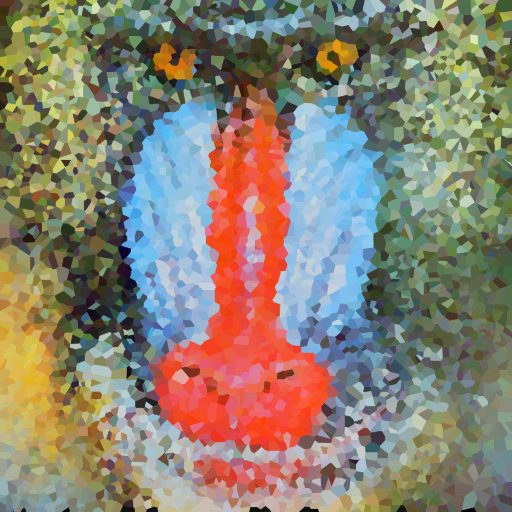}
        \caption{5k cells}
    \end{subfigure}
     \begin{subfigure}[b]{0.24\linewidth}
        \centering
        \includegraphics[width=0.95\linewidth]{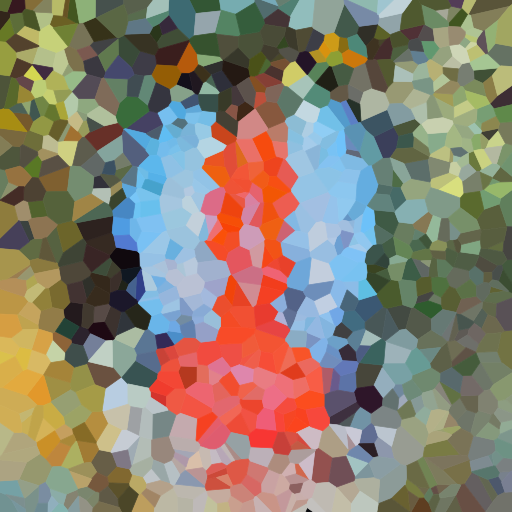}
        \caption{1k cells}
    \end{subfigure}
    \begin{subfigure}[b]{0.24\linewidth}
        \centering
        \includegraphics[width=0.95\linewidth]{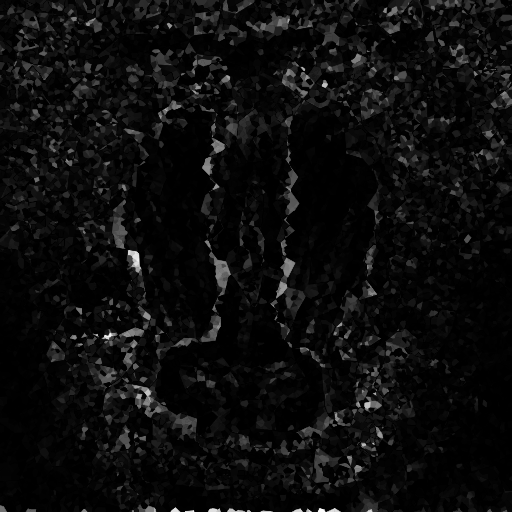}
        \caption{MSE (1k)}
    \end{subfigure}
    \caption{Example of Voronoi reduction in 2D. An image is represented using a set of Voronoi cells. The goal is to reduce the number of sites while minimizing mean squared error. }
    \label{fig:2D_reduction}
\end{figure}

%\subsection{NP-hardness}
%I think we can show NP-hardness by reduction to Planar K-Means
%\begin{theorem}
%The Voronoi Reduction Problem is NP-Hard. 
%\end{theorem}
%\begin{proof}
%    
%\end{proof}

%\subsection{Radiant Foam}
%
%RadFoam~\cite{radfoam} leverages 3D Voronoi diagrams to represent radiance fields by associating color $\mathbf{c}_i$ and density $\sigma_i$ parameters with each cell. 
%Cells are then rendered by raytracing the cells and blending them together according to the Volume Rendering equation to find the pixel color $\mathbf{c}_{p}$:
%$$\mathbf{c}_{p} = \sum_{n=1}^{N} \prod_{j=1}^n \Big(\textrm{exp}(-\sigma_j \delta_j)\Big) \cdot (1 - \textrm{exp}(-\sigma_n \delta_n)) \cdot \mathbf{c}_n \cdot d_n,$$
%where $d_n$ represents the distance over which the ray intersects the $n$-th cell. 
%Because the Voronoi tessellation is differentiable with respect to the seed positions, the system can optimize the geometry (the "foam") and the appearance simultaneously, allowing the cells to move and adapt to the underlying structure of the 3D scene. 

\section{Method}
%Our method employs a greedy cell-merging algorithm to decimate the representation. %The central difficulty is that collapsing a single Delaunay edge will also affect the geometry of all neighbouring cells in complex ways and may cause an arbitrary amount of edge flips in the Delaunay triangulation. Without modelling these interactions, the model quickly degenerates. 
Removing a Voronoi cell has complex effects on neighbouring cells and may cause an arbitrary number of edges in the Delaunay triangulation to flip (an illustrative example is included in Supplemental~\ref*{app:example}). As such, an effective Voronoi reduction algorithm must account both for the geometric changes of the cells and the topological changes to the dual triangulation.

\subsection{Iterative Cell Merging}
A natural procedure to reduce the number of Voronoi cells is to iteratively merge neighbouring cell pairs by combining two sites which share a Voronoi boundary. Neighbouring Voronoi cells are connected by a Delaunay edge, so this algorithm can alternatively be viewed as iterative edge collapse on the dual Delaunay triangulation. Figure~\ref{fig:CellMerge} shows a 2D example of a cell merge.
To determine the order of collapse, we select cell pairs greedily based on the integrated squared error that would be introduced by the merge. Fully re-building the Delaunay triangulation after every edge collapse and re-computing all cell-pair costs is too expensive in practice, so we attempt to collapse $\lfloor \beta \cdot |\mathcal{S}| \rfloor$ edges in each batch before re-computing edges and costs. 
Costs only need to be re-computed for neighbours of merged cells and for new Delaunay edges introduced through edge flips. 
%We choose $\beta = 0.02$ as a default, striking a good trade-off between performance and accuracy. 
Further discussion of this parameter is included in Supplemental~\ref*{app:B_ablation}.

%The number of re-computations per batch is, just like the number of edge flips and cell neighbours, constant in expectation. 
The optimal value for the merged cell $\mathbf{y}_{new}^*$ is the mean of $f$ over the merged cell
as derived in Supplemental~$\ref*{app:merge}$. Note that the cost computation in each step measures the error to the state before the merge, not to the original input. 
The final procedure is sketched in Algorithm \ref{alg:iterative_collapse}.
\begin{algorithm}
\caption{Iterative Delaunay Edge Collapse}
\label{alg:iterative_collapse}
\begin{algorithmic}[1]
\Require Set of initial samples $\mathcal{S}$, target sample count M, collapse batch size $\beta$
\State $\mathcal{T}_\mathcal{E} \gets \text{\textsc{ComputeDelaunayEdges}}(\mathcal{S})$
\For{\textbf{each} $e \in \mathcal{T}_\mathcal{E}$}
    \State $c(e) \gets \text{\textsc{ComputeCost}}(e)$
\EndFor
\While{$|\mathcal{S}| > M$}
    \State $\mathcal{E}_{sorted} \gets \text{\textsc{SortEdgesByCostAscending}}(\mathcal{T}_\mathcal{E}, c)$
    \State $B \gets \lfloor \beta \cdot |\mathcal{S}| \rfloor$
    \State $\mathcal{E}_{collapse} \gets \mathcal{E}_{sorted}[0:B]$ \Comment{Cheapest B edges}
    
    \State $\mathcal{S}_{removed} \gets \emptyset$
    \State $\mathcal{S}_{new} \gets \emptyset$
    
    \For{\textbf{each} $(u, v) \in \mathcal{E}_{collapse}$}
        \If{$u \notin \mathcal{S}_{removed}$ \textbf{and} $v \notin \mathcal{S}_{removed}$}
            \State $\mathbf{x}_{new} \gets \frac{\mathbf{x}_u + \mathbf{x}_v}{2}$ 
            \State $\mathbf{y}_{new} \gets \text{\textsc{InterpolateValue}}(\mathbf{y}_u, \mathbf{y}_v)$
            
            \State $\mathcal{S} \gets (\mathcal{S} \setminus \{(\mathbf{x}_u, \mathbf{y}_u), (\mathbf{x}_v, \mathbf{y}_v)\}) \cup \{(\mathbf{x}_{new}, \mathbf{y}_{new})\}$
            
            \State $\mathcal{S}_{removed} \gets \mathcal{S}_{removed} \cup \{u, v\}$
            \State $\mathcal{S}_{new} \gets \mathcal{S}_{new} \cup \{(\mathbf{x}_{new}, \mathbf{y}_{new})\}$
        \EndIf
    \EndFor
    
    \State $\mathcal{T}_\mathcal{E} \gets \text{\textsc{RecomputeDelaunayEdges}}(\mathcal{S})$
    \State $\mathcal{E}_{affected} \gets \text{\textsc{IdentifyAffectedEdges}}(\mathcal{T}_\mathcal{E},  \mathcal{\mathcal{S}}_{removed}, \mathcal{\mathcal{S}}_{new})$
    \Comment{Re-compute cost for neighbours of merged cells}
    \For{\textbf{each} $e \in \mathcal{E}_{affected}$} 
        \State $c(e) \gets \text{\textsc{ComputeCost}}(e)$
    \EndFor
    
    %\State $\mathcal{T}_\mathcal{E} \gets \mathcal{E}_{new}$
\EndWhile
\State \Return $\mathcal{S}$
\end{algorithmic}
\end{algorithm}

%The algorithm iteratively computes sets $\mathcal{S}_1, \mathcal{S}_2, ..., \mathcal{S}_n$ with $|\mathcal{S}_i| < |\mathcal{S}_{i+1}|$, by choosing the edge collapse which minimizes $ISE(\mathcal{S}_i,\mathcal{S}_{i+1})$

\begin{figure}[h]
    \centering
    \includegraphics[width=0.7\linewidth]{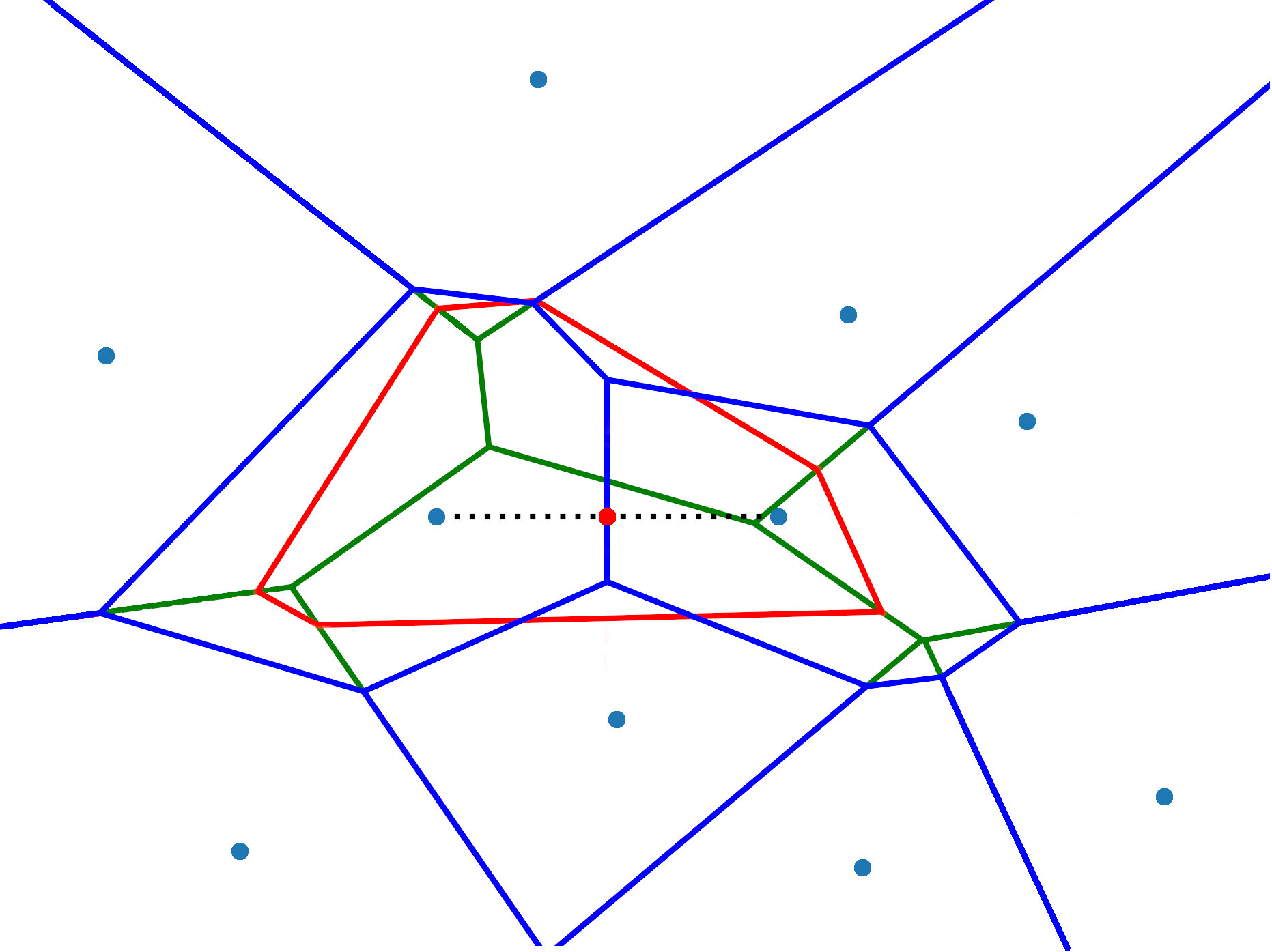}
    \caption{Collapsing a Delaunay edge (dotted) of a Voronoi diagram (blue) creates a merged cell (red). Green lines represent the Voronoi diagram without the point pair. }
    \label{fig:CellMerge}
\end{figure}

\subsection{Computing the cost}
Recall that after collapsing an edge $(\mathbf{x}_i, \mathbf{x}_j) \in \mathcal{T}_{\mathcal{E}}$ to its midpoint $\mathbf{x}_{new}=\frac{\mathbf{x}_i+\mathbf{x}_j}{2}$, the introduced $ISE$ to the point set $\mathcal{S}$ before the collapse is given by:

\begin{equation}
c(i, j) = ISE(\hat{\mathcal{S}}, \mathcal{S}) = \sum_{k=1}^N \sum_{l=1}^M (\mathbf{y}_k - \hat{\mathbf{y}}_l)^2 \cdot \text{Vol}(V_k \cap \hat{V}_l), 
\label{CostBreakdown}
\end{equation}
where $\hat{\mathcal{S}} = \left(\mathcal{S} \setminus \{(\mathbf{x}_i, \mathbf{y}_i), (\mathbf{x}_j, \mathbf{y}_j)\} \right) \cup (\mathbf{x}_{new}, \mathbf{y}_{new})$ is the reduced set after collapsing $(\mathbf{x}_i, \mathbf{x}_j)$. Assume for the notation that adding and removing points does not shift indices.

Due to the properties of the Voronoi diagram, collapsing a Delaunay edge only increases the error locally. 
Specifically, reducing the sample set from $\mathcal{S}$ to $\hat{\mathcal{S}}$
 only causes changes to the function in the union of the volume of the merged and original cells:
\begin{equation}
      \forall \mathbf{x} \in \mathbb{R}^{d_i} \setminus (V_i \cup V_j \cup \hat{V}_{new}):f(\mathbf{x}) = \hat{f}(\mathbf{x})   
\end{equation}

We can split the error terms from Equation~\ref{CostBreakdown} into gained and lost volume of neighbour cells, changes to the merged cells and effects on non-neighbouring cells. 
%Therefore, the vast majority of the terms in Equation \ref{CostBreakdown} vanish and the remaining error is given by:
\begin{align}
  c(i, j) &= \sum_{k \in \mathcal{N}(i) \setminus j}  (\mathbf{y}_i - \mathbf{y}_k)^2 \cdot \text{Vol}(V_i \cap \hat{V}_k)  \label{eq1} \\
  &+ \sum_{k \in \mathcal{N}(j) \setminus i}  (\mathbf{y}_j - \mathbf{y}_k)^2 \cdot \text{Vol}(V_j \cap \hat{V}_k) \label{eq2}\\
  &+ \sum_{k \in \mathcal{N}(i) \cup \mathcal{N}(j) \setminus \{i,j\}} (\hat{\mathbf{y}}_{new} - \mathbf{y}_k)^2 \cdot \text{Vol}(\hat{V}_{new} \cap V_k) \label{eq3} \\
  &+ (\hat{\mathbf{y}}_{new} - \mathbf{y}_i)^2 \cdot \text{Vol}(\hat{V}_{new} \cap V_i)\label{eq4}\\ 
  &+ (\hat{\mathbf{y}}_{new} - \mathbf{y}_j)^2 \cdot \text{Vol}(\hat{V}_{new} \cap V_j), \label{eq5}\\
  &+ \sum_{k \in \mathcal{S} \setminus (\mathcal{N}(i) \cup \mathcal{N}(j))} (\hat{\mathbf{y}}_{new} - \mathbf{y}_k)^2 \cdot \text{Vol}(\hat{V}_{new} \cap V_k) \label{eq6} 
\end{align}
with $\hat{V}_k$ representing the cell $k$ after the merge has taken place.
Equations \ref{eq1} and \ref{eq2} represent the volume gained by neighbouring cells from $V_i$ and $V_j$ respectively. Term \ref{eq3} represents the volume that neighbouring cells lose to the merged cell $\hat{V}_{new}$ and \ref{eq4} / \ref{eq5}  is the volume that the merged cell $\hat{V}_{new}$ ``inherits'' from its parent cells $V_i$ and $V_j$ respectively. Cases can occur in which the merged cell $\hat{V}_{new}$ intersects with a cell that was not a neighbour of $i$ or $j$, which is accounted for in Equation~\ref{eq6}. However, these cases are both infrequent and have negligible effect on the error, as explored in Supplemental~\ref*{app:Flip_Frequency}. Because handling this rare and insignificant case would be expensive, we will assume that the term is zero for practical cost computation.  

Figure \ref{fig:voronoi_volume_costs} visualizes the first five components of the edge collapse cost and Figure \ref{fig:3D_volumes} shows the individual intersection volumes for a 3D cell pair.  
%\todo{Shri: I think it would be better to explain this more clearly with refernece to the figure. Make something like initial sites were $x_0$, $x_1$ with voronoi cells $V_0$, $V_1$ as shown in figure. Then the reduction to the red site(mid point) in figure leads to $V_2$ and we can measure the difference in volume using the purple regions whose cost can be computed the following way, then the equations}
\begin{figure}[h]
    \centering
    \includegraphics[width=0.75\linewidth]{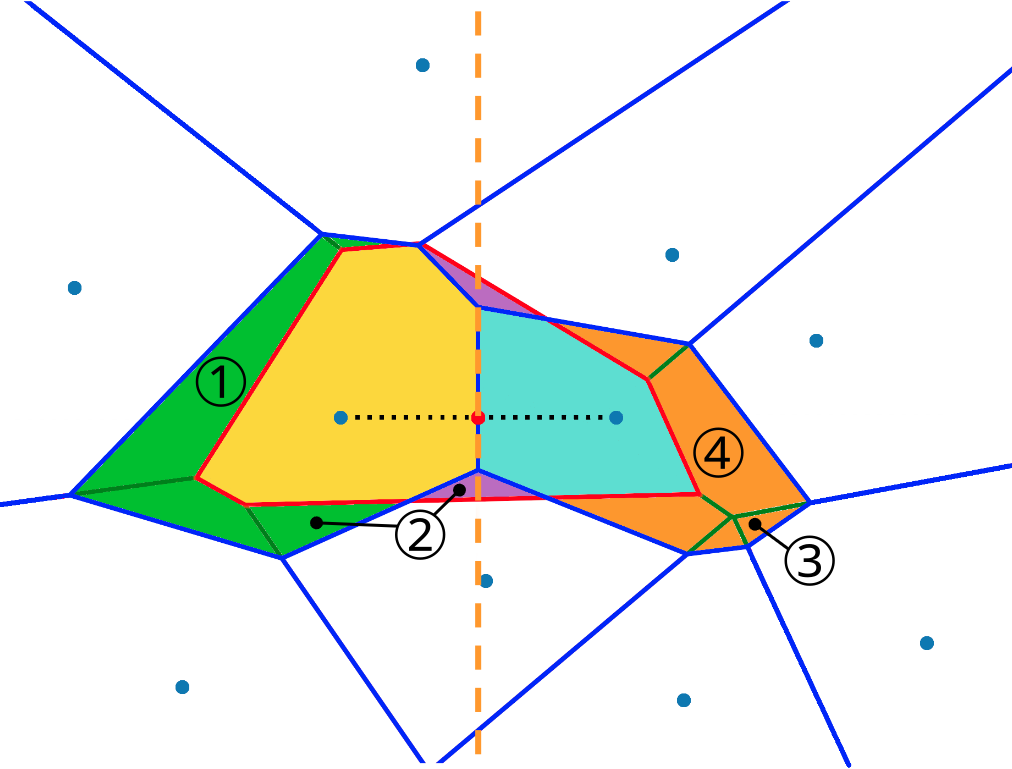}
    \caption{Two cells (blue) are merged along a Delaunay edge (dotted) to a new cell (red). In colored volumes the nearest neighbour has changed. The green and orange regions correspond to Equations \ref{eq1}/\ref{eq2} and the cyan and yellow regions to \ref{eq4}/\ref{eq5}. The volumes from Equation \ref{eq3} are marked in purple. Circled numbers indicate special cases. }
    \label{fig:voronoi_volume_costs}
\end{figure}

\begin{figure}[h]
    \centering
    \includegraphics[width=0.6\linewidth]{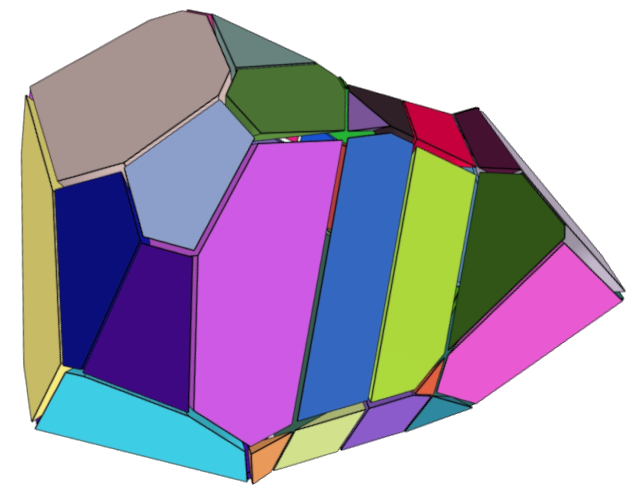}
    \caption{For a pair of 3D Voronoi cells, the volume gained and lost through the edge collapse is visualized for each neighbour cell in a different color.}
    \label{fig:3D_volumes}
\end{figure}

\subsubsection{Computing Intersection Volumes}
\label{sec:approximation}
To compute the intersection of convex cells, we use a multi-threaded implementation of Sutherland-Hodgman \cite{SutherlandHodgman}. Computing these intersections represents the main computational burden of the algorithm, representing up to $92\%$ of total runtime. With around 2000 amortized volume intersections for every single performed edge collapse, the optimization of this routine is paramount to the performance of the reduction, necessitating the use of approximations.

%For random point distributions, the expected vertex degree of a 3D Delaunay triangulation is approximately $16$ \cite{AverageDegree}. It follows that computing the gained and lost volume for every neighbour of both cells requires approximately $64$ convex cell intersections. After collapsing the edge, this cost needs to be re-computed for the expected 16 new edges of the merged point $\mathbf{x}_{new}$, resulting in a total of $64\cdot16=1024$ volume computations for removing a single point. Decimating Voronoi diagrams with millions of points requires billions of cell intersections in total, making the acceleration of this operation absolutely vital for the scalability of the algorithm.  

\providecommand{\figidea}[1]{\par\noindent\textcolor{blue}{\textbf{[Figure idea]} #1}\par}
\providecommand{\TODO}[1]{\textcolor{red}{\textbf{TODO}: #1}}
\providecommand{\Vol}{\operatorname{Vol}}
\providecommand{\conv}{\operatorname{conv}}
 
\paragraph*{Approximating Neighbour Intersection Volumes}
\label{sec:neighbour_volumes}

%Throughout this section, the sites $\mathbf{x}_i$ and $\mathbf{x}_j$ are merged
%into $\mathbf{x}_{new} = \frac{1}{2}(\mathbf{x}_i + \mathbf{x}_j)$, and quantities
%after the merge are marked with a hat: $\hat V_k$, $\hat{\mathcal{N}}(k)$,
%$\hat{\mathcal{F}}_{k,l}$ and $\hat{\mathcal{T}}$. 
 Let points $\mathbf{x}_i$ and $\mathbf{x}_j$ be merged
into $\mathbf{x}_{new} = \frac{1}{2}(\mathbf{x}_i + \mathbf{x}_j)$.
For $k \in \mathcal{N}(i)$ (we choose $i$ w.l.o.g. over $j$), the cost terms in Equations~\ref{eq1} and
\ref{eq3} require two volumes:
\begin{align}
  G_k &:= V_i \cap \hat V_k      && \text{(volume $V_k$ gains from $V_i$)}, \label{eq:gain_region}\\
  L_k &:= V_k \cap \hat V_{new}  && \text{(volume $V_k$ loses to $\hat V_{new}$)}. \label{eq:loss_region}
\end{align}
In this section we will assume that $k \not \in \mathcal{N}(j)$ and detail handling for the converse case at the end. 
These can be thought of as the volume swept between the faces $\mathcal{F}_{orig} = \mathcal{F}_{\mathbf{x}_{i}, \mathbf{x}_{k}}$ and $\mathcal{F}_{new} = \mathcal{F}_{\mathbf{x}_{new}, \mathbf{x}_{k}}$ constrained by the halfspaces that define $\mathcal{F}_{\mathbf{x}_{i}, \mathbf{x}_{k}}$, which we call the \emph{neighbour tunnel} of $k$ towards $i$:  
\begin{align}
  \mathcal{H}_{k,i} &:= \{\Pi^+_{\mathbf{x}_k,\mathbf{x}_l} \mid l\in\mathcal{N}(k)\cap\mathcal{N}(i)\}.
 % \mathcal{C}_{k,i} &:= \textstyle\bigcap \mathcal{H}_{k,i}
\end{align}

%Every half-space in
%$\mathcal{H}_{k,i}$ remains a constraint of $\hat V_k$ after the merge,
%since $j\notin\mathcal{N}(k)$ and all other common neighbours remain sites.
%Hence
%\begin{equation}
%  \hat V_k\subseteq\textstyle\bigcap \mathcal{H}_{k,i}  .
%  \label{eq:tunnel_contains}
%\end{equation}
%For a face $\mathcal{F}$ of $V_k$ or $\hat V_k$, we denote by
%\begin{equation}
%  \mathcal{E}(\mathcal{F}) := \{\,l \mid \mathcal{F}\cap\Pi_{\mathbf{x}_k,\mathbf{x}_l}
%  \text{ is an edge of }\mathcal{F}\}
%\end{equation}
%the sites whose wall supports an edge of $\mathcal{F}$. We have
%$\mathcal{E}(\mathcal{F}_{k,i})\subseteq\mathcal{N}(k)\cap\mathcal{N}(i)$.
%The walls of the remaining common neighbours are redundant on
%$\mathcal{F}_{k,i}$ (cf.\ Equation~\ref{eq:face}), but they may become active
%between the old and the new face. 
The original face can then be alternatively defined as 
\begin{equation}
\mathcal{F}_{orig} := \Pi_{\mathbf{x}_k,\mathbf{x}_i}\cap \textstyle\bigcap \mathcal{H}_{k,i}  
 \label{eq:old_face}
\end{equation}
 (cf. Equation~\ref{Eq:Face}) and so long as the underlying Delaunay triangulation $\mathcal{T}$ is unaffected by the merge, 
 \begin{equation}
 \hat{\mathcal{F}}_{new} := \Pi_{\mathbf{x}_k,\mathbf{x}_{new}}\cap \textstyle\bigcap \mathcal{H}_{k,i}.
 \label{eq:new_face}
 \end{equation} 
Figure~\ref{fig:Prismatoid_Tunnel} illustrates these relationships on an example intersection volume as pictured in Figure~\ref{fig:3D_volumes}, corresponding to the 2D case \ding{192} in Figure~\ref{fig:voronoi_volume_costs}.

%The original and new faces are
%\begin{equation}
%  \mathcal{F}_{orig} := \Pi_{\mathbf{x}_k,\mathbf{x}_i}\cap \textstyle\bigcap \mathcal{H}_{k,i},
%  \qquad
%  \mathcal{F}_{new}  := \Pi_{\mathbf{x}_k,\mathbf{x}_{new}}\cap \textstyle\bigcap \mathcal{H}_{k,i}.
%  \label{eq:tunnel_faces}
%\end{equation}
%A convex polygon is the intersection of the half-planes supporting its edges.
%Hence $\mathcal{F}_{orig}=\mathcal{F}_{k,i}$ exactly. Whether
%$\mathcal{F}_{new}$ equals the true face $$ depends on
%the changes to the triangulation, as discussed below.
 
%\figidea{3D rendering of a cell $V_k$ with the face $\mathcal{F}_{orig}$
%highlighted and the tunnel walls $\Pi_{\mathbf{x}_k,\mathbf{x}_l}$,
%$l\in\mathcal{N}(k)\cap\mathcal{N}(i)$, extended as translucent planes towards $\mathbf{x}_i$.
%The tilted plane $\Pi_{\mathbf{x}_k,\mathbf{x}_{new}}$ cuts the tunnel in
%$\mathcal{F}_{new}$, and the region between the two faces is shaded.}
%\begin{figure}
%    \centering
%    \includegraphics[width=0.65\linewidth]{Figures/Primatoid_New.png}
%    \caption{The change in volume of a neighbouring cell is given as the truncated prismatoid formed by the convex hull of the intersection of $\Pi_{\mathbf{x}_k,\mathbf{x}_i}$ and $\Pi_{\mathbf{x}_k,\mathbf{x}_{new}}$ with the neighbour tunnel.
%    }
%    \label{fig:Prismatoid_Tunnel}
%\end{figure}

The face $\mathcal{F}_{new}$ may intersect $\mathcal{F}_{orig}$ in the line $\Pi_{\mathbf{x}_k,\mathbf{x}_i} \cap \Pi_{\mathbf{x}_k,\mathbf{x}_{new}}$, in which case it separates the faces into $\mathcal{F}^+_{new}/\mathcal{F}^-_{new}$ and $\mathcal{F}^+_{orig}/\mathcal{F}^-_{orig}$. The swept volume between $\mathcal{F}^+_{new}, \mathcal{F}^+_{orig}$ then represents volume gained by $V_k$ and $\mathcal{F}^-_{new}, \mathcal{F}^-_{orig}$ volume lost, each forming non-intersecting pairs of convex polygons. This case is illustrated in Figure~\ref{fig:double_wedge} and indicated with \ding{193} in the 2D illustration of Figure~\ref{fig:voronoi_volume_costs}.

%\begin{figure}
%    \centering
%    \includegraphics[width=0.75\linewidth]{Figures/Double_Wedge_new.png}
%    \caption{If the faces before and after the merge intersect, some volume $G_k$ (green) is gained and some $L_k$ is lost (purple) from cell $V_k$. }
%    \label{fig:double_wedge}
%\end{figure}

%Define the affine function
%\begin{equation}
%  \sigma_i(\mathbf{p}) := \|\mathbf{p}-\mathbf{x}_{new}\|^2 - \|\mathbf{p}-\mathbf{x}_i\|^2 .
%\end{equation}
%Its zero set is the bisector $\Pi_{\mathbf{x}_i,\mathbf{x}_{new}}$. By
%Equation~\ref{eq:midpoint_identity}, this plane is parallel to
%$\Pi_{\mathbf{x}_i,\mathbf{x}_j}$, halfway between $\mathbf{x}_i$ and
%$\Pi_{\mathbf{x}_i,\mathbf{x}_j}$. For a polygon $\mathcal{F}$, let
%$\mathcal{F}^{\pm} := \mathcal{F}\cap\{\pm\sigma_i\ge 0\}$.
% 
%\begin{lemma}[Separation of gain and loss]
%\label{lem:separation}
%$G_k \subseteq \{\sigma_i \ge 0\}$ and $L_k \subseteq \{\sigma_i\le 0\}$.
%\end{lemma}
%\begin{proof}
%For $\mathbf{p}\in G_k$ we have
%$\|\mathbf{p}-\mathbf{x}_i\|\le\|\mathbf{p}-\mathbf{x}_k\|\le\|\mathbf{p}-\mathbf{x}_{new}\|$.
%For $\mathbf{p}\in L_k$ we have
%$\|\mathbf{p}-\mathbf{x}_{new}\|\le\|\mathbf{p}-\mathbf{x}_k\|\le\|\mathbf{p}-\mathbf{x}_i\|$.
%\end{proof}
 
%On $\Pi_{\mathbf{x}_k,\mathbf{x}_i}$ and on $\Pi_{\mathbf{x}_k,\mathbf{x}_{new}}$,
%the two faces can only meet on the line
%$\Pi_{\mathbf{x}_k,\mathbf{x}_i}\cap\Pi_{\mathbf{x}_i,\mathbf{x}_{new}}$.
%Hence the plane $\Pi_{\mathbf{x}_i,\mathbf{x}_{new}}$ splits both faces
%consistently into a gained part and a lost part.

We will refer to the convex hull of two convex polygons whose interiors do not intersect as a \emph{truncated prismatoid}.
 Unlike a classical prismatoid, the polygons need not be parallel. The definition
includes wedges, whose two bases share an edge, and pyramids, where one base is
a point. The volume of a truncated prismatoid can be efficiently computed using tetrahedral
decomposition (cf. Supplemental~\ref{PismaVolume}). 
%\begin{definition}[Truncated prismatoid]
%A \emph{truncated prismatoid} is a convex polyhedron $P$ defined by the convex hull of two convex polygons whose interior does not intersect. %whose vertices all lie
%%in the union of two planes $E_1, E_2$, with $P = \conv\big((P\cap E_1)\cup(P\cap E_2)\big)$.
%%The two \emph{bases} $P\cap E_1$ and $P\cap E_2$ are convex polygons and may
%%degenerate to a segment or a point.
%\end{definition}

\paragraph*{The flip-free case}
 
We first consider merges that change the Delaunay tetrahedralisation only
through the collapse of the edge $(i,j)$. Formally, let
$\hat{\mathcal{T}} = \mathcal{T}/(i,j)$ be the contraction of $\mathcal{T}$
along $(i,j)$: tetrahedra containing both $i$ and $j$ are removed, and $i$ and
$j$ are identified with $new$ in all remaining tetrahedra.
 \begin{theorem}[Exactness without flips]
\label{thm:noflip}
Let $k\in\mathcal{N}(i)\setminus\mathcal{N}(j)$ and $\hat{\mathcal{T}} = \mathcal{T}/(i,j)$.
Assume $\mathcal{S}\cup\{\mathbf{x}_{new}\}$ is in general position, and $V_i$ and
$\hat V_{new}$ are bounded. Then
\begin{enumerate}
  \item $\Pi_{\mathbf{x}_k,\mathbf{x}_{new}}\cap \textstyle\bigcap \mathcal{H}_{k,i} = \hat{\mathcal{F}}_{k,new}$;
  \item if $\mathcal{F}_{new} \cap \mathcal{F}_{orig} = \emptyset$, then
        $G_k = \conv\!\big(\mathcal{F}_{orig}\cup\mathcal{F}_{new}\big)$ and
        $L_k = \emptyset$ (or vice versa);
  \item otherwise, $G_k = \conv\!\big(\mathcal{F}^+_{orig}\cup\mathcal{F}^+_{new}\big)$ and
        $L_k = \conv\!\big(\mathcal{F}^-_{orig}\cup\mathcal{F}^-_{new}\big)$.
        Both are truncated prismatoids with positive volume.
\end{enumerate}
\end{theorem}

The full proof is included in Supplemental \ref{app:proofs}.
In the flip-free case, the vertices of $\mathcal{F}_{orig}$ and
$\mathcal{F}_{new}$ have a one-to-one correspondence along the Voronoi edges
$\Pi_{\mathbf{x}_k,\mathbf{x}_{l}}\cap\Pi_{\mathbf{x}_k,\mathbf{x}_{l'}}$.
%All lateral faces are therefore planar quadrilaterals, and pairing vertices by
%index in Algorithm~\ref{alg:prism} is exact. In the double wedge, the two
%crossing points are shared by both bases.
 
%The sign test replaces the polygon intersection test
%$\mathcal{F}_{M}\cap\mathcal{F}_{orig}=\emptyset$ of the original cost
%algorithm. We evaluate the affine function $\sigma_i$ at the vertices of both
%faces and clip both faces with the single plane
%$\Pi_{\mathbf{x}_i,\mathbf{x}_{new}}$.
 \begin{figure}[h]
    \centering
    \begin{subfigure}[t]{0.48\linewidth}
        \centering
        \includegraphics[width=\linewidth]{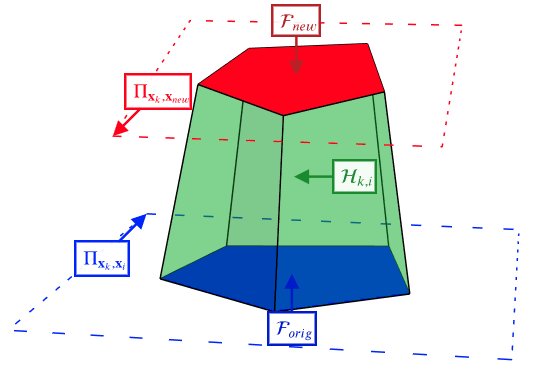}
        \caption{The change in volume of a neighbouring cell is given as the truncated prismatoid formed by the convex hull of the intersection of $\Pi_{\mathbf{x}_k,\mathbf{x}_i}$ and $\Pi_{\mathbf{x}_k,\mathbf{x}_{new}}$ with the neighbour tunnel.}
        \label{fig:Prismatoid_Tunnel}
    \end{subfigure}
    \hfill
    \begin{subfigure}[t]{0.48\linewidth}
        \centering
        \includegraphics[width=\linewidth]{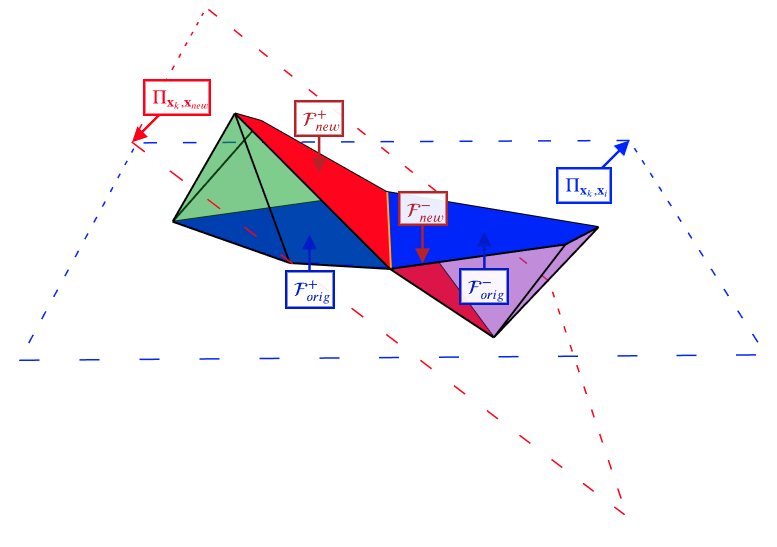}
        \caption{If the faces before and after the merge intersect, some volume $G_k$ (green) is gained and some $L_k$ (purple) is lost from cell $V_k$.}
        \label{fig:double_wedge}
    \end{subfigure}
    \caption{Volume changes of neighbouring cells during a merge.}
    \label{fig:volume_change}
\end{figure}

\subsubsection*{Merges with flips}
 
In general, the merge may also change Delaunay topology around $\mathbf{x}_k$ through bistellar flips. The polytope $G_k$ is then bounded by
$\Pi_{\mathbf{x}_k,\mathbf{x}_i}$, possibly by $\Pi_{\mathbf{x}_k,\mathbf{x}_{new}}$
and additional half-spaces not contained in $\mathcal{H}_{k,i}$. We distinguish
three situations, which may occur simultaneously. Our algorithm treats all of
them uniformly: it always evaluates Equations~\ref{eq:old_face},\ref{eq:new_face} and falls
back to a separate approximation only in Case~3. To characterize these cases in isolation, we will assume $L_k=\emptyset$ w.l.o.g. (i.e. the original and new face do not cross). 
 
\paragraph*{Case 1: Vertex loss.}
In this case every wall bounding $G_k$ belongs to $\mathcal{H}_{k,i}$, but one of the walls $\Pi_{\mathbf{x}_k,\mathbf{x}_l}$ with $l \in \mathcal{N}(k) \cap \mathcal{N}(i)$ will be ``cut off'' by two adjacent walls before intersecting $\Pi_{\mathbf{x}_k,\mathbf{x}_{new}}$. The
lateral face on $\Pi_{\mathbf{x}_k,\mathbf{x}_l}$ then ``pinches'' out at a vertex
$\mathbf{v}$ strictly between the two planes, and $\mathcal{F}_{new}$ loses one vertex. 
The face $\mathcal{F}_{new}$ can still be exactly computed by Equation~\ref{eq:new_face}, but the
additional vertex $\mathbf{v}$ turns $G_k$ into a general convex polytope.
 
In our implementation, we approximate $G_k$ by the truncated prismatoid
$\conv(\mathcal{F}_{orig}\cup\mathcal{F}_{new})$, underestimating the true volume. The missing volume
lies in small caps around the pinch vertices. Figure \ref{fig:pinch_vertex} visualizes this case.
%Importantly, vertex loss does not
%change which cells are neighbours:
 
%\begin{prop}[Vertex loss preserves neighbourhoods]
%\label{prop:vertex_loss}
%Assume every wall bounding $G_k$ belongs to $\mathcal{H}_{k,i}$, and let $l$ be
%the site of a pinched lateral face with adjacent walls $l_1, l_2$. Then:
%\begin{enumerate}
%  \item every $l'$ whose wall supports an edge of
%        $\mathcal{F}^+_{orig}$ remains a neighbour of $k$; in particular,
%        $l\in\hat{\mathcal{N}}(k)$;
%  \item every $l'\in\mathcal{E}(\hat{\mathcal{F}}_{k,new})$ is a neighbour of $new$; in
%        particular, $l_1,l_2\in\hat{\mathcal{N}}(new)$;
%  \item the triangle $\{k,new,l\}$ is not in $\hat{\mathcal{T}}$, while the
%        tetrahedron $\{k,l_1,l,l_2\}$ is.
%\end{enumerate}
%\end{prop}
%
% 
%In the simplest instance, this is a $2\!\to\!3$ bistellar flip. The two
%tetrahedra $\{k,new,l_1,l\}$ and $\{k,new,l,l_2\}$ of $\mathcal{T}/(i,j)$ are
%replaced by $\{k,new,l_1,l_2\}$, $\{new,l,l_1,l_2\}$ and $\{k,l_1,l,l_2\}$. This
%creates the edge $(l_1,l_2)$ but removes no edge incident to $k$ or $new$.
 In dual space, $\mathcal{F}_{new}$ corresponds to an edge and each of its vertices is dual to a tetrahedron. Thus, $|\mathcal{F}_{new}| < |\mathcal{F}_{orig}|$ indicates that one of the tetrahedra has flipped, for instance by creating a new edge between two neighbours, while $\mathcal{N}(k)$ remains unaffected.
Such a flip has no counterpart in 2D, where Voronoi faces are segments with
exactly two vertices, and by planarity every flip adds or removes an edge
incident to $k$.
 
%Proposition~\ref{prop:vertex_loss} concerns the gained side. On the lost side,
%a vertex loss of $\mathcal{F}^-_{new}$ means that $\hat V_{new}$ swallows a
%Voronoi vertex of $V_k$. This can remove a neighbour of $k$ entirely, since the
%region is given exactly by $L_k = V_k\cap\Pi^+_{\mathbf{x}_{new},\mathbf{x}_k}$.

%\figidea{Pinching in 3D: the tunnel between $\mathcal{F}_{orig}$ (blue) and
%$\mathcal{F}_{new}$ (red). One lateral face is a triangle ending at the pinch
%vertex $\mathbf{v}$. Include an inset with the corresponding $2\!\to\!3$ flip,
%showing the tetrahedra before and after and the new edge $(l_1,l_2)$.}
 
\paragraph*{Case 2: New walls.}
As $V_k$ expands into the volume $G_k$, it might gain a new neighbour cell $V_l$ with $l\in\mathcal{N}(i)$. Geometrically, this means that $\Pi_{\mathbf{x}_k, \mathbf{x}_l}$ intersects with  $\bigcap\mathcal{H}_{k, i}$ between $\mathcal{F}_{orig}$ and $\mathcal{F}_{new}$.
In dual space, the new connection $(k,l) \in \mathcal{T}_\mathcal{E}$ is a result of a bi-stellar flip. 
Because $l\in\mathcal{N}(i)$ must hold, it would be possible to identify these cases through a search over $\mathcal{N}(i)$.
However, computing its
volume means clipping a general convex polytope, with a data-dependent number
of active planes and vertices. This causes divergent control flow and memory
access on the GPU.
 
As an engineering choice, we therefore ignore the walls of new neighbours
$l\in\mathcal{N}(i)\setminus\mathcal{N}(k)$. We compute $\mathcal{F}_{new}$
from the tunnel alone, as if the neighbourhood of $k$ were unchanged. The approximated $\mathcal{F}_{new}$ is then not clipped by $\Pi_{\mathbf{x}_k,\mathbf{x}_l}$, thus
overestimating the gained volume. Figure \ref{fig:new_neighbour} pictures an example of this case in 3D and marker \ding{194} in Figure \ref{fig:voronoi_volume_costs} marks the 2D equivalent case. 
Note that cases~1 and~2 can occur for the same
cell, with opposite biases.
%Note that pruning the tunnel to $\mathcal{E}(\mathcal{F}_{k,i})$, i.e.\ to
%the walls that are active on the old face, is not lossless. Walls of common
%neighbours that are redundant on $\mathcal{F}_{orig}$ can become active between
%the faces, and pruning would move them into Case~2.

%\figidea{Case 2: the same tunnel as in the Case 1 figure, with an additional
%wall $\Pi_{\mathbf{x}_k,\mathbf{x}_l}$, $l\in\mathcal{N}(i)\setminus\mathcal{N}(k)$,
%cutting a corner off the region between the faces. Shade the overestimated
%volume.}

\begin{figure*}[htbp]
    \centering
    \begin{subfigure}[t]{0.32\textwidth}
             \centering
     \includegraphics[width=0.45\linewidth]{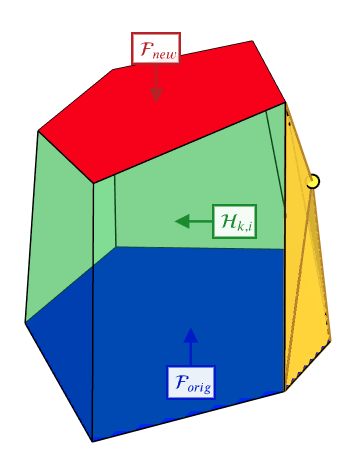}
     \caption{\textbf{Case 1.} When $|\mathcal{F}_{orig}|>|\mathcal{F}_{new}|$, a vertex occurs instead between the two faces, creating a pinch. Approximating the volume as a truncated prismatoid underestimates the volume by a tetrahedron spanned by the pinch vertex $\mathbf{v}$ (yellow).}
     \label{fig:pinch_vertex}
    \end{subfigure}
    \hfill
    \begin{subfigure}[t]{0.32\textwidth}
        \centering
     \includegraphics[width=0.8\linewidth]{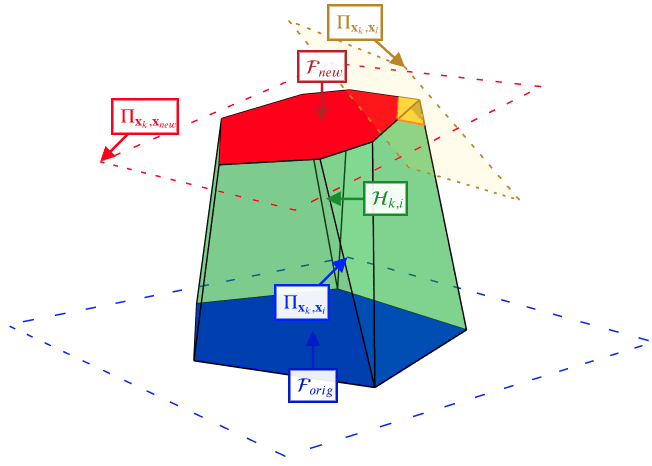}
     \caption{\textbf{Case 2.} When $\mathcal{N}(k)$ contains a new cell $l$ after the merge, $\mathcal{F}_{new}$ is bounded by a new half-space $\Pi_{\mathbf{x}_k, \mathbf{x}_l}$ and $(l,k) \in \mathcal{T}_\mathcal{E}$. Approximating the volume as a truncated prismatoid overestimates by the yellow volume. }
     \label{fig:new_neighbour}
    \end{subfigure}
    \hfill
    \begin{subfigure}[t]{0.32\textwidth}
        \centering
     \includegraphics[width=0.65\linewidth]{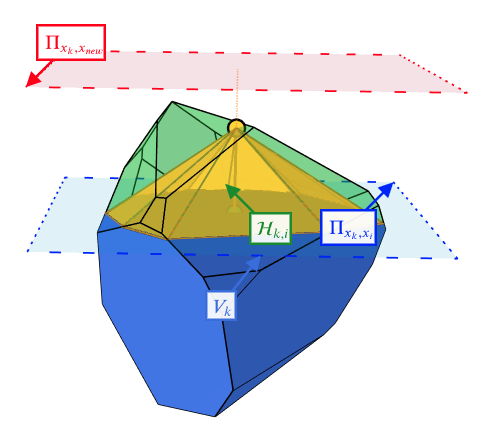}
     \caption{\textbf{Case 3.} When $\mathcal{F}_{new} = \emptyset$, the edge $(i, k) \in \mathcal{T}_\mathcal{E}$ has flipped. We approximate the intersection volume (green) as a pyramid (yellow) with base $\mathcal{F}_{orig}$. }
     \label{fig:pyramid}
    \end{subfigure}
    \caption{Delaunay flips introduce special cases into our approximation.}
    \label{fig:main_figure}
\end{figure*}
 
\paragraph*{Case 3: Tunnel closure.}
In this case the tunnel closes before reaching
$\Pi_{\mathbf{x}_k,\mathbf{x}_{new}}$, meaning $\mathcal{F}_{new}=\emptyset$, because $k$ and $new$
are no longer Delaunay neighbours, i.e. $k$ does not inherit its edge to $i$.
%\begin{prop}[Tunnel closure]
%\label{prop:closure}
%\begin{enumerate}
%  \item $\hat{\mathcal{F}}^+_{k,new}\subseteq\mathcal{F}^+_{new}$. Hence, if
%        the tunnel closes, $G_k$ shares no facet with $\hat V_{new}$.
%  \item If $(k,new)\notin\hat{\mathcal{T}}_\mathcal{E}$, then $L_k=\emptyset$,
%        and only the gained volume contributes to the cost.
%\end{enumerate}
%\end{prop}
Thus, $G_k$ is given by the general convex polytope $\Pi^+_{\mathbf{x}_i, \mathbf{x}_k} \bigcap \mathcal{H}_{k,i}$ and we do not gain any geometric information from $\mathcal{F}_{new}$.
In this case, we approximate the closed region by a pyramid with base $\mathcal{F}_{orig}$. Let
$\mathbf{c}$ be the vertex centroid of $\mathcal{F}_{orig}$, and let
$\mathbf{n} = (\mathbf{x}_i-\mathbf{x}_k)/\|\mathbf{x}_i-\mathbf{x}_k\|$ be its
normal pointing into $V_i$. The apex is placed at $\mathbf{c}+h\,\mathbf{n}$,
where $h$ is the distance at which the ray $(\mathbf{c}, \mathbf{n})$ first leaves the tunnel:
\begin{equation}
  h = \min_{\substack{l\in\mathcal{N}(k)\cap\mathcal{N}(i)\\ \mathbf{n}\cdot(\mathbf{x}_l-\mathbf{x}_k)>0}}
      \frac{\big(\tfrac12(\mathbf{x}_k+\mathbf{x}_l)-\mathbf{c}\big)\cdot(\mathbf{x}_l-\mathbf{x}_k)}
           {\mathbf{n}\cdot(\mathbf{x}_l-\mathbf{x}_k)} .
\end{equation}
The volume of pyramids can be computed extremely efficiently as
$\Vol = \tfrac13\,\operatorname{Area}(\mathcal{F}_{orig})\,h$ without requiring triangulation of the lateral surface. The pyramid may over- or
underestimate the true intersection volume, depending if $G_k$ is bounded by bisectors with non-neighboring points like in case 2.
An example of the pyramid case in 3D is illustrated in Figure \ref{fig:pyramid} and marker \ding{195} in Figure \ref{fig:voronoi_volume_costs} marks the 2D equivalent case. Table~\ref{tab:cases} summarises the approximations for all cases and Supplemental~\ref*{app:Flip_Frequency} describes empirical frequencies for Delaunay flips during merges.
 
%The test is conservative. By Proposition~\ref{prop:closure}, closure of the
%tunnel implies closure of the true region on the gained side. The converse does
%not hold: if the true region closes only because of a wall outside
%$\mathcal{H}_{k,i}$, the tunnel face is non-empty and the neighbour is handled
%as in Cases~1 and~2.

%\figidea{Case 3: the tunnel converging to a tip before reaching
%$\Pi_{\mathbf{x}_k,\mathbf{x}_{new}}$, with the approximating pyramid (red)
%overlaid on the true convex cell (blue). This extends Figure~9a with the ray
%from $\mathbf{c}$ along $\mathbf{n}$ and the wall it hits first.}
 
\subsubsection*{Neighbours of both merged cells}
All previously discussed techniques assume that $k \not \in \mathcal{N}(i) \cap \mathcal{N}(j)$. These cases can however be effectively handled by splitting the cell along $\Pi_{\mathbf{x}_i, \mathbf{x}_j}$, creating two disjoint convex cells $k_1 \in \mathcal{N}(i)$ and $k_2 \in \mathcal{N}(j)$. Then, the volumes $\hat{V}_{k_1} \cap V_i$ and $\hat{V}_{k_2} \cap V_j$ can be straightforwardly approximated using above procedure without special handling.

%If $k\in\mathcal{N}(i)\cap\mathcal{N}(j)$, we split the gained region along
%$\Pi_{\mathbf{x}_i,\mathbf{x}_j}$ into $\hat V_k\cap V_i$ and
%$\hat V_k\cap V_j$. We treat each part with the construction above, applied to
%the pairs $(k,i)$ and $(k,j)$ respectively. The wall
%$\Pi^+_{\mathbf{x}_k,\mathbf{x}_j}$ of the site $j$, which no longer exists
%after the merge, is replaced by $\Pi^+_{\mathbf{x}_i,\mathbf{x}_j}$ in
%$\mathcal{H}_{k,i}$, and analogously for $\mathcal{H}_{k,j}$. By
%Lemma~\ref{lem:absorb}, the two parts are separated by the interior of
%$\hat V_{new}$, so no volume is counted twice. The lost region $L_k$ is split
%along the same plane, and both parts are weighted with
%$(\hat{\mathbf{y}}_{new}-\mathbf{y}_k)^2$.

\begin{table}[t]
\centering
\caption{Handling of neighbour intersection volumes for
$k\in\mathcal{N}(i)\setminus\mathcal{N}(j)$.}
\label{tab:cases}
\small
\begin{tabular}{@{}llll@{}}
\toprule
Configuration & Computed as & Detected & Error \\
\midrule
No flip & prismatoid(s) & -- & exact \\
Case 1: vertex loss & prismatoid%, pinch vertices dropped 
& yes & underestimate \\
Case 2: new walls & prismatoid%, new walls ignored 
& no & overestimate \\
Case 3: tunnel closure & pyramid over $\mathcal{F}_{orig}$ & yes & either \\
\bottomrule
\end{tabular}
\end{table}

\begin{algorithm}[t]
\caption{Compute Cost}
\label{alg:cost_revised}
\begin{algorithmic}[1]
\Function{ComputeCost}{$(i,j)\in\mathcal{T}_\mathcal{E}$}
  \State $\mathbf{x}_{new}\gets\frac12(\mathbf{x}_i+\mathbf{x}_j)$
  \State $cost\gets 0$
  \For{$m\in\{i,j\}$}
    \State $o\gets$ $\{i,j\}\setminus\{m\}$
    \State $\hat V\gets\Vol(V_m)$ \Comment{Volume inherited by $\hat V_{new}$}
    \For{$k\in\mathcal{N}(m)\setminus\{i,j\}$}
      \State $\mathcal{H}\gets$ \Call{Tunnel}{$k,m$} \Comment{Walls of $\mathcal{N}(k)\cap\mathcal{N}(m)$}
      \If{$k\in\mathcal{N}(o)$} \Comment{Split $V_k$ along $\Pi_{\mathbf{x}_i,\mathbf{x}_j}$}
        \State $\mathcal{H}\gets\big(\mathcal{H}\setminus\{\Pi^+_{\mathbf{x}_k,\mathbf{x}_o}\}\big)\cup\{\Pi^+_{\mathbf{x}_m,\mathbf{x}_o}\}$
      \EndIf
      \State $\mathcal{F}_{orig}\gets\Pi_{\mathbf{x}_k,\mathbf{x}_m}\cap\mathcal{H}$
      \State $\mathcal{F}_{new}\gets\Pi_{\mathbf{x}_k,\mathbf{x}_{new}}\cap\mathcal{H}$
      \State $\mathcal{F}^\pm_{orig},\mathcal{F}^\pm_{new}\gets$ split by $\Pi_{\mathbf{x}_m,\mathbf{x}_{new}}$
      \If{$\mathcal{F}^+_{new} \cup\mathcal{F}^-_{new}=\emptyset$} \Comment{Case 3}
        \State $V^+\gets$ \Call{PyramidVolume}{$\mathcal{F}^+_{orig},\mathcal{H}$}
      \Else
        \State $V^+\gets$ \Call{PrismVolume}{$\mathcal{F}^+_{orig},\mathcal{F}^+_{new}$}
      \EndIf
      \State $V^-\gets$ \Call{PrismVolume}{$\mathcal{F}^-_{orig},\mathcal{F}^-_{new}$}
       \State $\hat{\mathbf{y}}_{new}\gets$ \Call{InterpolateValue}{$\mathbf{y}_i,\mathbf{y}_j$}
      \State $cost\gets cost+V^+(\mathbf{y}_m-\mathbf{y}_k)^2+V^-(\hat{\mathbf{y}}_{new}-\mathbf{y}_k)^2$
      \State $\hat V\gets\hat V-V^+$ \Comment{Proposition~\ref{prop:NoNeedForNewCellVolume}}
    \EndFor
    \State $cost\gets cost+\hat V\,(\hat{\mathbf{y}}_{new}-\mathbf{y}_m)^2$
  \EndFor
  \State \Return $cost$
\EndFunction
\end{algorithmic}
\end{algorithm}

%\paragraph*{Parallelizability}
%By reducing the intersection volumes from arbitrary convex cells to the volume between two polygons of degree N (if the new face $\mathcal{F}_{new}$ has fewer vertices, it is padded with duplicates, which form no-volume tetrahedra), the volumes of all .  

\subsection{Merged Cell Intersections}
The approximation as the volume between two polygons works well for the intersections with neighbours, which rarely deviate from that shape, but fails to accurately capture the volume of the merged cell $\hat{V}_{new}$. Instead of performing the expensive exact computation, we can reframe the problem to completely avoid direct volume computation.    

\begin{prop}[Rest Volumes]
Let $\mathcal{S}$ be a set of samples, $(i,j) \in \mathcal{T}_\mathcal{E} (\mathcal{S})$ and $\hat{\mathcal{S}} = \left(\mathcal{S} \setminus \{(\mathbf{x}_i, \mathbf{y}_i), (\mathbf{x}_j, \mathbf{y}_j)\} \right) \cup (\mathbf{x}_{new}, \mathbf{y}_{new})$. Then,

\begin{equation} Vol(V_i) - \sum_{k \in \mathcal{N}(i) \setminus j} \text{Vol}(V_i \cap \hat{V}_k)  = \text{Vol}(\hat{V}_{new} \cap V_i)\end{equation} and analogously for $j$.    
\label{prop:NoNeedForNewCellVolume}
\end{prop}
Compare Figure~\ref{fig:voronoi_volume_costs} for visual intuition on this fact.
Computing $\sum_{k \in \mathcal{N}(i)} \text{Vol}(V_i \cap \hat{V}_k)$ is already necessary to find $c(i,j)$ and an initial pre-pass can be used to find the volumes $V(i)$ for all cells. Details on efficient pre-computation for the cell volumes is included in Supplemental~\ref*{App:Volume}.
Then, using Proposition \ref{prop:NoNeedForNewCellVolume}, we can compute $\text{Vol}(\hat{V}_{new} \cap V_i)$ as $ Vol(V_i) - \sum_{k \in \mathcal{N}(i)} \text{Vol}(V_i \cap \hat{V}_k)$ (and analogously for $V_j$) to bypass computing the more complex intersection of the original and merged cells. The full cost computation is sketched in Algorithm~\ref{alg:cost_revised}.

\paragraph*{Parallelization}
Algorithm~\ref{alg:cost_revised} is ideally suited for GPU acceleration: Each neighbour cell $k \in \mathcal{N}(i) \cup \mathcal{N}(j)$ can be handled independently and identically: First, find the half-spaces defining $\mathcal{F}_{orig}$ in angular order. Then, compute the vertices of $\mathcal{F}_{orig}$ and $\mathcal{F}_{new}$ through triple-plane intersections.
The key to the efficiency is that $max(|\mathcal{F}_{orig}|)$ is small in practice, with $|\mathcal{F}_{orig}| \leq 11$ in all our tests, meaning we can pad all tunnels $\mathcal{H}$ in a batch to the maximal degree in order to remove divergence with only minimal extra computation. 
Finally, we sort all the resulting face pairs by degree and compute their volumes in batches. This ensures that the tetrahedralization and volume calculation is parallelized over face pairs of equal degree, eliminating all warp divergence.   
%This powerful pipeline leads to 10\times speed-ups in cost computation over exact cell intersection.
%\subsection{Computational Complexity}
%The total computational complexity of reducing a Voronoi diagram with $N$ sites is $
%\mathcal{O}(N)$ in expectation.

\section{Evaluation}
%To evaluate our algorithm, we use models trained on Radfoam for common novel-view synthesis datasets. They represent real-world scenes as a color and density function that is defined by hundreds of thousands or even millions of 3D Voronoi sites, making it an ideal test case for our method. There are also obvious practical benefits to reduce the Voronoi site count for Radfoam models while preserving accurate detail such as faster rendering and reduced storage and memory requirements.
In this section, we compare the performance of our algorithm against baseline methods. 
\begin{figure}[htbp]
  \centering
  % Subfigure 1
  \hfill
  \begin{subfigure}[b]{0.4\linewidth}
    \centering
    \includegraphics[width=\linewidth, angle=270]{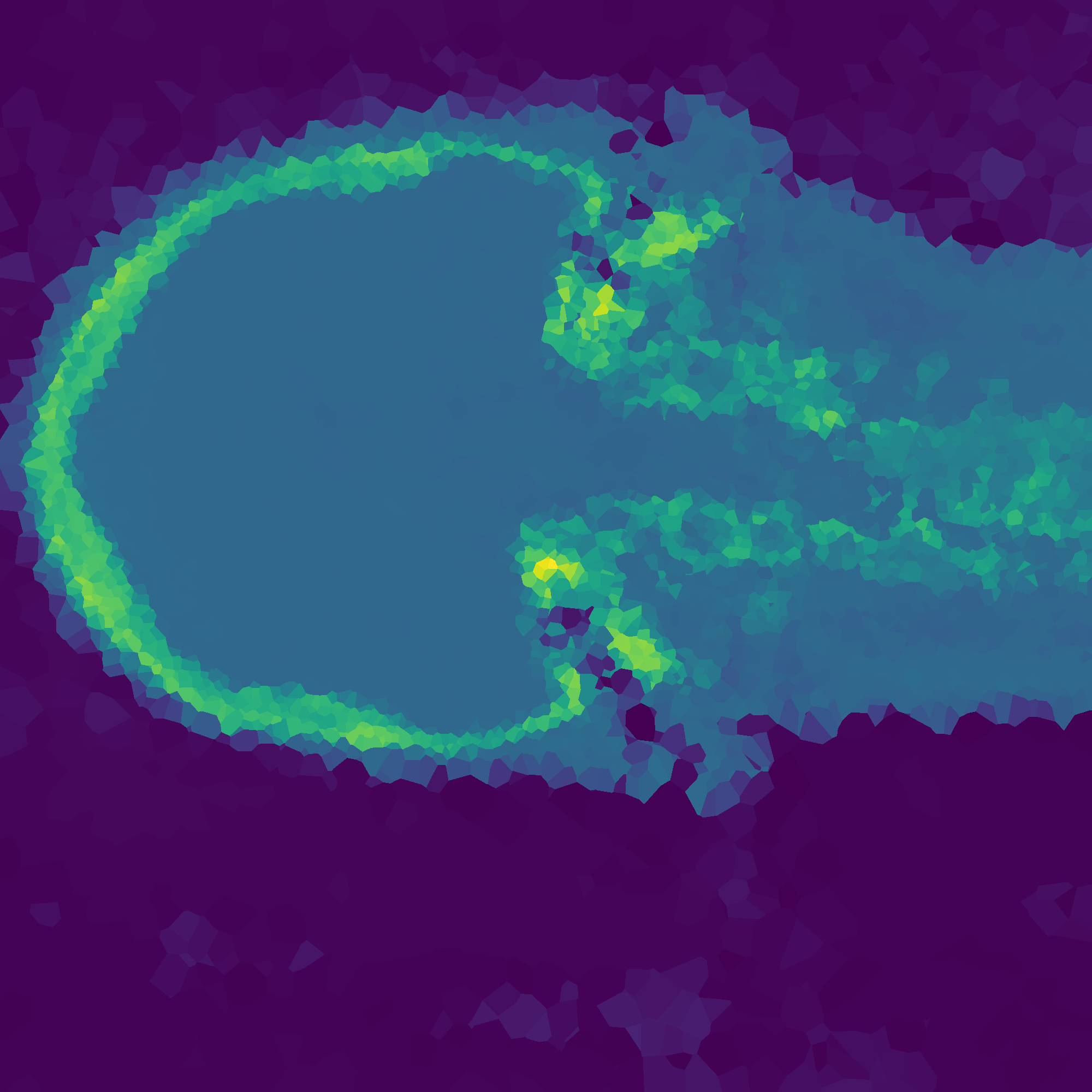}
    \caption{CTHead (uniform)}
    \label{fig:sub1}
  \end{subfigure}
  % Subfigure 2
  \hfill
  \begin{subfigure}[b]{0.4\linewidth}
    \centering
    \includegraphics[width=\linewidth]{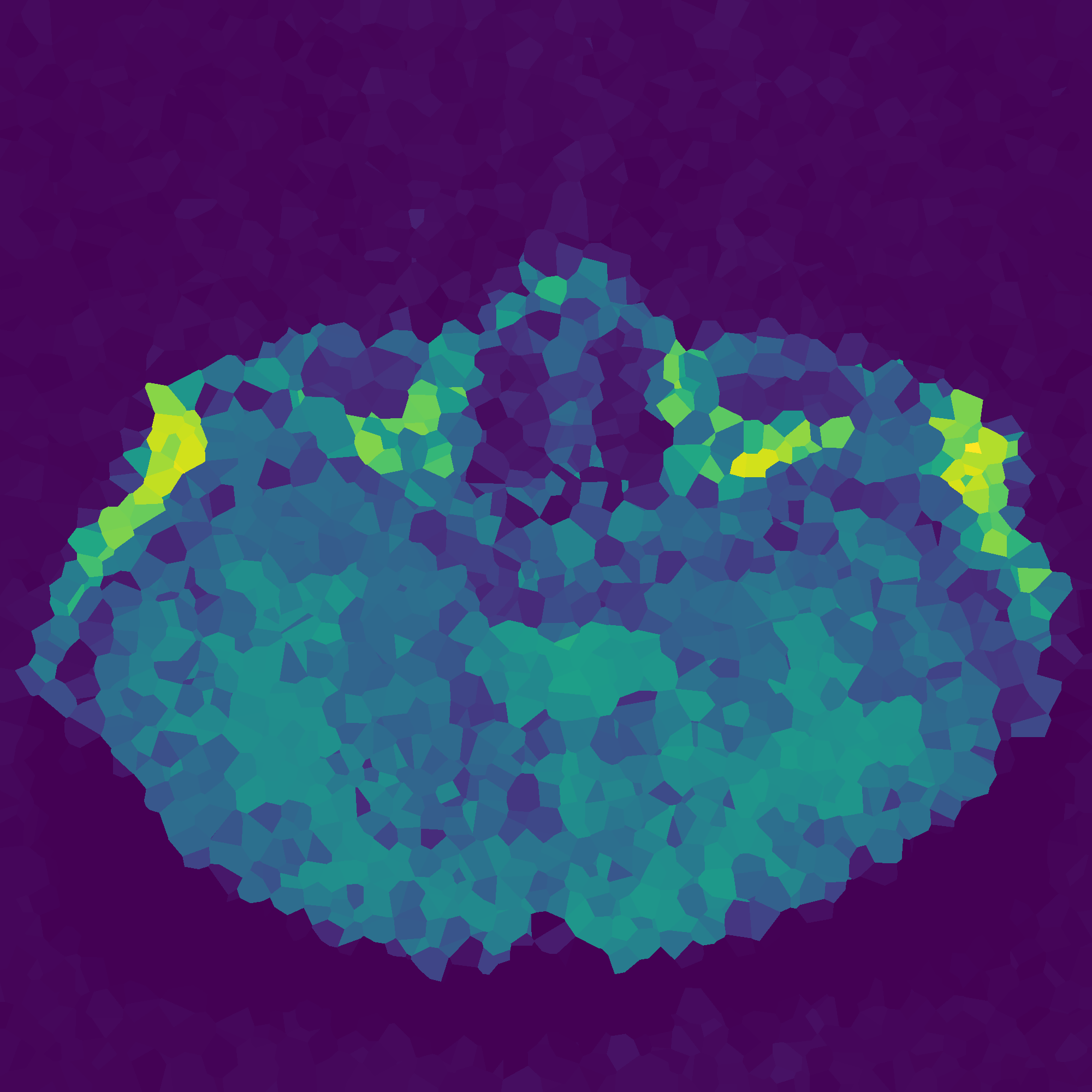}
    \caption{MRBrain (gradient)}
    \label{fig:sub2}
  \end{subfigure}
  \hfill\\
  \hfill
  % Subfigure 3
  \begin{subfigure}[b]{0.4\linewidth}
    \centering
    \includegraphics[width=\linewidth]{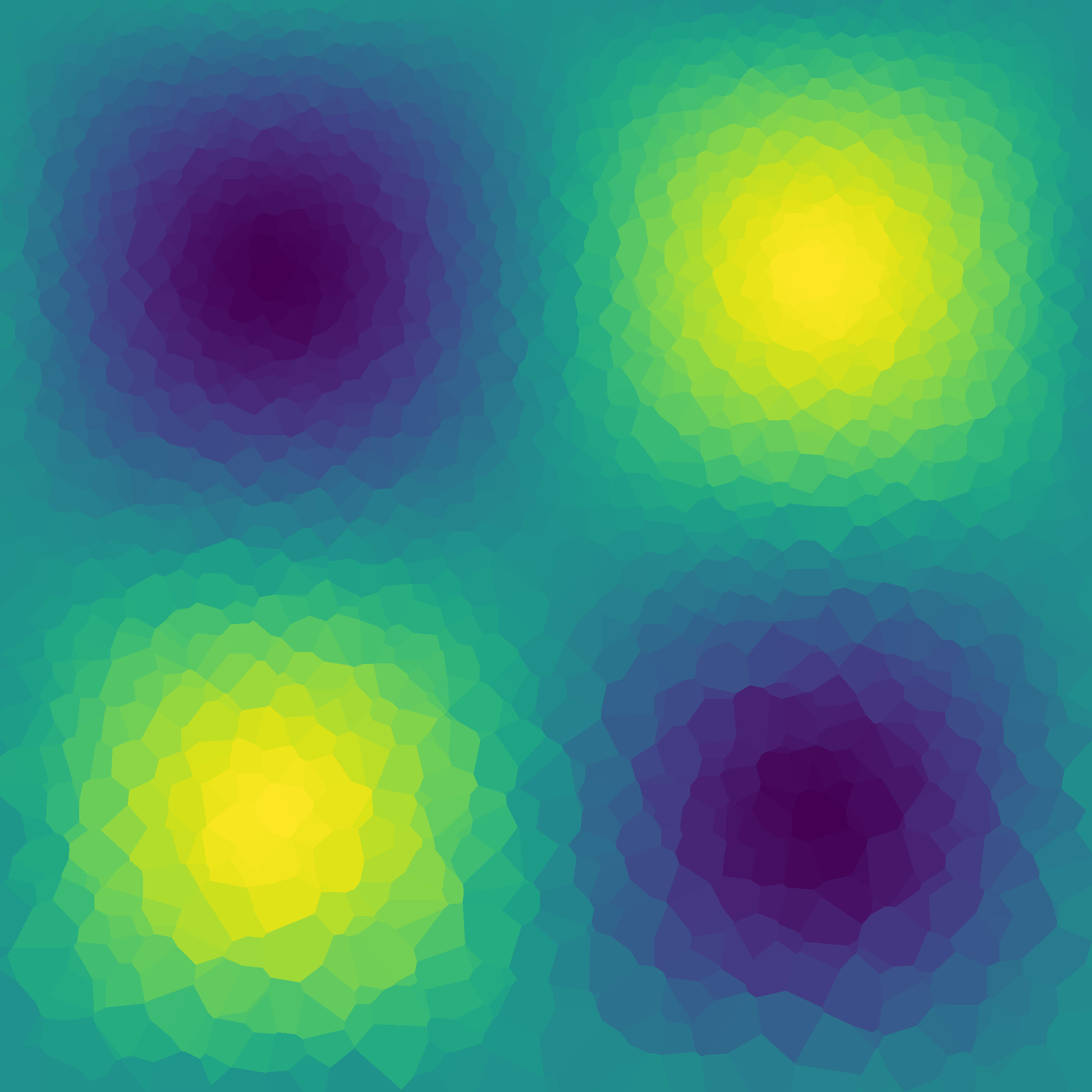}
    \caption{Sine Low (Exponential)}
    \label{fig:sub3}
  \end{subfigure}
  \hfill
  % Subfigure 4
  %\begin{subfigure}[b]{0.15\textwidth}
  %  \centering
  %  \includegraphics[width=\linewidth]{Figures/Dataset_Slices/gaussian_torus.png}
  %  \caption{Torus (Gaussian)}
  %  \label{fig:sub4}
  %\end{subfigure}\hfill
  % Subfigure 5
  \begin{subfigure}[b]{0.4\linewidth}
    \centering
    \includegraphics[width=\linewidth]{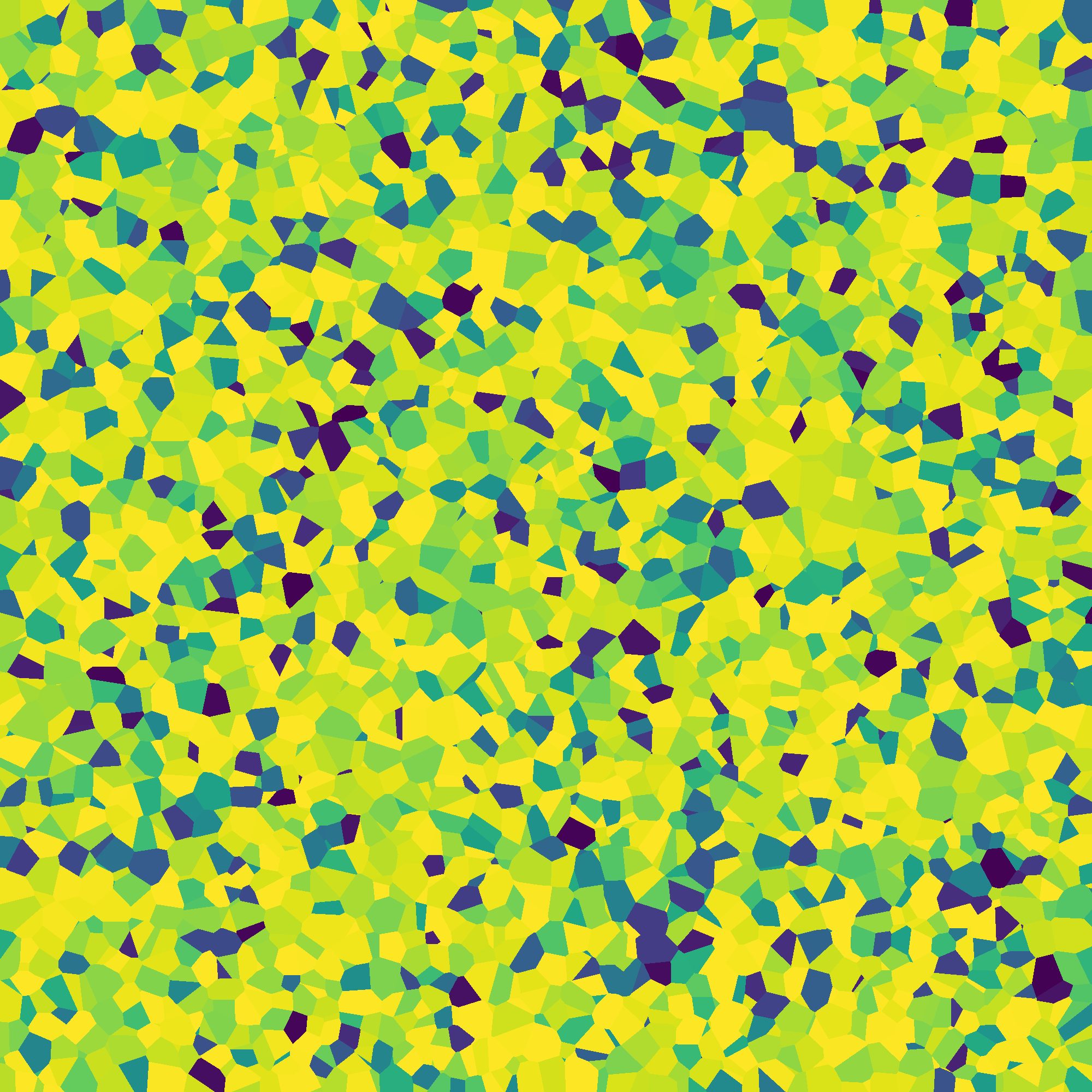}
    \caption{Hexagonal Titanium}
    \label{fig:sub5}
  \end{subfigure}
  \hfill
  \caption{Volume slices for four input scenes from evaluated datasets.}
  \label{fig:data_slices}
\end{figure}
\subsection{Compared methods}
As we are the first to our knowledge to propose the Voronoi reduction problem, we assemble baseline methods from point cloud down-sampling, volumetric mesh simplification, clustering and general heuristics to demonstrate the necessity of a more sophisticated algorithm:
\begin{enumerate}
   
   % \item \textbf{Remove least significant cells:} 
%Following~\cite{MiniSplatting}, we use the. 
%This baselines leverages \emph{image-based heuristics}, and completely neglects the spatial structure of the pre-trained model
    \item \textbf{Volume:} Cells are initially ranked by volume and then removed in order from smallest to largest.
    \item \textbf{Neighbour:} Cells where the function value is similar to their neighbours are good candidates for removal. Cells are ranked by the mean squared error of all neighbour values with the current cell and then greedily removed: 
    \begin{equation}
        c(i) = \sum_{n \in \mathcal{N}(i)} \frac{(\mathbf{y}_n - \mathbf{y}_i)^2}{|\mathcal{N}(i)|} % max(\sigma_n, \sigma_i) \cdot \big(\mathbf{c}_i -  \hat{\mathbf{c}}_n\big)^2 + \lambda_{\sigma} \big( \sigma_i - \hat{\sigma}_n \big)^2
    \end{equation}
    In essence, this baseline accounts for topology of the Voronoi diagram through the local neighbourhood, but neglects to account for the geometry, i.e. the volume, of the cells.
    \item \textbf{Farthest Point Sampling (FPS):} 
    Farthest point sampling (FPS) greedily builds a subset of points that is spread as evenly as possible over the input. It starts from a seed point $x_0$ (random or fixed) and keeps a set Q of selected points. At each step it adds the point whose distance to its nearest selected point is largest:
\begin{equation}
    x_{t+1} = argmax_{x \in S \setminus Q} min_{q \in Q} \| x - q\|.
\end{equation}
The reduction removes points in order from $x_{N-1}$ to $x_0$.
    \item \textbf{Grid:} Ranks points by iteratively building an Octree: At level 0, a single point in the bounding box is selected randomly. At level 1, up to eight points are selected for each non-empty octant and at level 2 up to 64 and so on recursively. Points are then sorted by Octree level descending and removed in that order.     
    \item \textbf{Cignoni et al.:} As a representative method for volumetric mesh reduction, we re-implement the work by Cignoni et al.~\cite{tetmesh2}, specifically the proposed \emph{QN} variant. % We adapt QN in two ways. Our models are defined by their sites rather than by an explicit tetrahedral complex, so a collapse deletes a site and the connectivity is recovered by re-triangulating the survivors. %Furthermore, QN's quadrics measure the error of the piecewise-linear field over the tetrahedra, whereas all methods here are evaluated on the piecewise-constant field over the resulting Voronoi diagram, so the baseline is ranked under a slightly different reconstruction than the one it is scored on.
    QN is formulated for edits to a fixed tetrahedral complex. In our setting, sites are instead deleted and survivors re-triangulated, so the proposed topological guards do not apply, and the quadrics measure a piecewise-linear tetrahedral field while we evaluate a piecewise-constant Voronoi one.    
     \item \textbf{Ward:} This baseline uses the same Delaunay edge collapse algorithm to the midpoint, but simplifies the cost function with the assumption that the merged cell covers the same volume as the original cell pair:
    \begin{equation}
        (\mathbf{y}_{new}-\mathbf{y}_i)^2\cdot \text{Vol}(V_i) + (\mathbf{y}_{new}-\mathbf{y}_j)^2\cdot \text{Vol}(V_j)
    \end{equation}
    with $\mathbf{y}_{new} = \frac{\text{Vol}(V_i) \mathbf{y}_i + \text{Vol}(V_j) \mathbf{y}_j}{\text{Vol}(V_i) + \text{Vol}(V_j)}$. This baseline essentially ignores the impact of merges on neighbouring cells and is equivalent to Ward's agglomerative clustering criterion \cite{Ward} on Delaunay edges.
    \item \textbf{Ours (Fast):} Computes Equations~\ref{eq1}-\ref{eq5} using the prismatoid approximation discussed in Section~\ref{sec:approximation}. 
    \item \textbf{Ours (Exact):} Computes Equations~\ref{eq1}-\ref{eq5} using exact convex cell intersection.
\end{enumerate}
%For \emph{Ward}, \emph{Ours (Fast)} and Ours \emph{(Exact)}, we set $\beta=0.02$ for the evaluation. 
As evaluated in Supplemental~\ref*{app:B_ablation}, higher values of $\beta$ often counter-intuitively result in better metrics. Instead of tuning this parameter, we choose a conservative value of $0.02$ for the evaluation, remaining close to iterative edge collapse, while keeping runtimes reasonable. We also use the same batch size for recomputing triangulation and MSE in the Neighbour method.

Full Code for reproducing our method and baselines is included in the Supplemental material.

\begin{figure*}[h]
    \centering
    \renewcommand{\arraystretch}{0.0}
\setlength{\tabcolsep}{1pt}
    \begin{tabular}{ccccccccc}
  \textbf{Input} & \textbf{Neighbour} & \textbf{Volume} & \textbf{FPS} & \textbf{Grid} & \textbf{Cignoni et al.} & \textbf{Ward} & \textbf{Ours (Fast)} & \textbf{Ours (Exact)} \\[1pt]
%\rotlabel{\scriptsize Gyroid smooth \\ (uniform)} &
\includegraphics[width=0.1\linewidth]{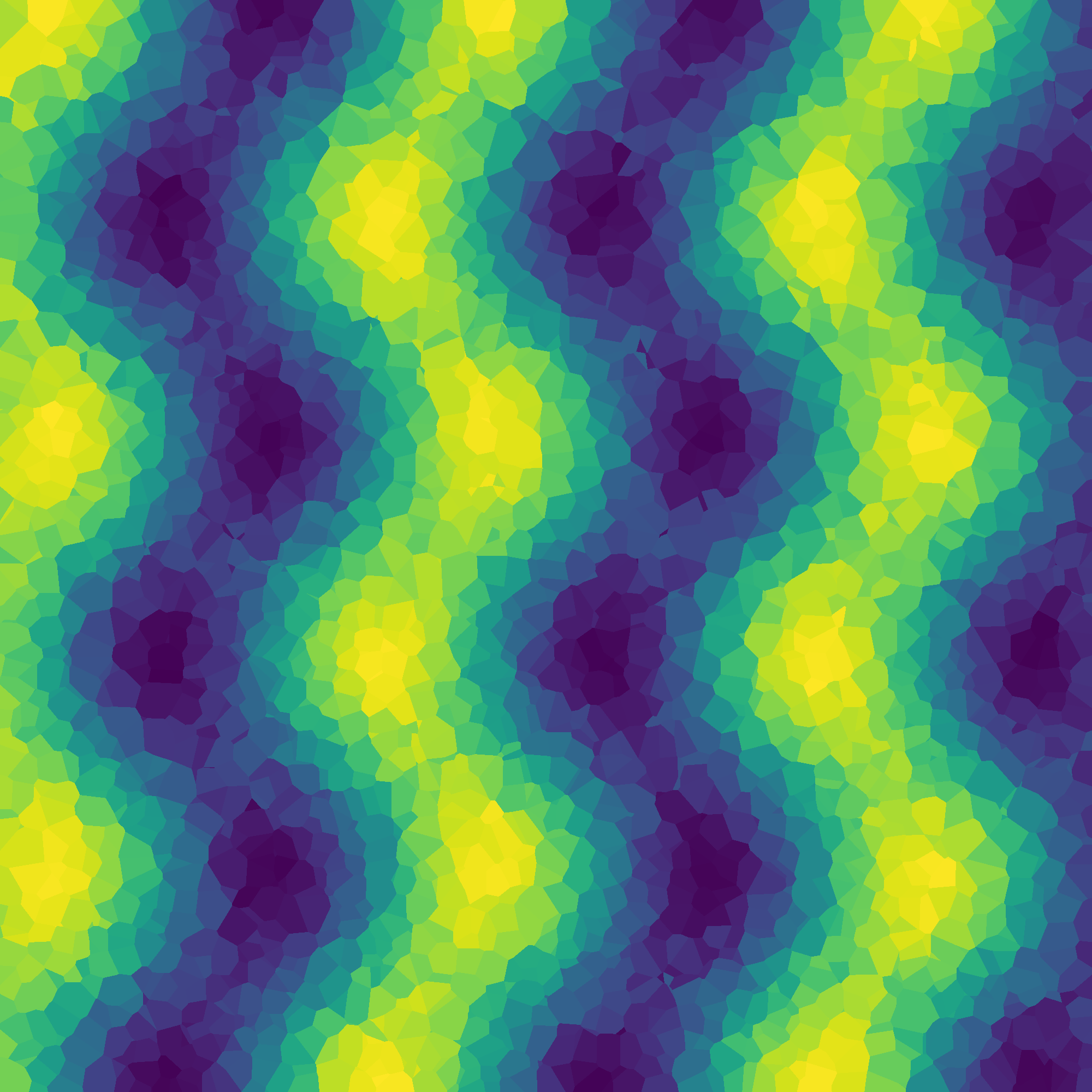} &
\includegraphics[width=0.1\linewidth]{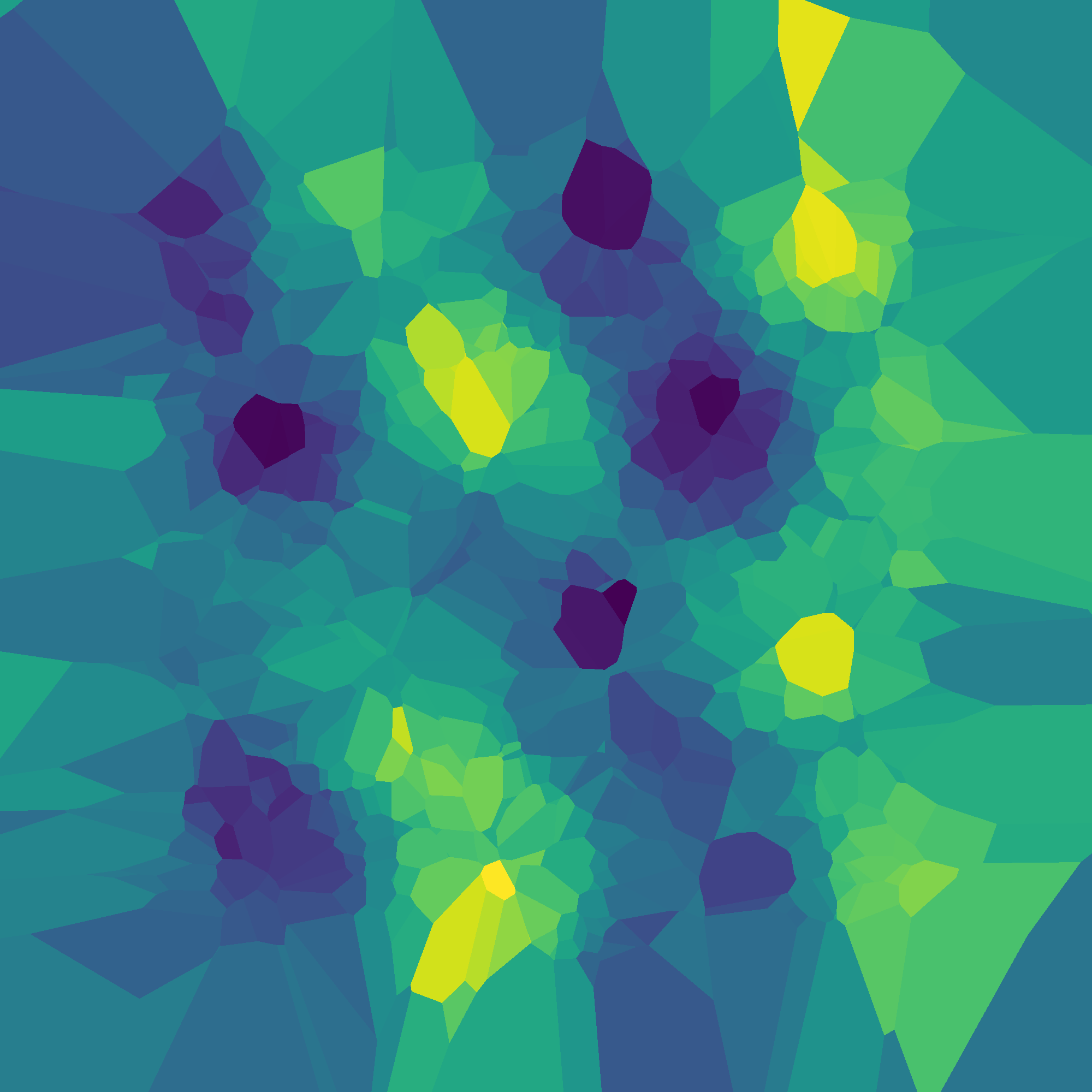} &
\includegraphics[width=0.1\linewidth]{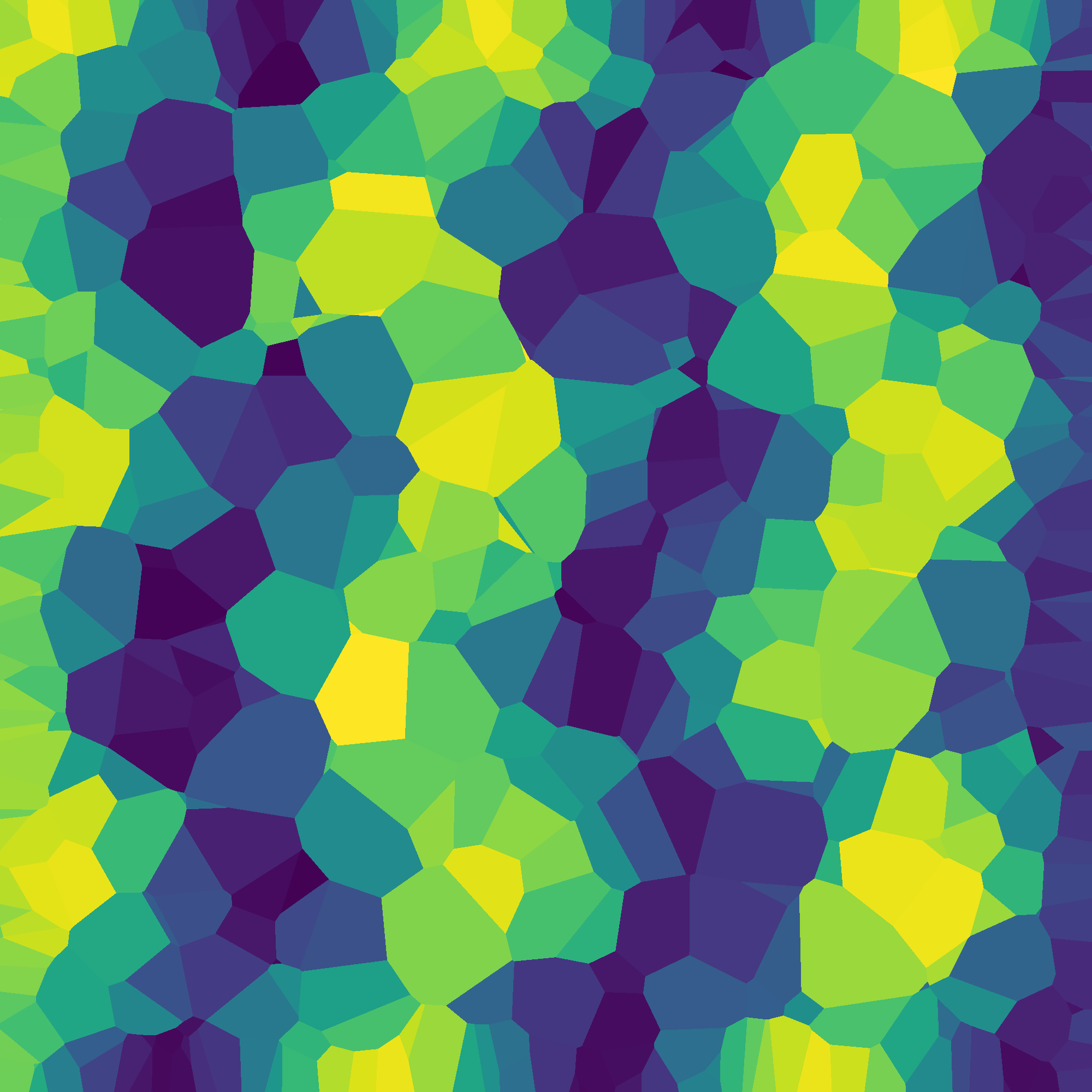} &
\includegraphics[width=0.1\linewidth]{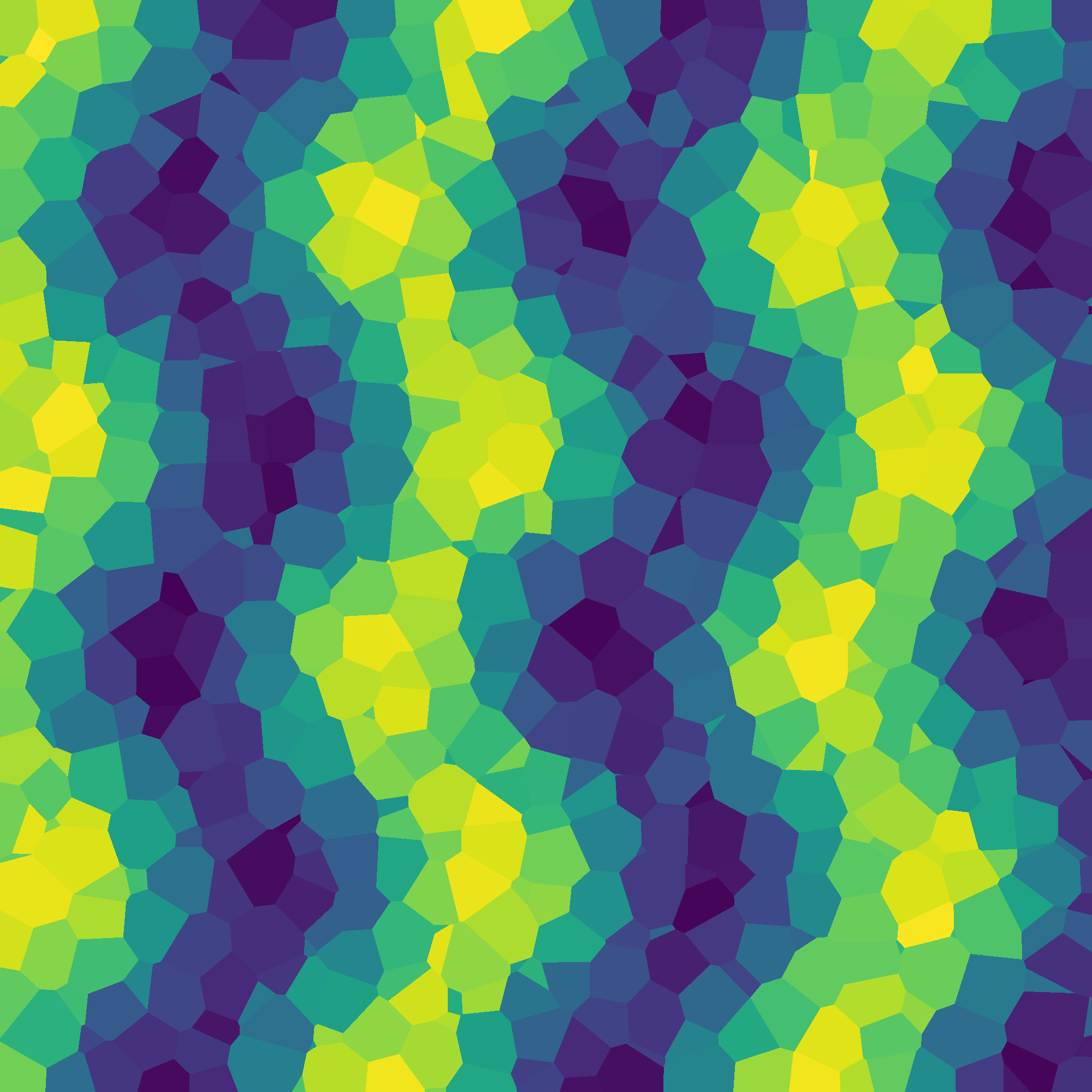} &
\includegraphics[width=0.1\linewidth]{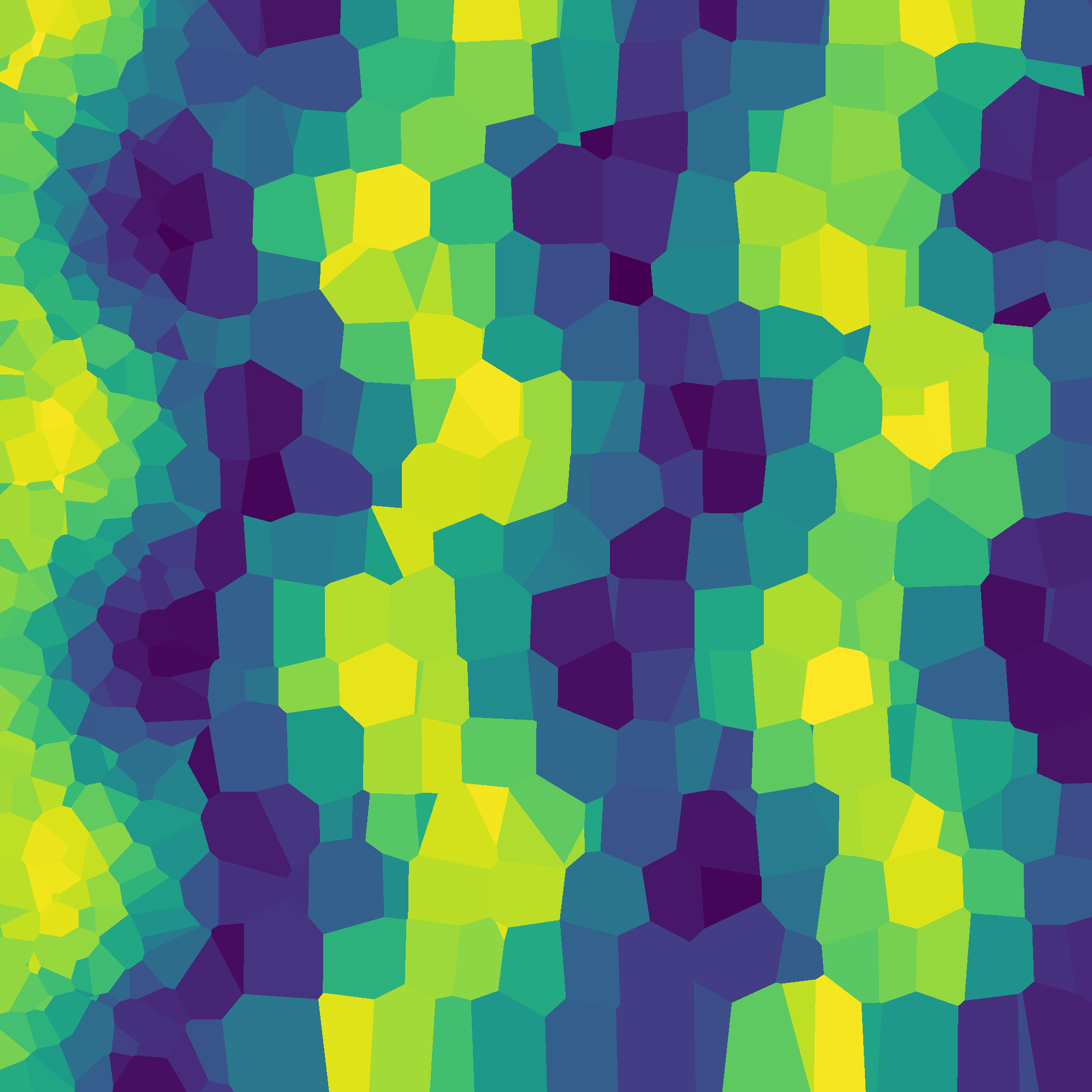} &
\includegraphics[width=0.1\linewidth]{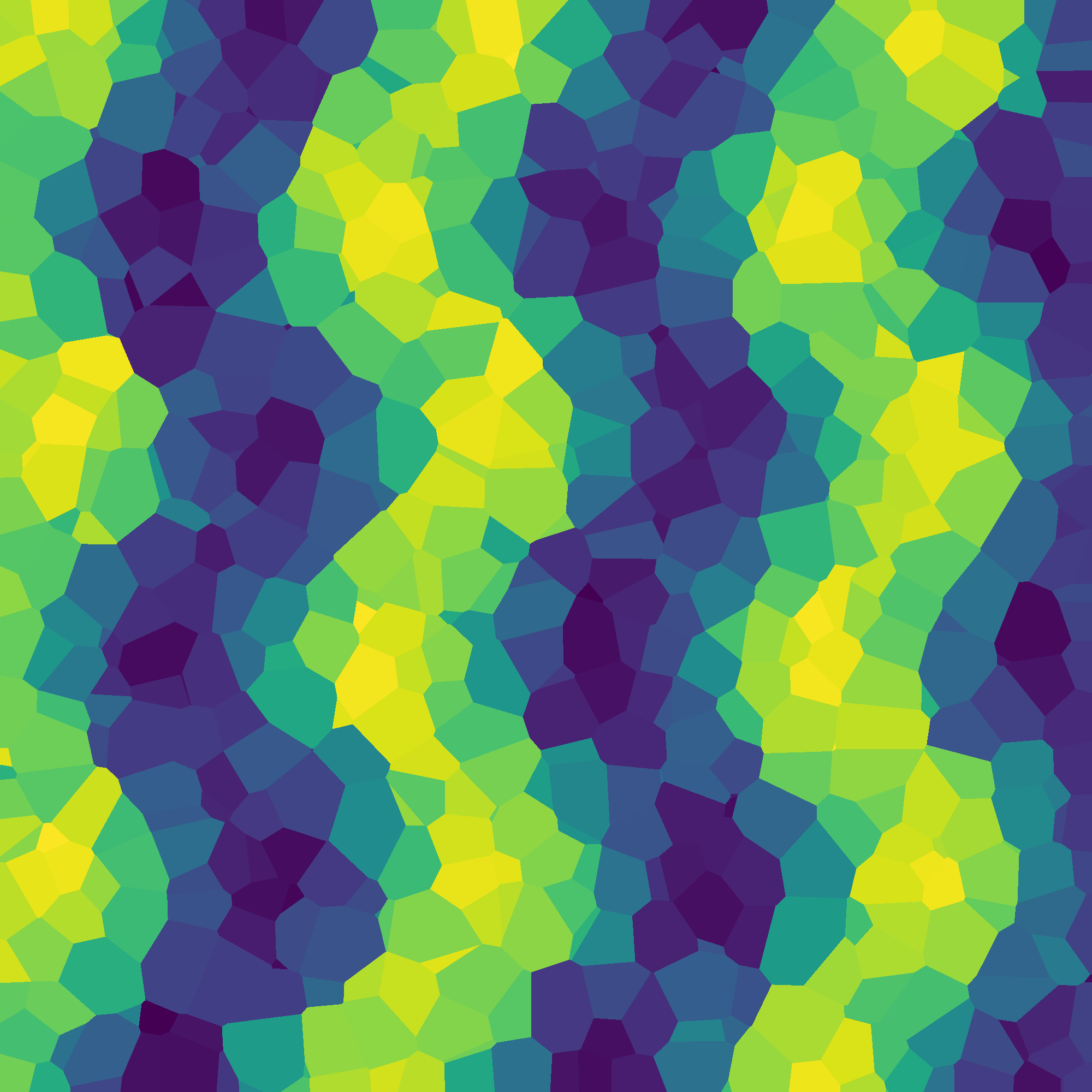} &
\includegraphics[width=0.1\linewidth]{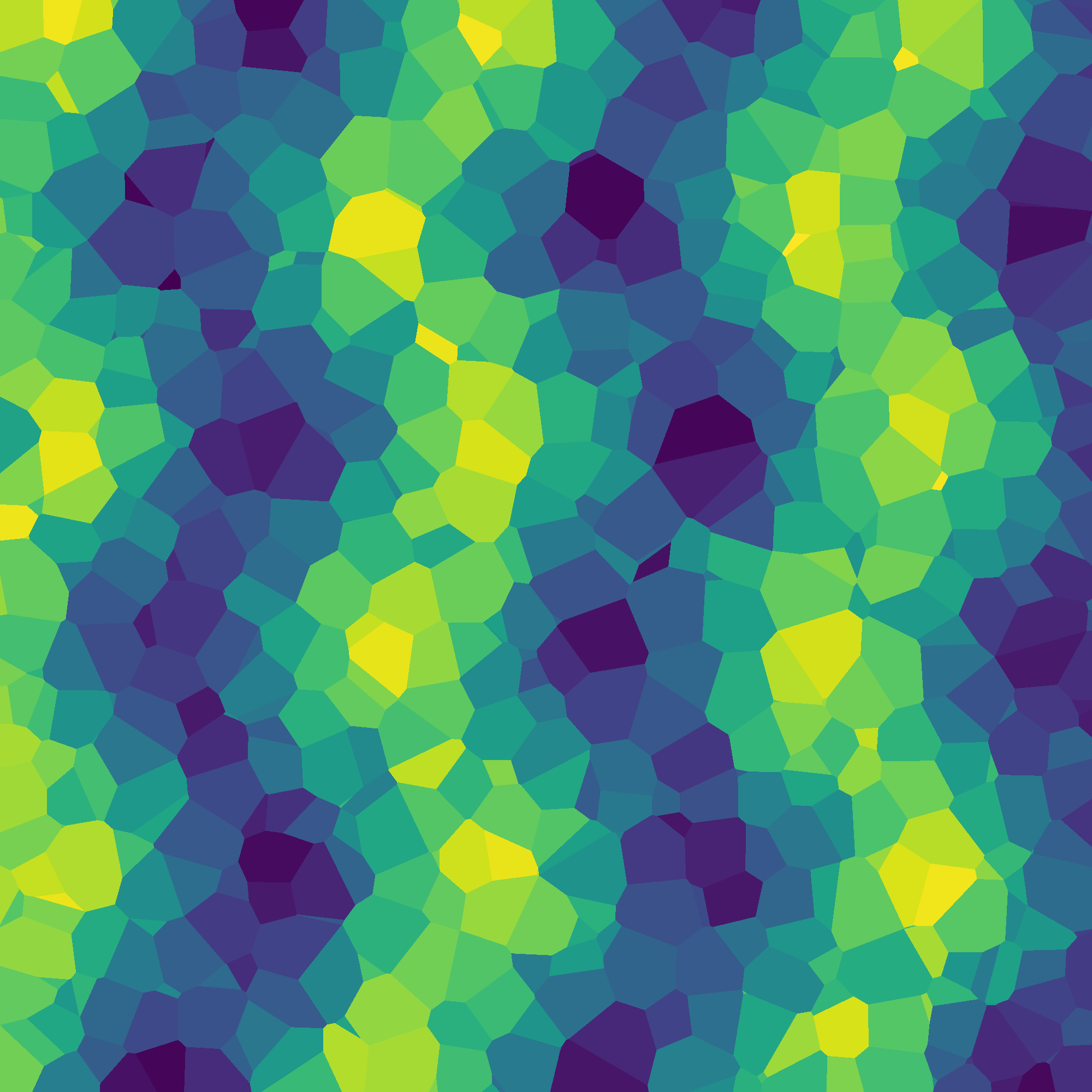} &
\includegraphics[width=0.1\linewidth]{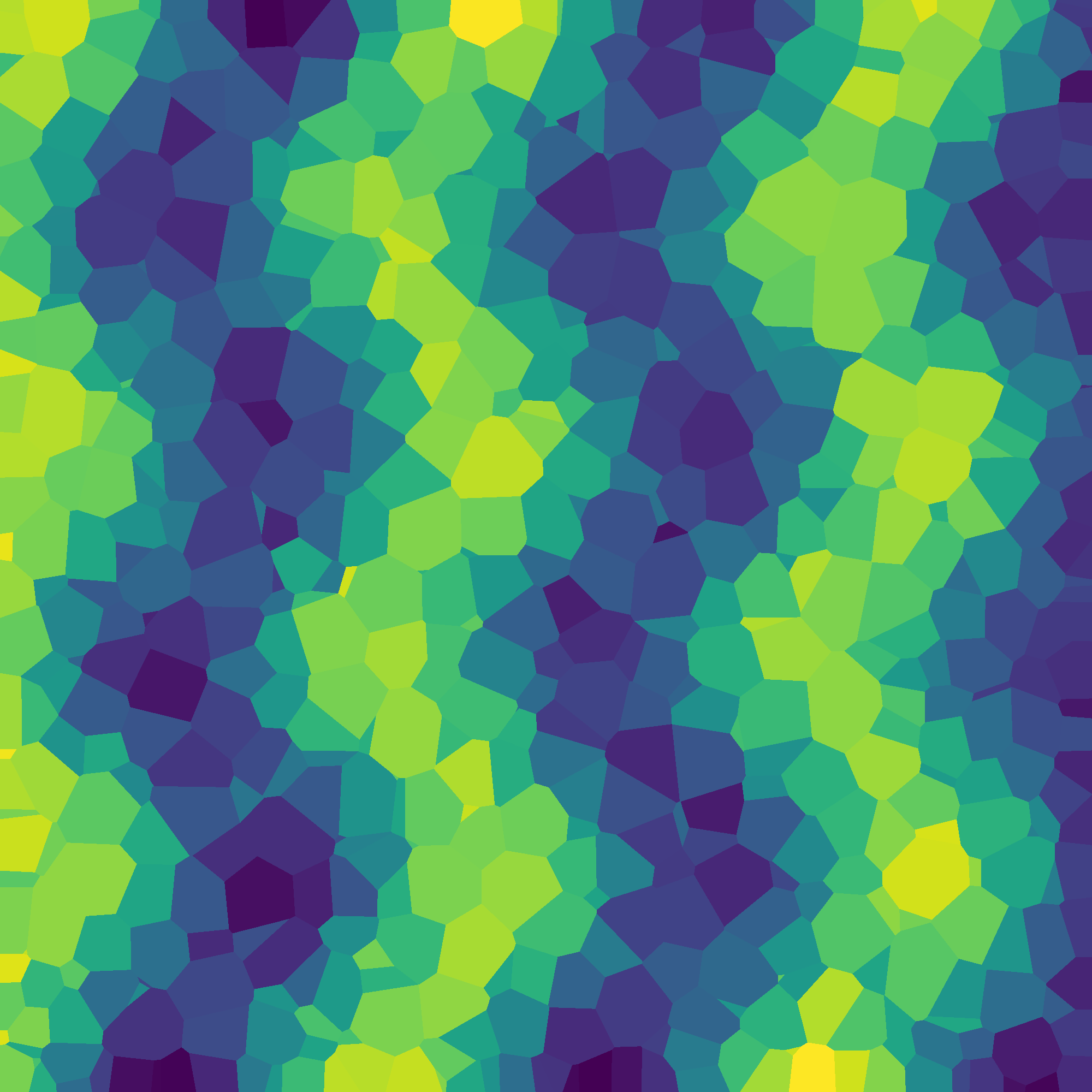} &
\includegraphics[width=0.1\linewidth]{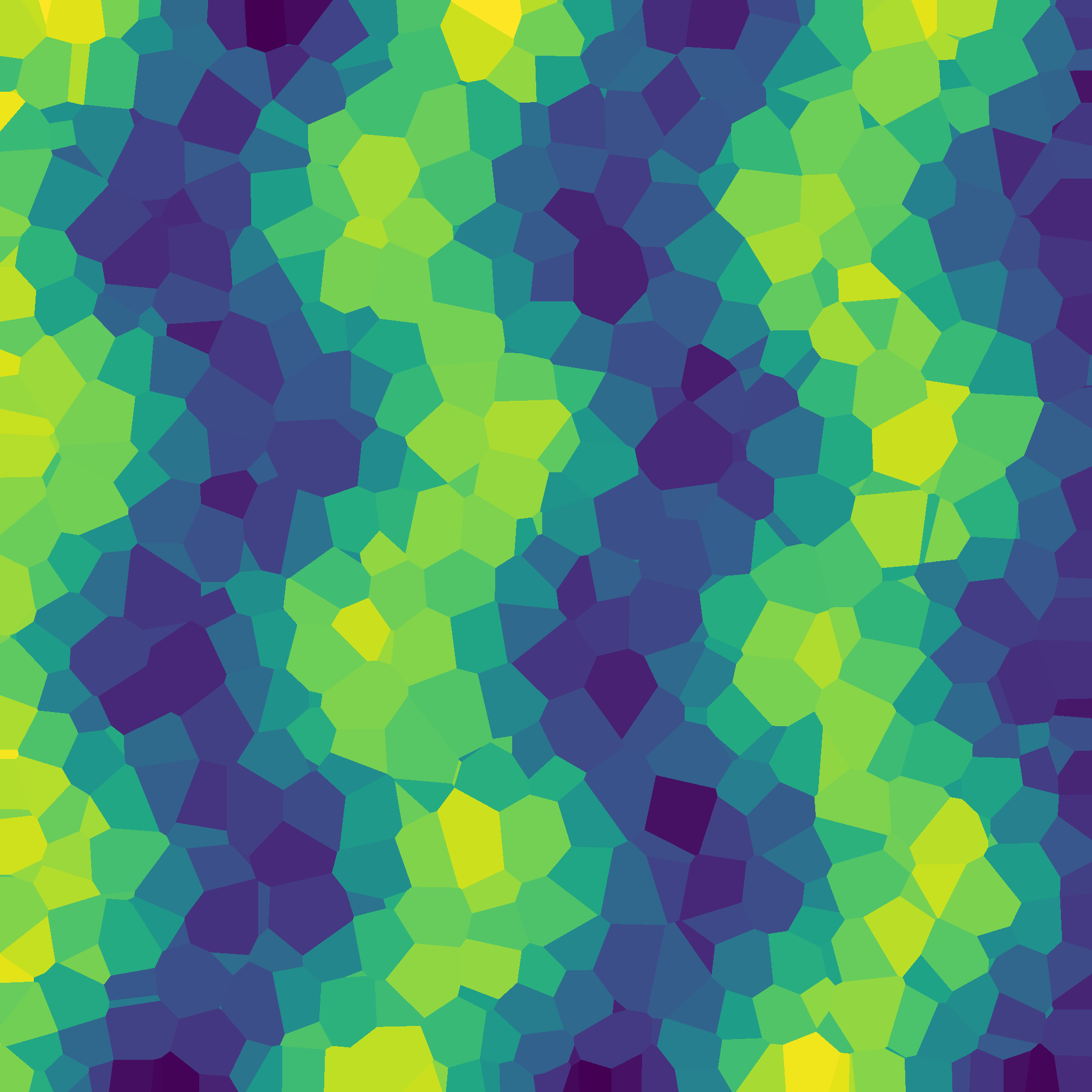} \\[1pt]
%\rotlabel{\scriptsize CTHead \\ (gradient)} &
\includegraphics[ width=0.1\linewidth, angle = 270, origin=c]{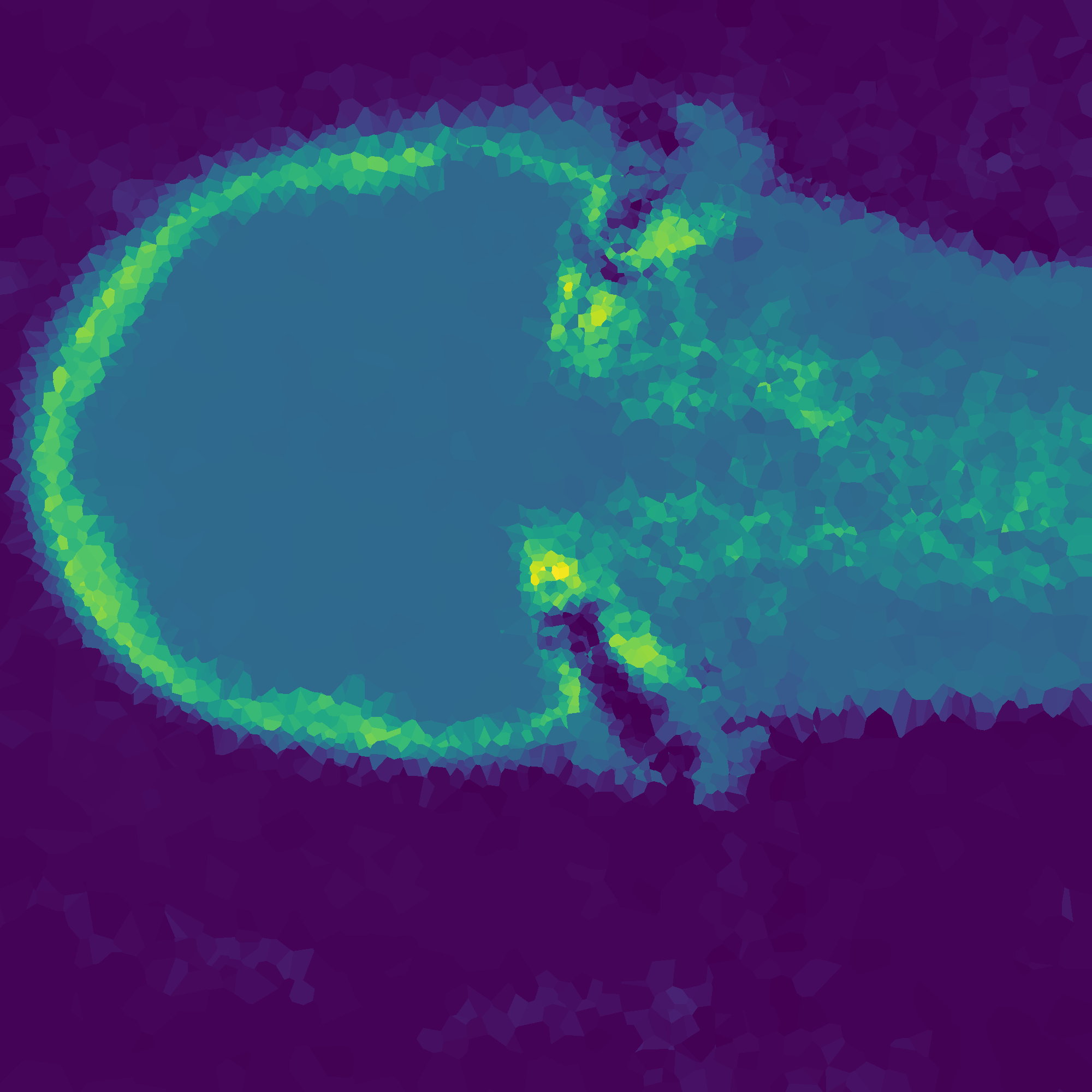} &
\includegraphics[ width=0.1\linewidth, angle = 270, origin=c]{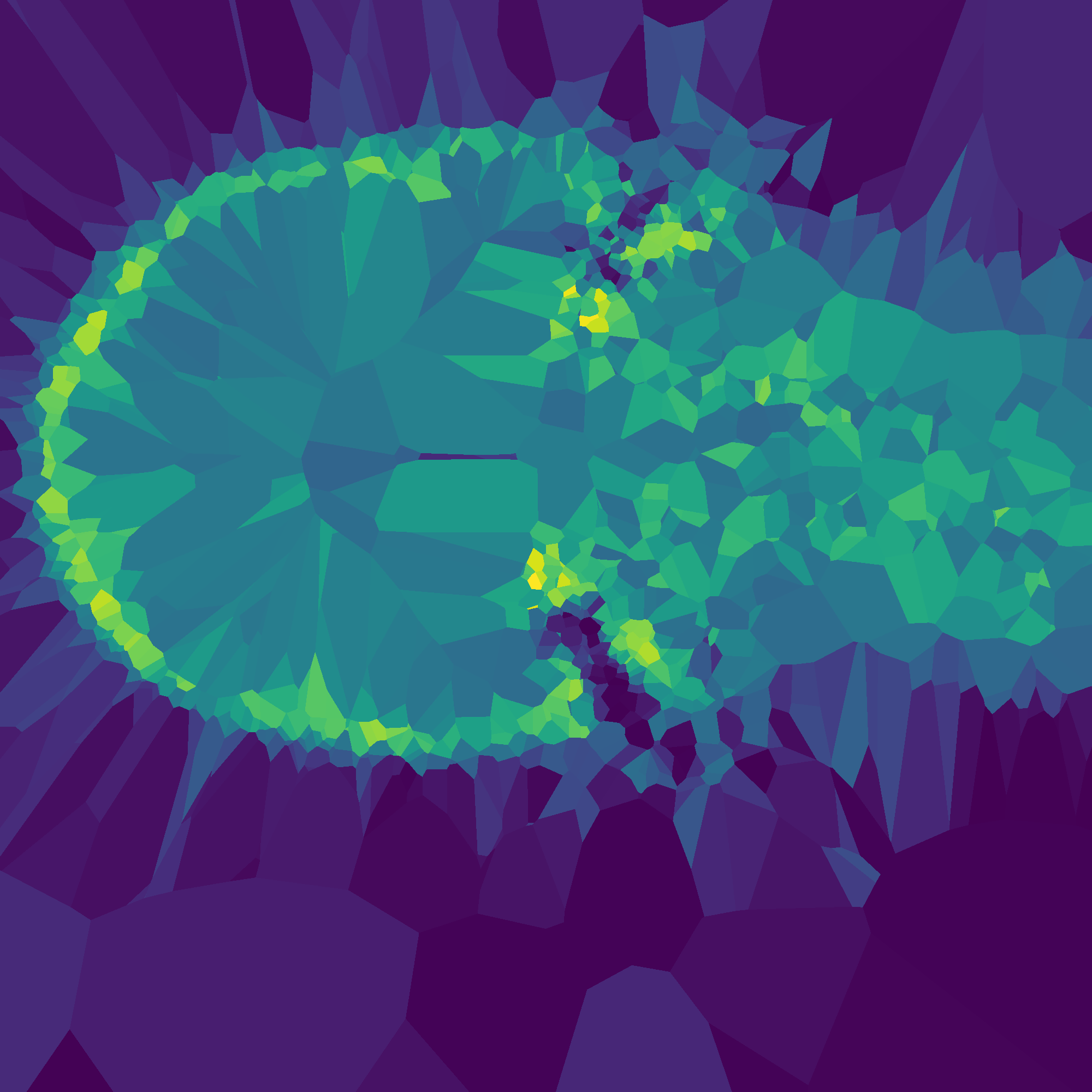} &
\includegraphics[ width=0.1\linewidth, angle = 270, origin=c]{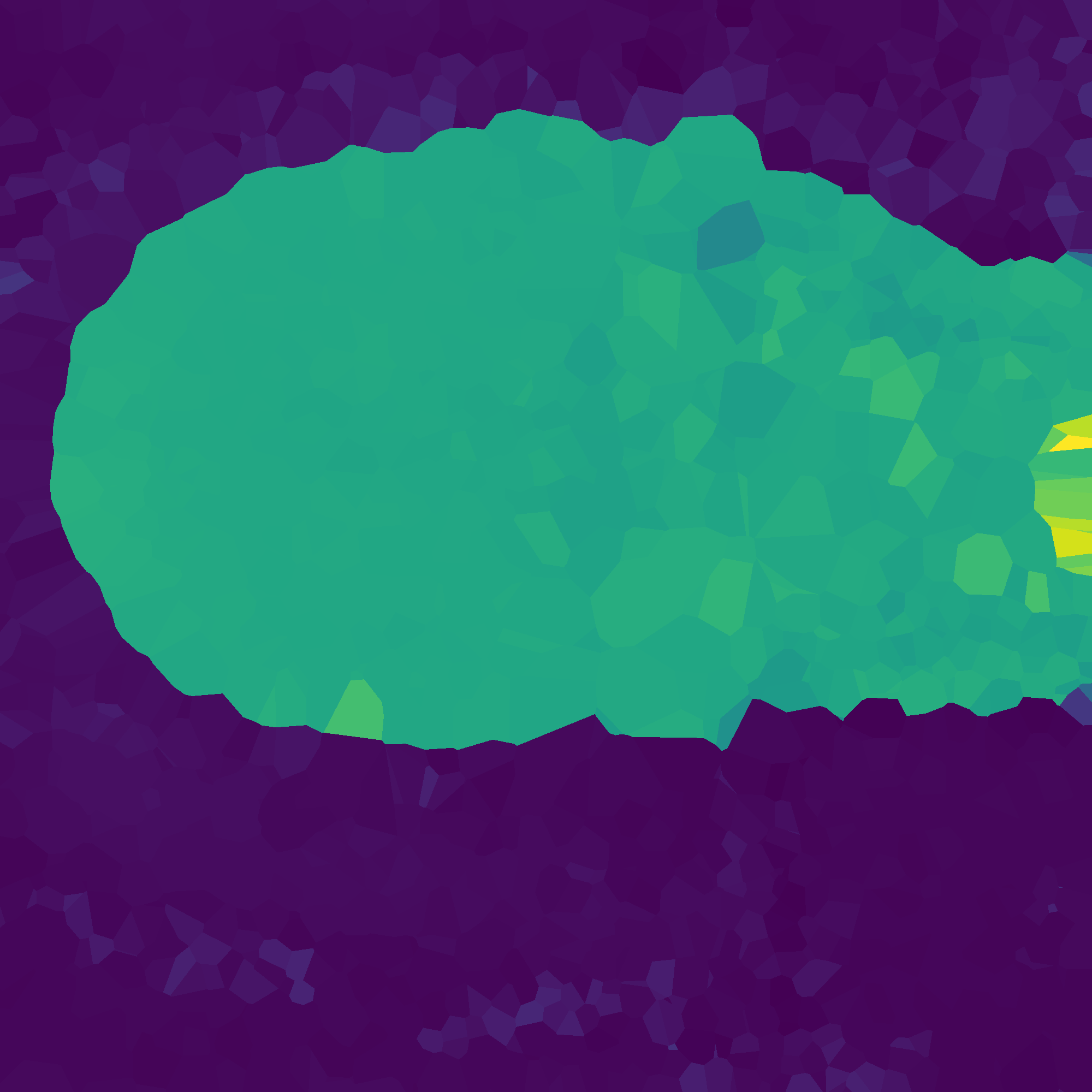} &
\includegraphics[ width=0.1\linewidth, angle = 270, origin=c]{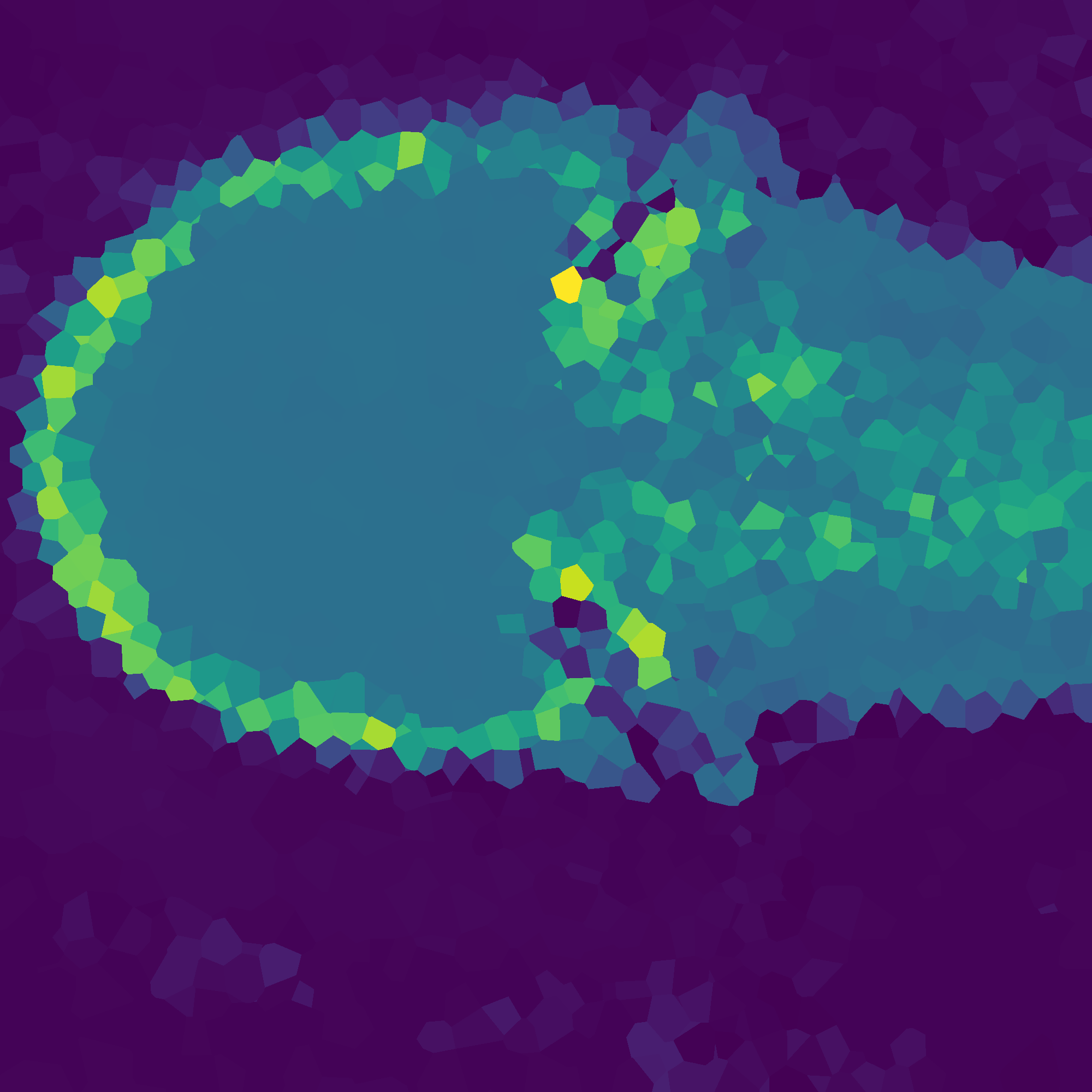} &
\includegraphics[ width=0.1\linewidth, angle = 270, origin=c]{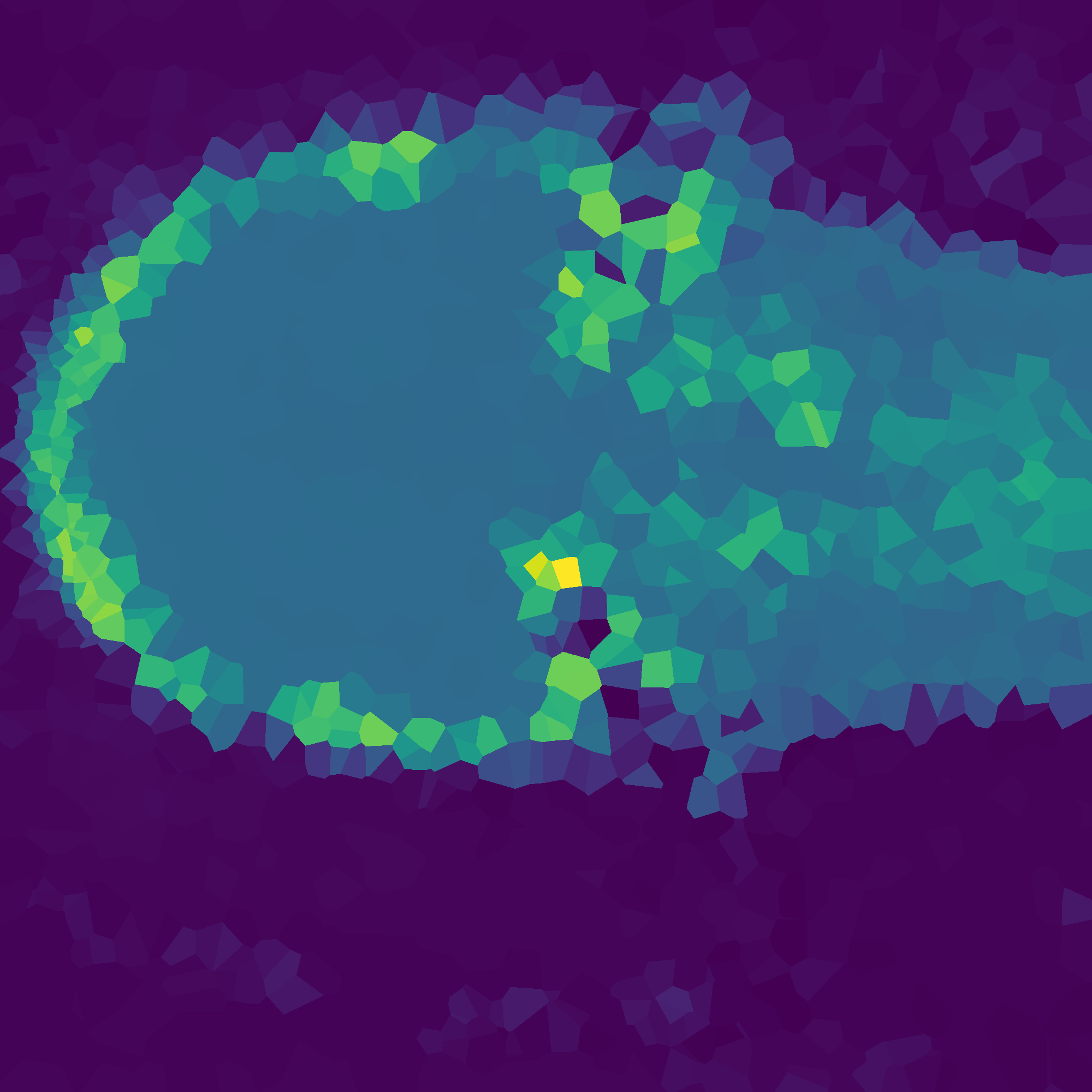} &
\includegraphics[ width=0.1\linewidth, angle = 270, origin=c]{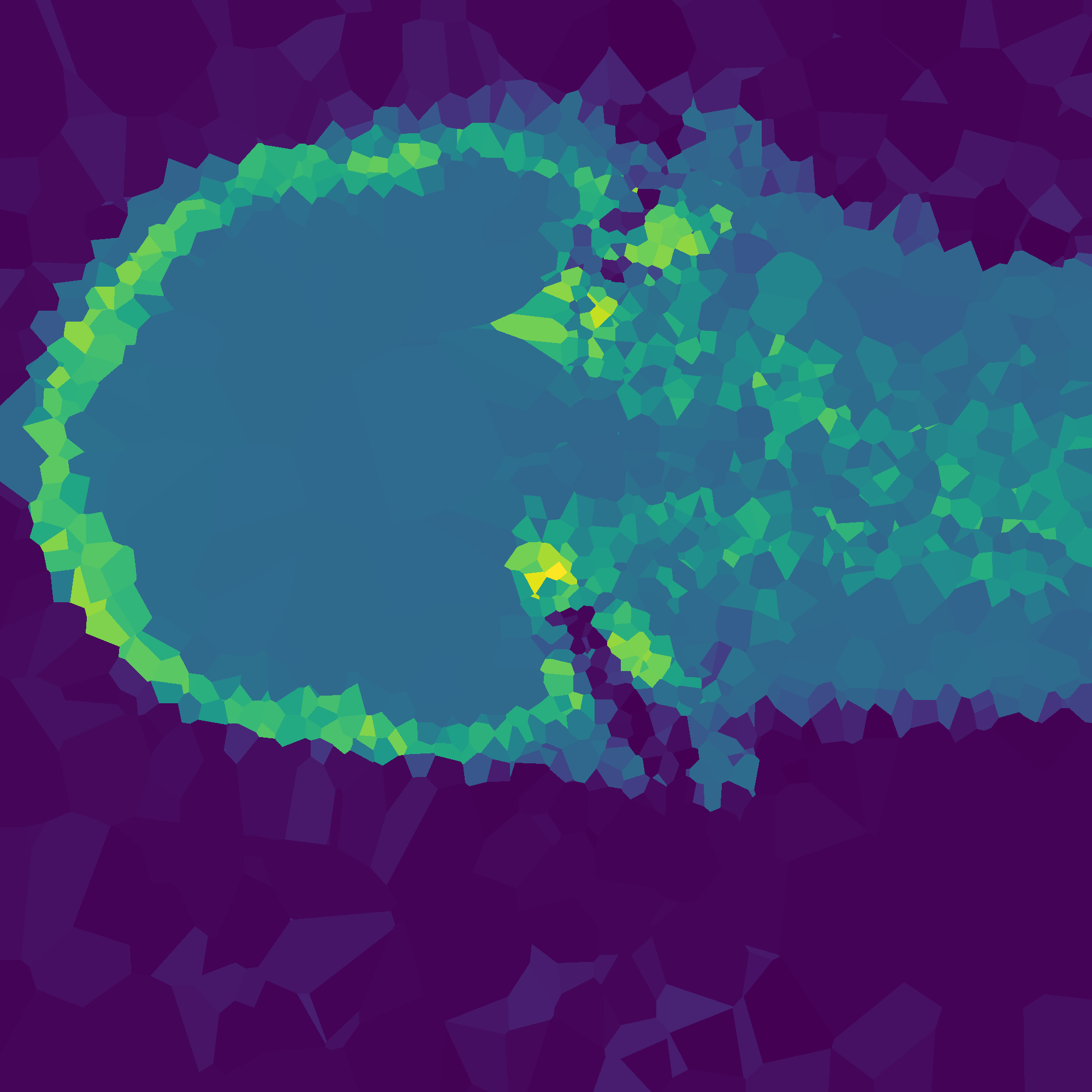} &
\includegraphics[ width=0.1\linewidth, angle = 270, origin=c]{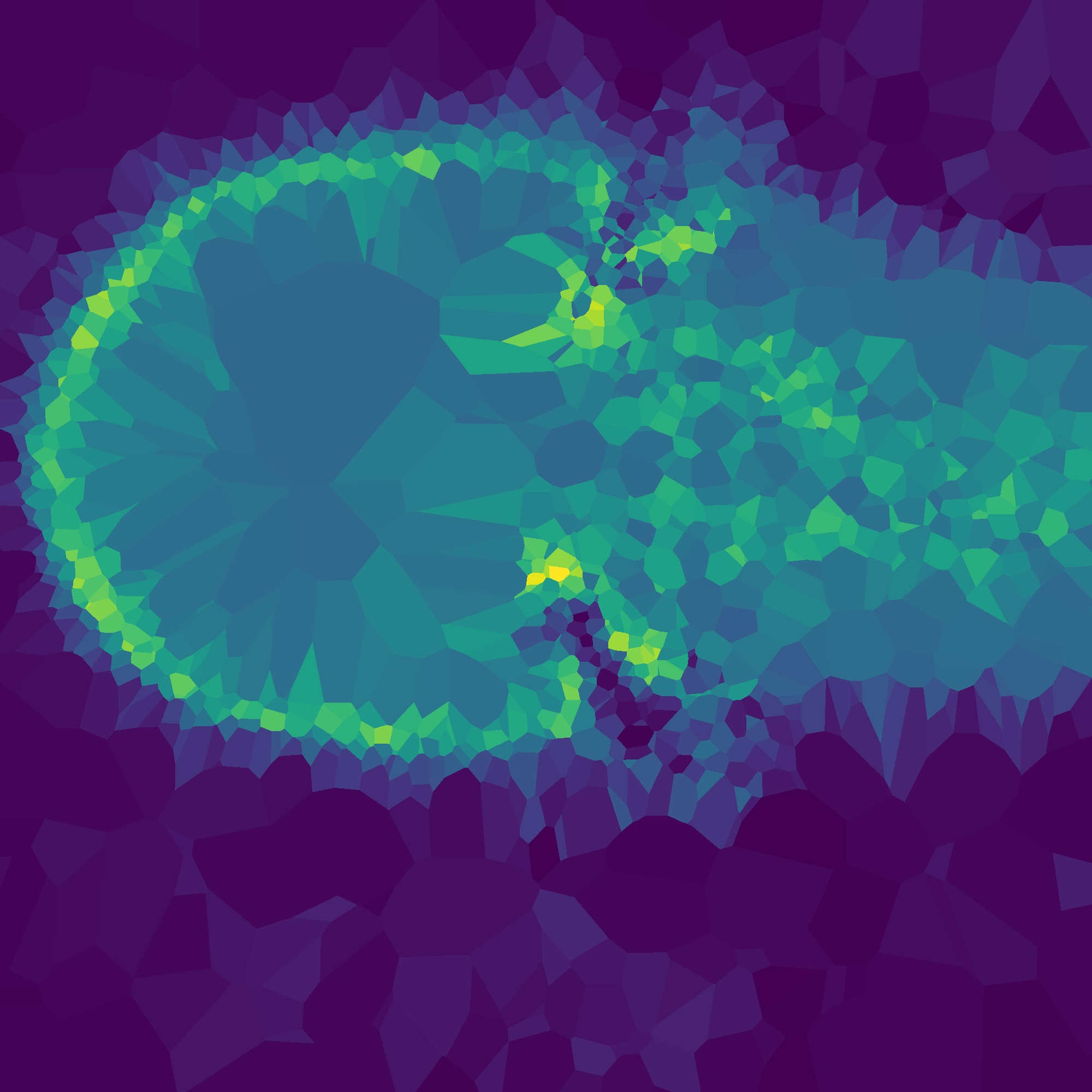} &
\includegraphics[ width=0.1\linewidth, angle = 270, origin=c]{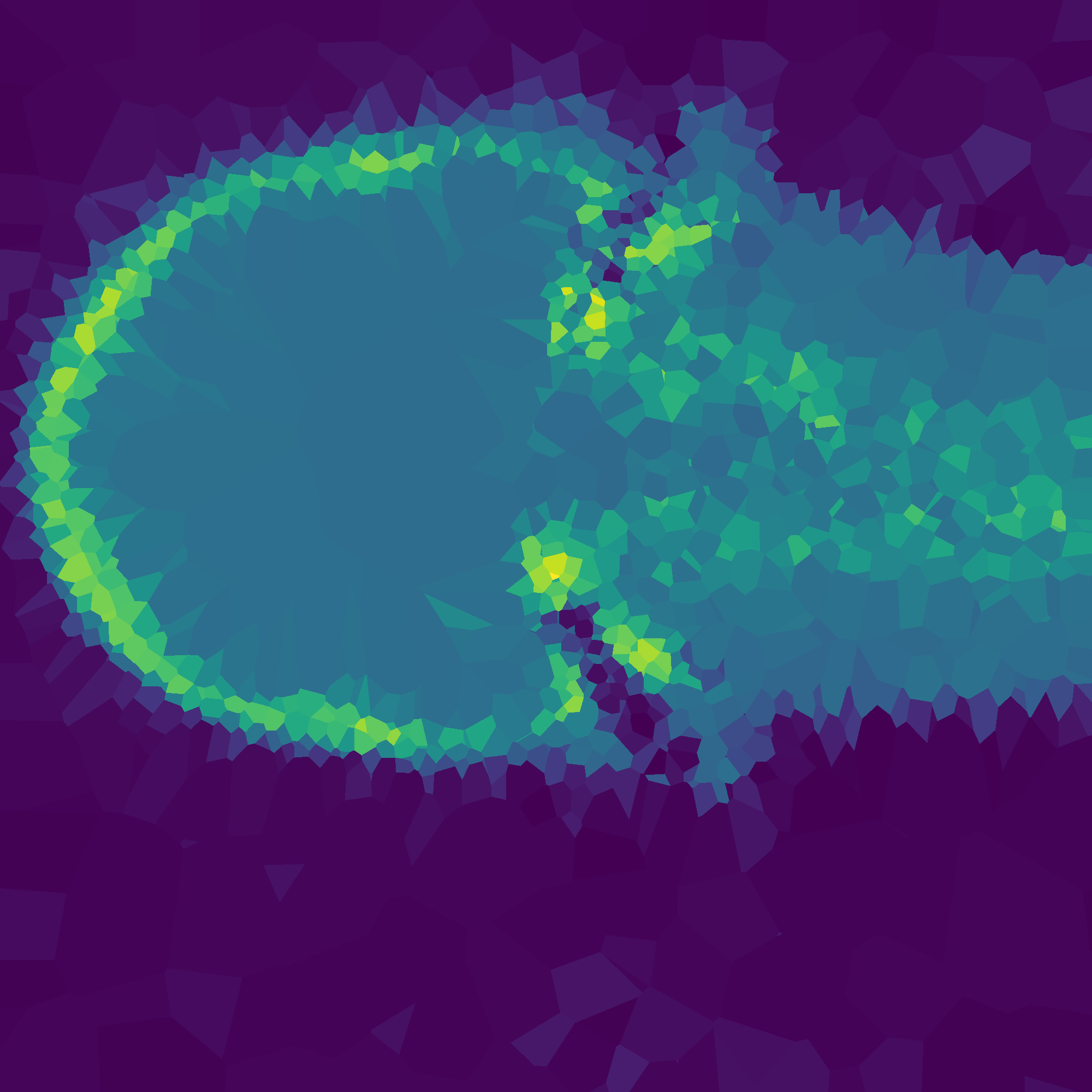} &
\includegraphics[ width=0.1\linewidth, angle = 270, origin=c]{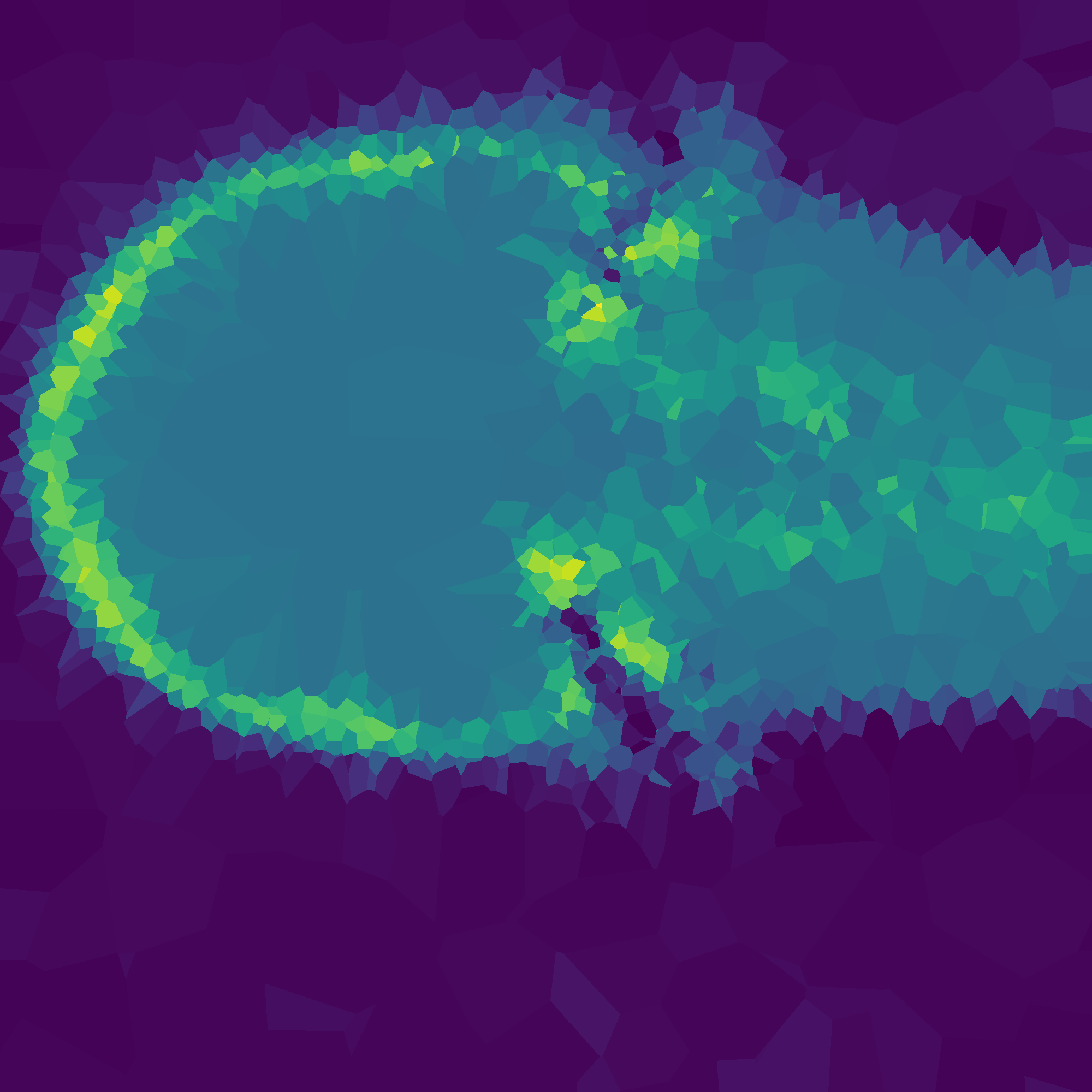} \\[1pt]
\includegraphics[ width=0.1\linewidth, angle = 270, origin=c]{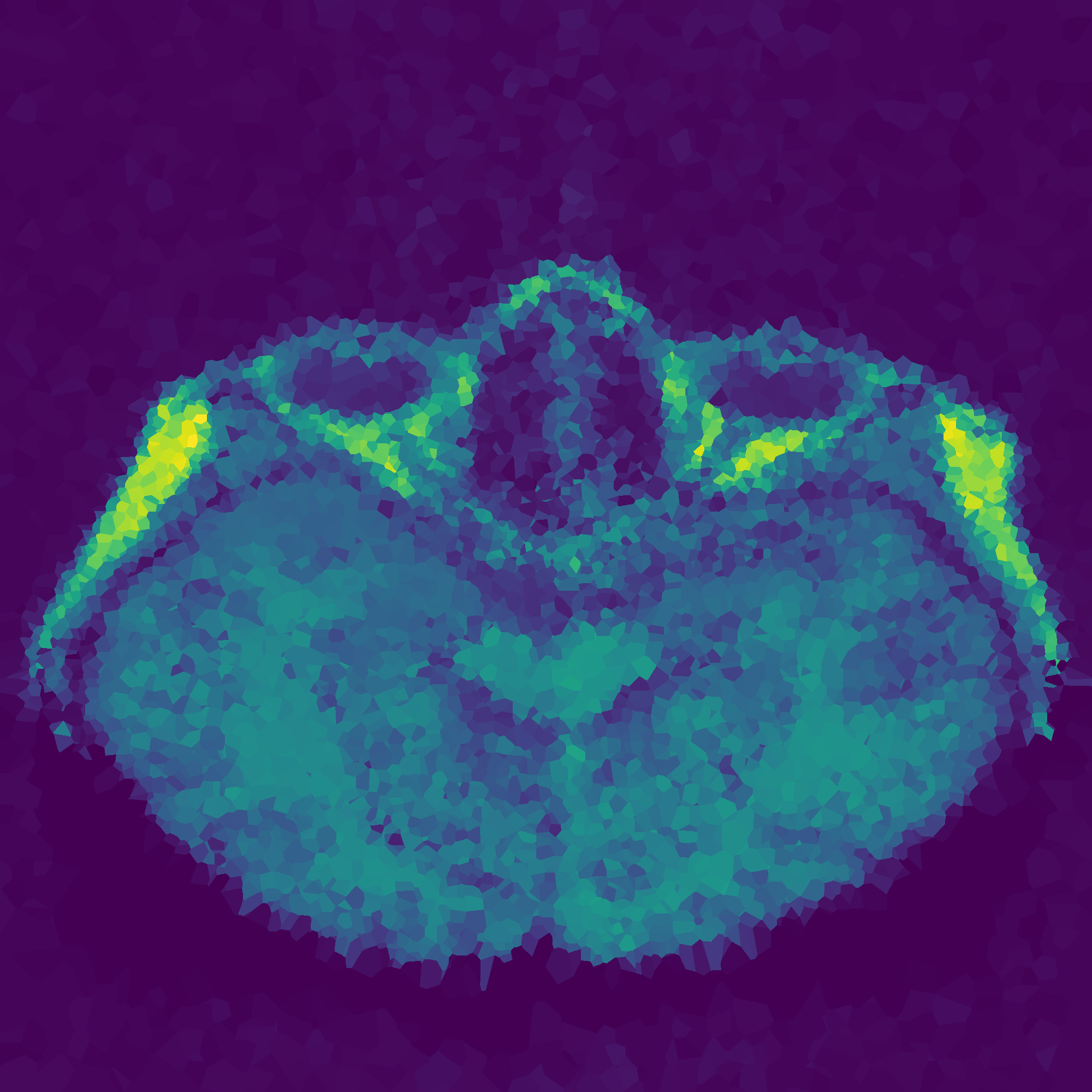} &
\includegraphics[ width=0.1\linewidth, angle = 270, origin=c]{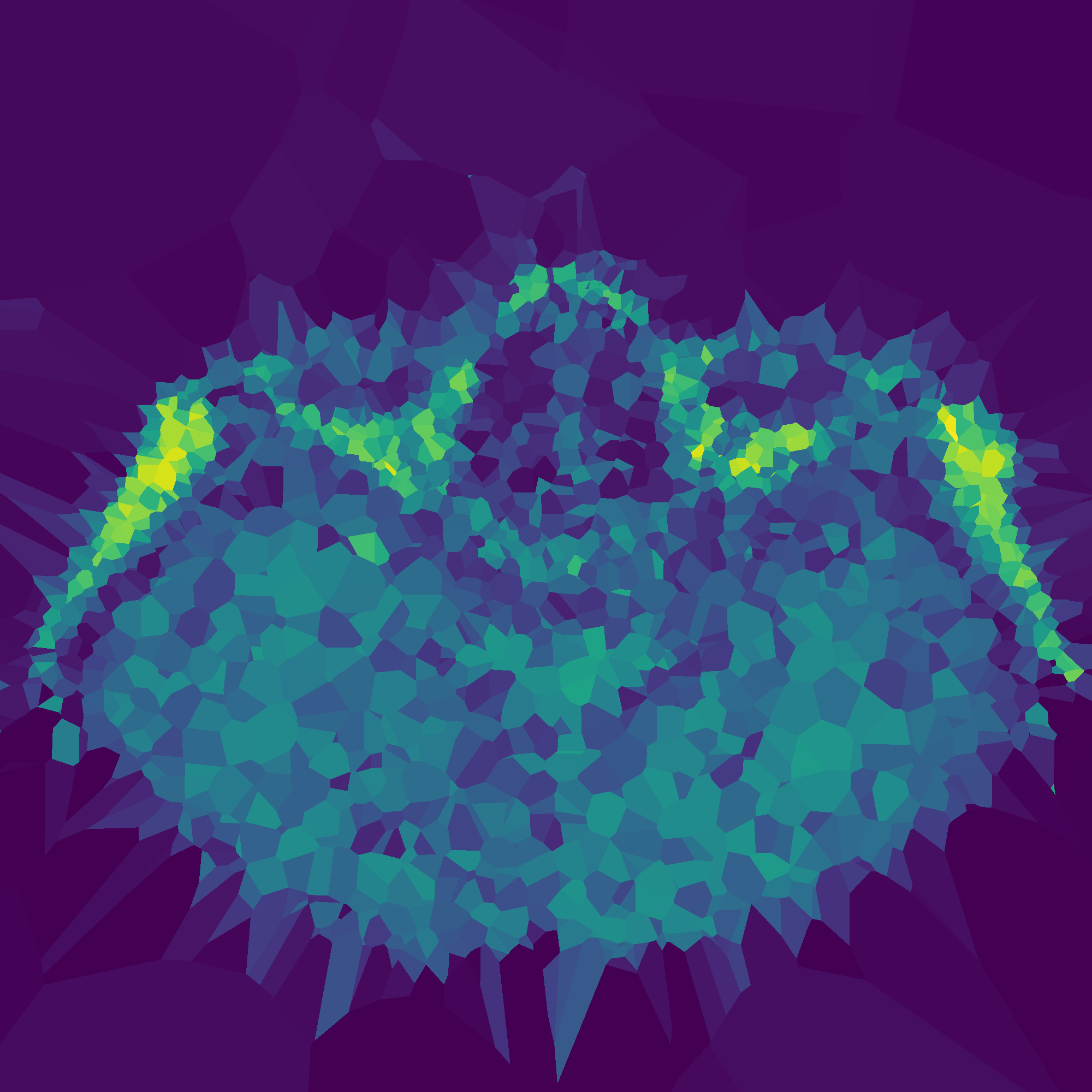} &
\includegraphics[ width=0.1\linewidth, angle = 270, origin=c]{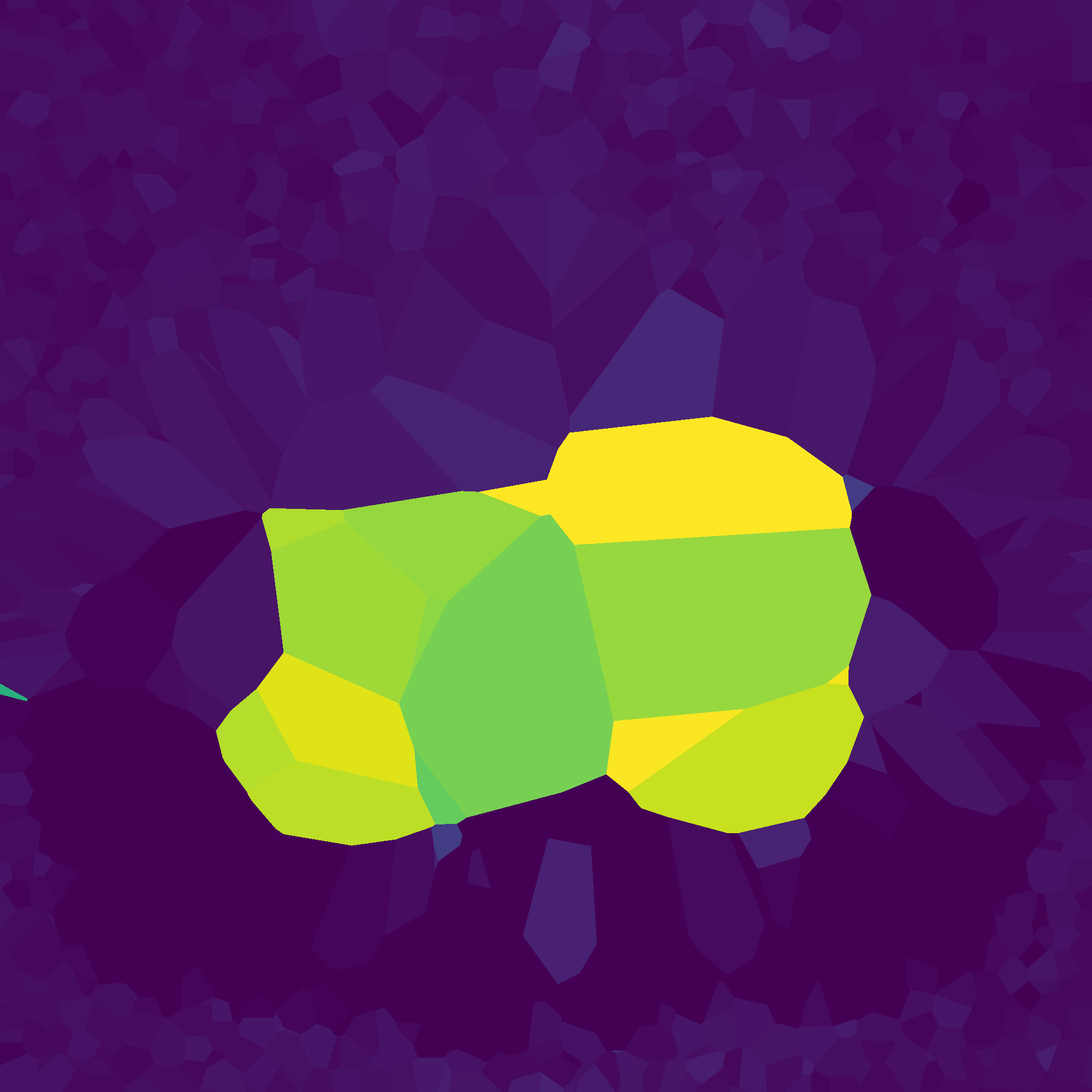} &
\includegraphics[ width=0.1\linewidth, angle = 270, origin=c]{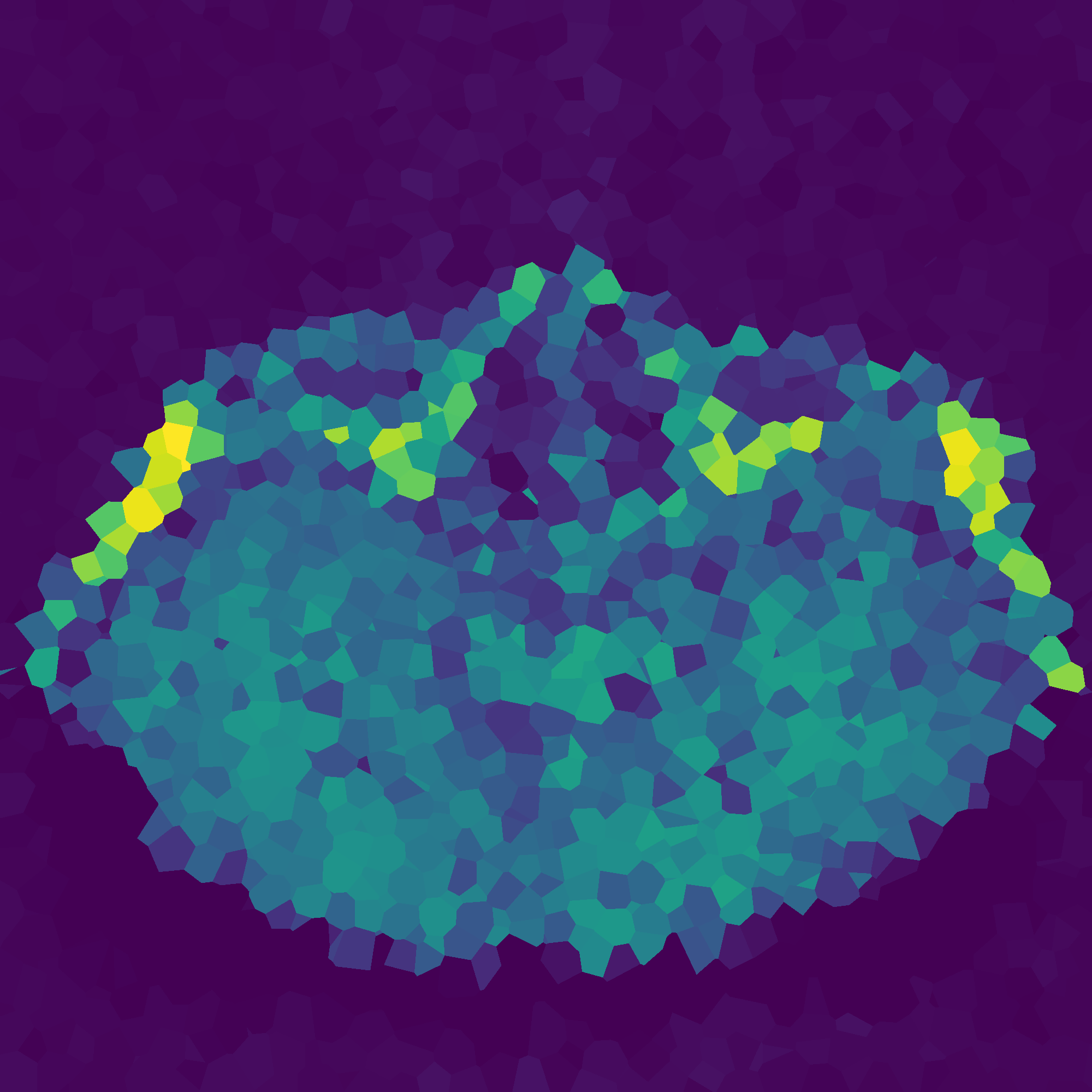} &
\includegraphics[ width=0.1\linewidth, angle = 270, origin=c]{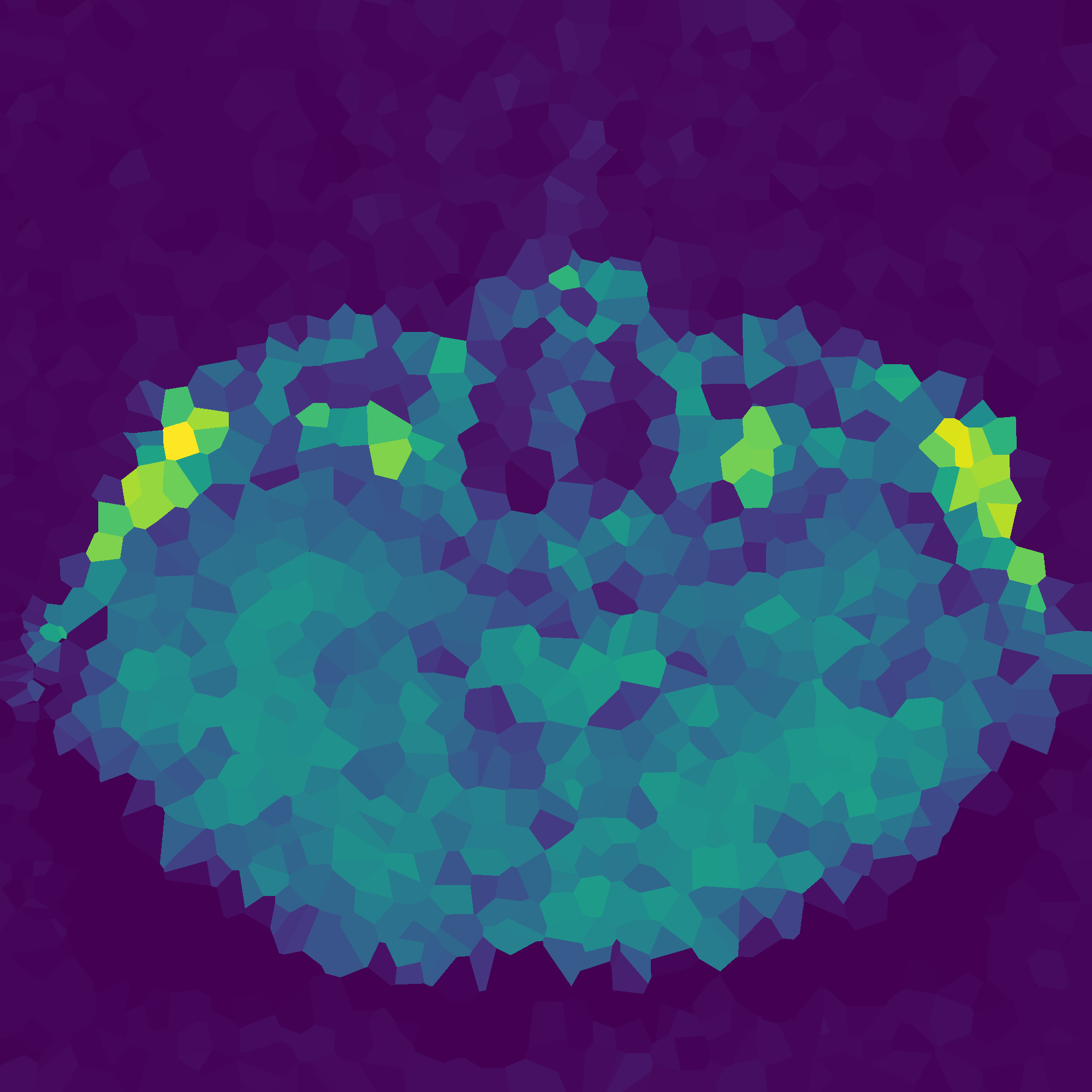} &
\includegraphics[ width=0.1\linewidth, angle = 270, origin=c]{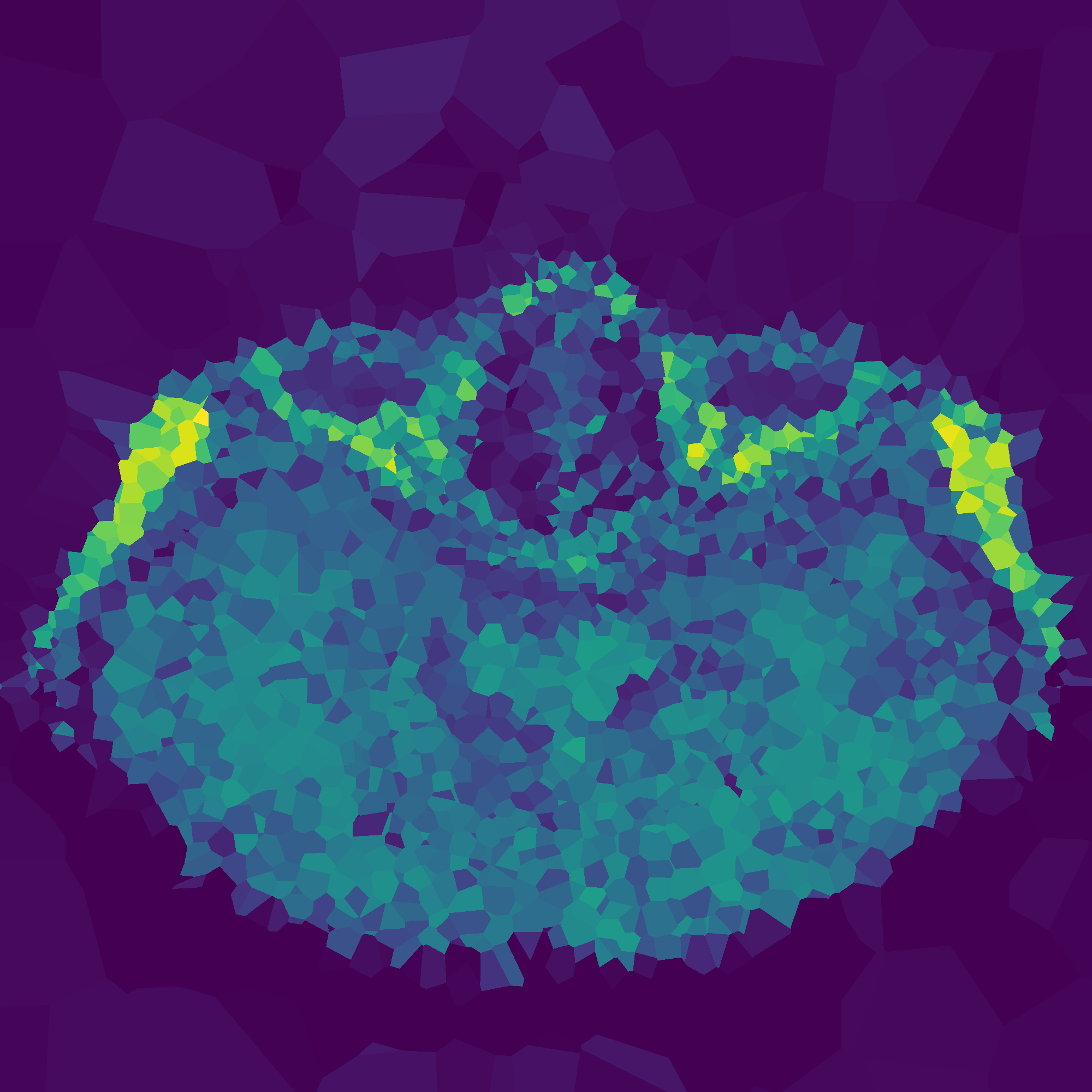} &
\includegraphics[ width=0.1\linewidth, angle = 270, origin=c]{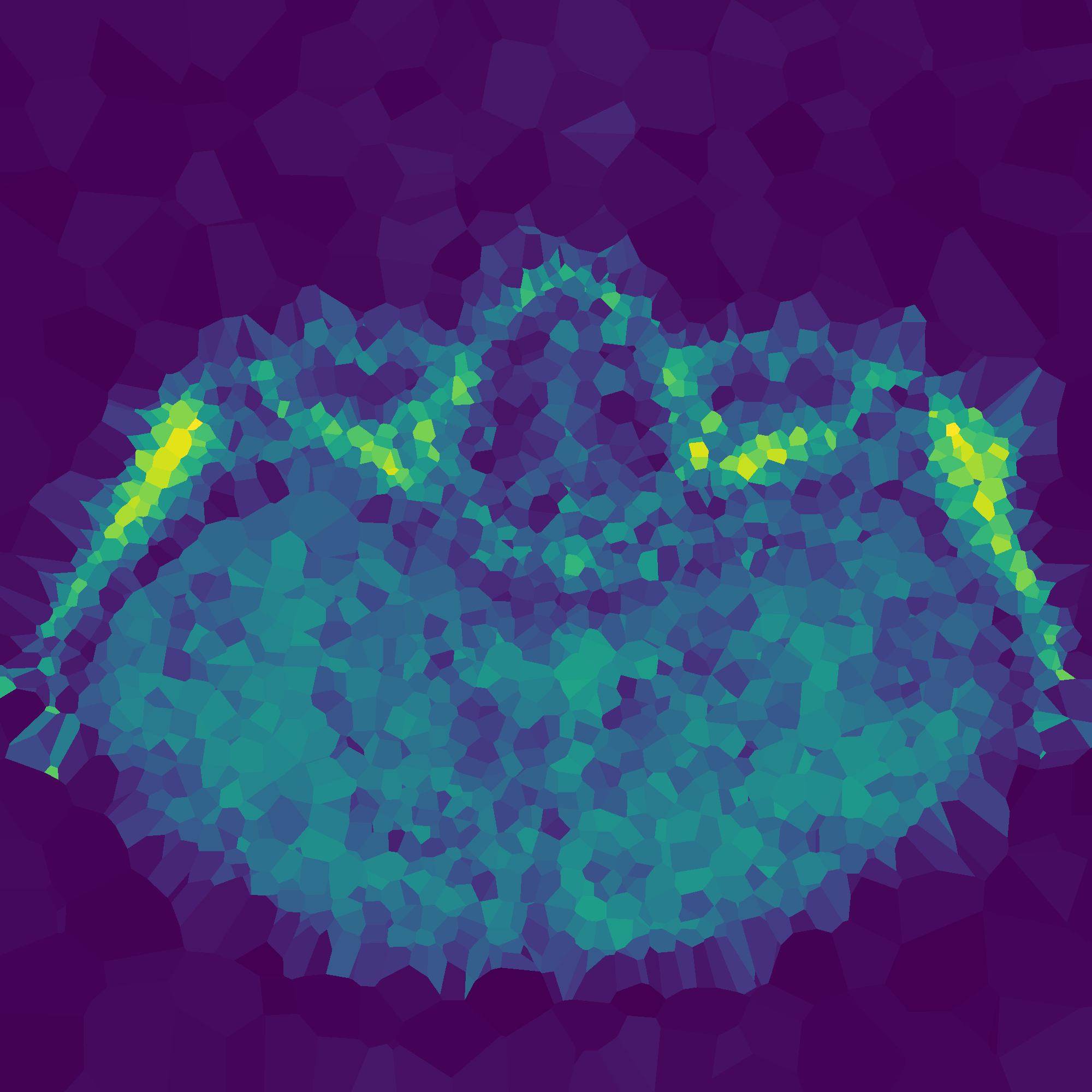} &
\includegraphics[ width=0.1\linewidth, angle = 270, origin=c]{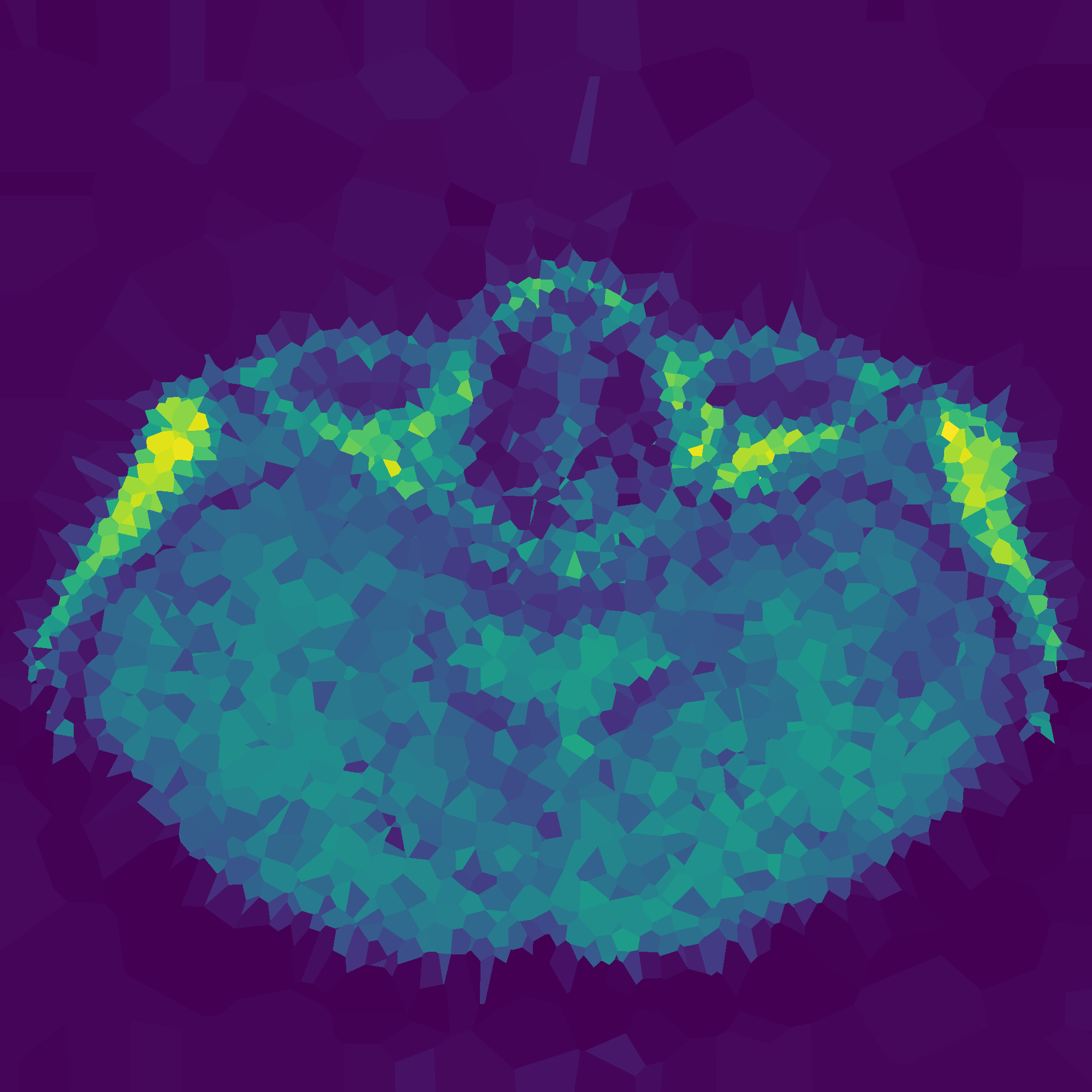} &
\includegraphics[ width=0.1\linewidth, angle = 270, origin=c]{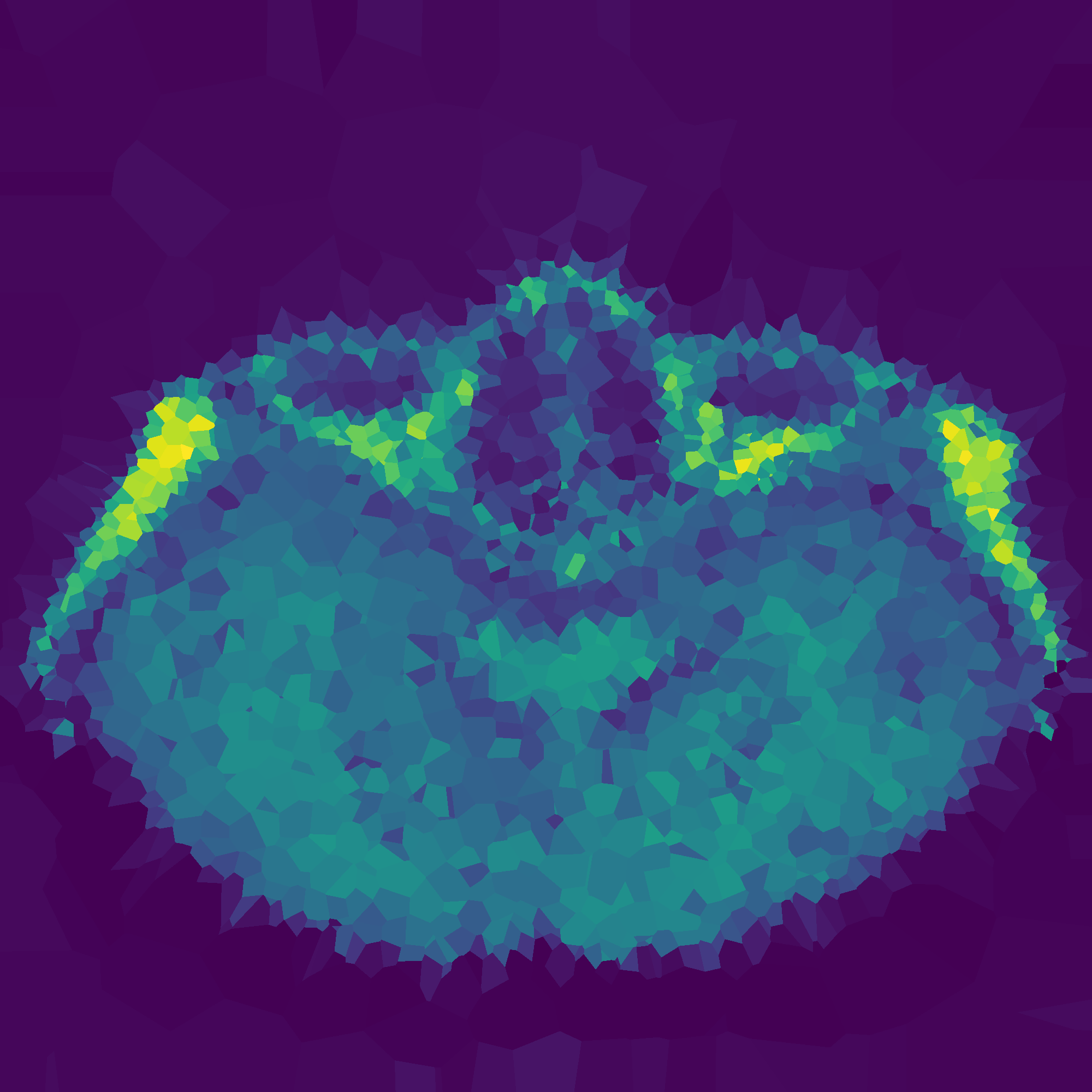} \\

%\addlinespace[4mm]
%\rotlabel{Row 3} &
%\includegraphics[width=0.09\linewidth]{example-image-a} &
%\includegraphics[width=0.09\linewidth]{example-image-a} &
%\includegraphics[width=0.09\linewidth]{example-image-a} &
%\includegraphics[width=0.09\linewidth]{example-image-a} &
%\includegraphics[width=0.09\linewidth]{example-image-a} &
%\includegraphics[width=0.09\linewidth]{example-image-a} &
%\includegraphics[width=0.09\linewidth]{example-image-a} &
%\includegraphics[width=0.09\linewidth]{example-image-a} &
%\includegraphics[width=0.09\linewidth]{example-image-a} \\
\end{tabular}
    \caption{Volume slice visualizations of models after reduction for the \emph{Gyroid smooth} (uniform) and \emph{CTHead} (gradient) datasets.}
    \label{fig:Visualizations}
\end{figure*}

\subsection{Data}
We evaluate our implementation across a wide variety of site distributions and underlying functions with both synthetic and real-world data. Figure~\ref{fig:data_slices} shows representative volume slices of evaluation inputs and Supplemental~\ref*{app:datasets} contains further description of how data was generated and all scenes are appended in unified format as Supplemental Material.

\paragraph*{Volumetric Scans}
To test the performance of our algorithm on real-world volumetric functions, we use data from the Stanford Volume Data archive~\cite{stanford_volume}, featuring CT scans of a human head and the Stanford Bunny, as well as an MR scan of a human head. The data is distributed as a dense voxel grid, which we downsample to $250,000$ sites using four different strategies: \texttt{uniform} (random uniform sampling), \texttt{importance} (sample proportional to voxel density), \texttt{gradient} (sample proportional to 3D Sobel gradient magnitudes) and \texttt{grid} (regular 3D lattice jittered to general position).

\paragraph*{Synthetic Dataset}
We generate a synthetic benchmark consisting of 17 controlled 3D datasets with 100k Voronoi sites each in a bounded, cubic volume. It was designed to systematically stress-test Voronoi decimation algorithms across decoupled combinations of spatial point densities and functional surface complexities, isolating algorithmic behaviour from real-world scan noise.

\paragraph*{Materials Science}
Polycrystalline materials naturally form Voronoi-like structures, making 3D Voronoi diagrams the natural representation for finite-element simulation and modelling~\cite{polycrystal}. Neper~\cite{neper} is a powerful software to construct and  run simulations on polycrystal structures. We generate five different 3D Voronoi models that emulate real-world material structures and compute the stress within the material to use as a function value for the reduction.

\subsection{Monte Carlo ISE}
Computing the exact ISE of a reduced point set is computationally infeasible for scenes with hundreds of thousands of cells, but can be effectively estimated without bias using Monte Carlo:
\begin{equation}
\hat{ISE} = \frac{Vol(\Omega)}{S}\sum_{i=1}^S \big(f(\mathbf{x}_i) - \hat{f}(\mathbf{x}_i)\big)^2,    
\end{equation}

with $S$ samples $\mathbf{x}_i$ uniformly sampled from the bounding box $\Omega$ of input Voronoi sites. The approximated functions can be efficiently evaluated using a kd-tree for nearest-neighbour lookups. For our evaluation we choose $S=5\cdot 10^6$ to ensure a converged result. For better comparability, we report $\hat{MSE}=\hat{ISE}/Vol(\Omega)$. To compute Area Under Curve (AUC), we run the Monte Carlo approximation every $5\%$ and linearly interpolate the values.

%\begin{table*}[htbp]
%\centering
%\caption{Real Astronomical Observations and Neper Microstructure Evaluation (90\% LOD Reduction). Values reported as Final MSE (Trajectory AUC Mean MSE) .}
%\label{tab:real_astro_neper}
%\begin{tabular}{lcccccc}
%\toprule
%Dataset Scene & Variance & Volume & FPS & Grid & Naive $L_2$ & \textbf{VoroDuce} \\
%\midrule
%sdss\_cosmic\_web & 1.030 (0.742) & 0.199 (0.050) & 0.395 (0.137) & 0.508 (0.151) & 0.903 (0.381) & \textbf{0.189 (0.042)} \\
%neper\_poisson\_100k & 0.082 (0.028) & 0.113 (0.028) & 0.067 (0.020) & 0.113 (0.033) & 0.071 (0.027) & \textbf{0.055 (0.015)} \\
%\midrule
%\textbf{Average} & 0.402 (0.265) & 0.144 (0.040) & 0.180 (0.063) & 0.244 (0.074) & 0.360 (0.147) & \textbf{0.117 (0.028)} \\
%\bottomrule
%\end{tabular}
%\end{table*}

%\subsection{Radiant Foam}

%Figure \ref{fig:Results} shows results for reducing primitive counts of models trained on scenes from the MipNeRF-360~\cite{MIPNERF360} (1 million cells each) and DTU~\cite{DTU} (250k cells each) datasets. We evaluate the visual error introduced in each step of the reduction using standard image quality metrics Peak Signal-to-Noise Ratio (PSNR), Structural Similarity Index Measure (SSIM) and Learned Perceptual Image Patch Similarity (LPIPS). Additionally, we plot the Monte Carlo estimate of the full scene ISE with $5\cdot 10^7$ samples. For the experiments, we use a collapse batch size $N=0.02$, meaning that at most 2\% of Voronoi sites are removed, before the Delaunay triangulation is recomputed.

\subsection{Results}
The resulting error metrics for all scenes and baselines are included in Table \ref{tab:combined_benchmarks} and corresponding ISE curves are plotted in Supplemental \ref*{app:ISE_curves}. Our method outperforms or matches baselines in almost all cases, both with fast and exact error computation. Inputs that use grid sampling are particularly challenging for our method, since a regular lattice contains many co-spherical points, representing a worst-case scenario for degenerate configurations in Delaunay triangulations. The instability of the underlying triangulation leads to sub-optimal results for our method. Exact volume computation outperforms the fast, approximate method in average, but only by a small margin. The biggest impact for the exact volume computation is shown on the more chaotic Microstructure dataset. This is also where Cignoni et al., the best performing baseline, places last among all tested methods.  Although very similar in structure to VoroDuce, Ward performs consistently poorly, showing the importance of considering both the geometry and topology of the Voronoi mesh. The sliced volume visualizations of resulting models in Figure~\ref{fig:Visualizations} show how VoroDuce more faithfully preserves fine details and overall structure. 
\definecolor{topone}{HTML}{9ECAE1}
\definecolor{toptwo}{HTML}{C6DBEF}
\definecolor{topthree}{HTML}{DEEBF7}
\newcommand{\firstplace}[1]{\cellcolor{topone}#1}
\newcommand{\secondplace}[1]{\cellcolor{toptwo}#1}
\newcommand{\thirdplace}[1]{\cellcolor{topthree}#1}
\begin{table*}[h]
\centering
\caption{Combined evaluation benchmarks for Stanford Volume Scans ($N=250,000$ cells $\rightarrow$ $37,500$ cells), Synthetic Volumetric Data ($N=100,000$ cells $\rightarrow$ $10,000$ cells), and Neper Polycrystal Microstructures ($N=100,000$ cells $\rightarrow$ $15,000$ cells). Each method has separate Final MSE and Trajectory AUC Mean MSE subcolumns, reported to 3 decimal places. Final MSE and AUC are ranked independently within each row; minima are in \textbf{bold}, and first, second, and third places are shaded blue (\colorbox{topone}{\textbf{1st}}, \colorbox{toptwo}{2nd}, and \colorbox{topthree}{3rd}).}
\label{tab:combined_benchmarks}
\setlength{\tabcolsep}{2.5pt}
\renewcommand{\arraystretch}{0.86}
\begin{adjustbox}{max totalsize={\textwidth}{0.78\textheight},center}
\begin{tabular}{ll*{8}{cc}}
\toprule
\multirow{2}{*}{Sampling} & \multirow{2}{*}{Scene} & \multicolumn{2}{c}{Neighbour} & \multicolumn{2}{c}{Volume} & \multicolumn{2}{c}{FPS} & \multicolumn{2}{c}{Grid} & \multicolumn{2}{c}{Cignoni et al.} & \multicolumn{2}{c}{Ward} & \multicolumn{2}{c}{Ours (Fast)} & \multicolumn{2}{c}{Ours (Exact)} \\
\cmidrule(lr){3-4} \cmidrule(lr){5-6} \cmidrule(lr){7-8} \cmidrule(lr){9-10} \cmidrule(lr){11-12} \cmidrule(lr){13-14} \cmidrule(lr){15-16} \cmidrule(lr){17-18}
& & MSE & AUC & MSE & AUC & MSE & AUC & MSE & AUC & MSE & AUC & MSE & AUC & MSE & AUC & MSE & AUC \\
\midrule
\multirow{3}{*}{uniform}
  & Head (CT) & 0.179 & 0.026 & 0.536 & 0.176 & 0.418 & 0.145 & 0.454 & 0.187 & \thirdplace{0.123} & \thirdplace{0.013} & 0.306 & 0.076 & \secondplace{0.067} & \secondplace{0.007} & \firstplace{\textbf{0.059}} & \firstplace{\textbf{0.006}} \\
  & Brain (MR) & 0.056 & 0.016 & 0.226 & 0.152 & 0.175 & 0.116 & 0.187 & 0.083 & \thirdplace{0.029} & \secondplace{0.004} & 0.138 & 0.076 & \secondplace{0.020} & \thirdplace{0.006} & \firstplace{\textbf{0.017}} & \firstplace{\textbf{0.002}} \\
  & bunny & 0.008 & 0.001 & 2.087 & 0.652 & 1.724 & 0.613 & 2.001 & 0.853 & \secondplace{0.003} & 0.001 & 6.706 & 6.459 & \secondplace{0.003} & 0.001 & \firstplace{\textbf{0.001}} & \firstplace{\textbf{0.000}} \\
\midrule
\multirow{3}{*}{importance}
  & Head (CT) & 0.501 & 0.098 & 1.363 & 0.270 & 0.351 & 0.108 & 0.416 & 0.147 & \thirdplace{0.220} & \thirdplace{0.036} & 0.440 & 0.144 & \secondplace{0.151} & \secondplace{0.024} & \firstplace{\textbf{0.145}} & \firstplace{\textbf{0.023}} \\
  & Brain (MR) & 0.161 & 0.146 & 0.959 & 0.298 & 0.160 & 0.148 & 0.182 & 0.067 & \thirdplace{0.092} & \secondplace{0.019} & 0.220 & 0.263 & \secondplace{0.054} & \thirdplace{0.036} & \firstplace{\textbf{0.050}} & \firstplace{\textbf{0.010}} \\
  & bunny & 0.001 & \firstplace{\textbf{0.000}} & 0.104 & 0.013 & 0.682 & 0.175 & 1.126 & 0.290 & 0.001 & 0.001 & 1.055 & 0.470 & 0.001 & \firstplace{\textbf{0.000}} & \firstplace{\textbf{0.000}} & \firstplace{\textbf{0.000}} \\
\midrule
\multirow{3}{*}{gradient}
  & Head (CT) & 0.770 & 0.197 & 1.003 & 0.311 & 0.361 & 0.121 & 0.439 & 0.160 & \thirdplace{0.306} & \thirdplace{0.081} & 0.439 & 0.151 & \secondplace{0.191} & \secondplace{0.052} & \firstplace{\textbf{0.188}} & \firstplace{\textbf{0.051}} \\
  & Brain (MR) & 0.210 & 0.115 & 0.763 & 0.547 & 0.173 & 0.172 & 0.193 & 0.082 & \thirdplace{0.113} & \secondplace{0.028} & 0.177 & 0.166 & \secondplace{0.074} & \thirdplace{0.047} & \firstplace{\textbf{0.070}} & \firstplace{\textbf{0.017}} \\
  & bunny & 8.896 & 1.884 & 2.265 & 0.861 & 2.441 & 0.931 & 2.668 & 1.066 & \thirdplace{1.262} & \thirdplace{0.292} & 5.941 & 1.743 & \secondplace{1.017} & \secondplace{0.217} & \firstplace{\textbf{0.986}} & \firstplace{\textbf{0.211}} \\
\midrule
\multirow{3}{*}{grid}
  & Head (CT) & 0.166 & 0.022 & 0.663 & 0.221 & 0.416 & 0.157 & 0.434 & 0.189 & \thirdplace{0.143} & \thirdplace{0.015} & 0.438 & 0.090 & \secondplace{0.128} & \secondplace{0.013} & \firstplace{\textbf{0.115}} & \firstplace{\textbf{0.011}} \\
  & Brain (MR) & 0.050 & 0.020 & 0.273 & 0.334 & 0.177 & 0.178 & 0.172 & 0.080 & \firstplace{\textbf{0.034}} & \firstplace{\textbf{0.004}} & 0.199 & 0.155 & \thirdplace{0.038} & \thirdplace{0.016} & \secondplace{0.036} & \firstplace{\textbf{0.004}} \\
  & bunny & 0.011 & 0.001 & 2.916 & 1.061 & 2.013 & 0.902 & 2.350 & 1.137 & \secondplace{0.004} & 0.001 & 5.136 & 4.211 & \firstplace{\textbf{0.003}} & \firstplace{\textbf{0.000}} & \secondplace{0.004} & \firstplace{\textbf{0.000}} \\
\midrule
\multicolumn{2}{l}{\textbf{Average}} & 0.917 & 0.210 & 1.096 & 0.408 & 0.758 & 0.314 & 0.885 & 0.362 & \thirdplace{0.194} & \thirdplace{0.041} & 1.766 & 1.167 & \secondplace{0.146} & \secondplace{0.035} & \firstplace{\textbf{0.139}} & \firstplace{\textbf{0.028}} \\
\bottomrule
\addlinespace[1.5ex]
\toprule
\multirow{2}{*}{Sampling} & \multirow{2}{*}{Scene} & \multicolumn{2}{c}{Neighbour} & \multicolumn{2}{c}{Volume} & \multicolumn{2}{c}{FPS} & \multicolumn{2}{c}{Grid} & \multicolumn{2}{c}{Cignoni et al.} & \multicolumn{2}{c}{Ward} & \multicolumn{2}{c}{Ours (Fast)} & \multicolumn{2}{c}{Ours (Exact)} \\
\cmidrule(lr){3-4} \cmidrule(lr){5-6} \cmidrule(lr){7-8} \cmidrule(lr){9-10} \cmidrule(lr){11-12} \cmidrule(lr){13-14} \cmidrule(lr){15-16} \cmidrule(lr){17-18}
& & MSE & AUC & MSE & AUC & MSE & AUC & MSE & AUC & MSE & AUC & MSE & AUC & MSE & AUC & MSE & AUC \\
\midrule
\multirow{8}{*}{uniform}
  & Checkerboard & 39.779 & 13.804 & 31.772 & 10.226 & \thirdplace{25.675} & 9.555 & 29.533 & 11.896 & 32.064 & \thirdplace{7.811} & 44.759 & 23.411 & \secondplace{21.551} & \secondplace{5.956} & \firstplace{\textbf{20.731}} & \firstplace{\textbf{5.788}} \\
  & Fractal & 0.404 & 0.105 & 0.113 & 0.024 & \firstplace{\textbf{0.058}} & \thirdplace{0.017} & 0.086 & 0.027 & 0.070 & 0.019 & 0.065 & 0.023 & \firstplace{\textbf{0.058}} & \firstplace{\textbf{0.014}} & \firstplace{\textbf{0.058}} & \firstplace{\textbf{0.014}} \\
  & Gyroid smooth & 6.487 & 1.893 & 1.897 & 0.392 & \firstplace{\textbf{0.936}} & \thirdplace{0.275} & 1.435 & 0.440 & 1.237 & 0.333 & 1.208 & 0.394 & \secondplace{1.131} & \firstplace{\textbf{0.266}} & \thirdplace{1.139} & \secondplace{0.268} \\
  & Gyroid solid & 16.985 & 1.906 & 13.054 & 3.731 & 9.123 & 3.272 & 11.483 & 4.273 & \thirdplace{8.459} & \thirdplace{0.981} & 39.759 & 19.359 & \secondplace{7.759} & \secondplace{0.767} & \firstplace{\textbf{6.467}} & \firstplace{\textbf{0.666}} \\
  & Sine high & 8.809 & 3.478 & 4.410 & 1.147 & \thirdplace{2.901} & \thirdplace{0.905} & 3.897 & 1.304 & 3.582 & 0.968 & 4.746 & 1.416 & \secondplace{2.492} & \secondplace{0.732} & \firstplace{\textbf{2.408}} & \firstplace{\textbf{0.723}} \\
  & Sine low & 0.750 & 0.263 & 0.186 & 0.037 & \thirdplace{0.086} & \thirdplace{0.025} & 0.137 & 0.041 & 0.122 & 0.032 & 0.091 & 0.031 & \firstplace{\textbf{0.074}} & \firstplace{\textbf{0.019}} & \secondplace{0.075} & \firstplace{\textbf{0.019}} \\
  & Sphere & \firstplace{\textbf{0.000}} & \firstplace{\textbf{0.000}} & 3.094 & 0.830 & 1.817 & 0.663 & 2.454 & 0.856 & \firstplace{\textbf{0.000}} & \firstplace{\textbf{0.000}} & 7.802 & 4.007 & \firstplace{\textbf{0.000}} & \firstplace{\textbf{0.000}} & \firstplace{\textbf{0.000}} & \firstplace{\textbf{0.000}} \\
  & Torus & \firstplace{\textbf{0.000}} & \firstplace{\textbf{0.000}} & 2.386 & 0.651 & 1.431 & 0.514 & 2.055 & 0.666 & \firstplace{\textbf{0.000}} & \firstplace{\textbf{0.000}} & 5.961 & 3.451 & \firstplace{\textbf{0.000}} & \firstplace{\textbf{0.000}} & \firstplace{\textbf{0.000}} & \firstplace{\textbf{0.000}} \\
\midrule
\multirow{3}{*}{gaussian}
  & Fractal & 1.098 & 0.442 & 0.227 & 0.032 & \thirdplace{0.041} & \thirdplace{0.008} & 0.063 & 0.012 & 0.076 & 0.009 & 0.055 & 0.014 & \firstplace{\textbf{0.030}} & \firstplace{\textbf{0.004}} & \firstplace{\textbf{0.030}} & \firstplace{\textbf{0.004}} \\
  & Sine low & 4.100 & 1.114 & 0.259 & 0.044 & \thirdplace{0.064} & \thirdplace{0.012} & 0.102 & 0.020 & 0.170 & 0.021 & 0.088 & 0.021 & \firstplace{\textbf{0.044}} & \firstplace{\textbf{0.006}} & \firstplace{\textbf{0.044}} & \firstplace{\textbf{0.006}} \\
  & Torus & 0.018 & 0.001 & 6.719 & 1.140 & 1.453 & 0.432 & 2.162 & 0.551 & \thirdplace{0.012} & 0.001 & 5.785 & 3.564 & \firstplace{\textbf{0.008}} & \firstplace{\textbf{0.000}} & \secondplace{0.011} & 0.001 \\
\midrule
\multirow{3}{*}{clustered}
  & Fractal & 1.034 & 0.462 & 0.105 & 0.018 & \thirdplace{0.026} & \thirdplace{0.005} & 0.043 & 0.008 & 0.048 & 0.006 & 0.043 & 0.010 & \firstplace{\textbf{0.019}} & \firstplace{\textbf{0.003}} & \firstplace{\textbf{0.019}} & \firstplace{\textbf{0.003}} \\
  & Sine low & 4.039 & 1.808 & 0.214 & 0.032 & \thirdplace{0.040} & \thirdplace{0.007} & 0.071 & 0.011 & 0.084 & 0.012 & 0.067 & 0.017 & \firstplace{\textbf{0.026}} & \firstplace{\textbf{0.004}} & \firstplace{\textbf{0.026}} & \firstplace{\textbf{0.004}} \\
  & Torus & \firstplace{\textbf{0.000}} & \firstplace{\textbf{0.000}} & 2.242 & 0.443 & 1.247 & 0.294 & 1.264 & 0.346 & \firstplace{\textbf{0.000}} & \firstplace{\textbf{0.000}} & 20.160 & 7.564 & \firstplace{\textbf{0.000}} & \firstplace{\textbf{0.000}} & \firstplace{\textbf{0.000}} & \firstplace{\textbf{0.000}} \\
\midrule
\multirow{3}{*}{exponential}
  & Fractal & 0.744 & 0.222 & 0.294 & 0.042 & \thirdplace{0.056} & 0.014 & 0.084 & 0.021 & 0.080 & 0.014 & 0.064 & 0.020 & \firstplace{\textbf{0.051}} & \firstplace{\textbf{0.010}} & \firstplace{\textbf{0.051}} & \firstplace{\textbf{0.010}} \\
  & Sine low & 1.756 & 0.420 & 0.701 & 0.095 & \thirdplace{0.083} & \thirdplace{0.020} & 0.129 & 0.032 & 0.144 & 0.025 & 0.094 & 0.028 & \secondplace{0.068} & \firstplace{\textbf{0.013}} & \firstplace{\textbf{0.067}} & \firstplace{\textbf{0.013}} \\
  & Torus & \firstplace{\textbf{0.000}} & \firstplace{\textbf{0.000}} & 3.014 & 0.584 & 1.319 & 0.398 & 1.953 & 0.518 & \firstplace{\textbf{0.000}} & \firstplace{\textbf{0.000}} & 5.428 & 2.310 & \firstplace{\textbf{0.000}} & \firstplace{\textbf{0.000}} & \firstplace{\textbf{0.000}} & \firstplace{\textbf{0.000}} \\
\midrule
\multicolumn{2}{l}{\textbf{Average}} & 5.059 & 1.525 & 4.158 & 1.145 & 2.727 & 0.966 & 3.350 & 1.236 & \thirdplace{2.715} & \thirdplace{0.602} & 8.010 & 3.861 & \secondplace{1.960} & \secondplace{0.459} & \firstplace{\textbf{1.831}} & \firstplace{\textbf{0.442}} \\
\bottomrule
\addlinespace[1.5ex]
\toprule
\multicolumn{2}{l}{\multirow{2}{*}{Microstructure Morphology}} & \multicolumn{2}{c}{Neighbour} & \multicolumn{2}{c}{Volume} & \multicolumn{2}{c}{FPS} & \multicolumn{2}{c}{Grid} & \multicolumn{2}{c}{Cignoni et al.} & \multicolumn{2}{c}{Ward} & \multicolumn{2}{c}{Ours (Fast)} & \multicolumn{2}{c}{Ours (Exact)} \\
\cmidrule(lr){3-4} \cmidrule(lr){5-6} \cmidrule(lr){7-8} \cmidrule(lr){9-10} \cmidrule(lr){11-12} \cmidrule(lr){13-14} \cmidrule(lr){15-16} \cmidrule(lr){17-18}
& & MSE & AUC & MSE & AUC & MSE & AUC & MSE & AUC & MSE & AUC & MSE & AUC & MSE & AUC & MSE & AUC \\
\midrule
\multicolumn{2}{l}{Columnar Additive (Elongated)} & 0.527 & 0.243 & \thirdplace{0.406} & \thirdplace{0.173} & 0.425 & 0.190 & 0.430 & 0.203 & 0.534 & 0.178 & 0.507 & 0.267 & \secondplace{0.350} & \secondplace{0.133} & \firstplace{\textbf{0.299}} & \firstplace{\textbf{0.123}} \\
\multicolumn{2}{l}{Dual-Phase Steel (Bimodal)} & 3.139 & 1.430 & \thirdplace{2.502} & 1.433 & 2.562 & 1.122 & 2.727 & 1.217 & 3.100 & \thirdplace{1.083} & 2.851 & 1.393 & \secondplace{2.011} & \secondplace{0.770} & \firstplace{\textbf{1.812}} & \firstplace{\textbf{0.740}} \\
\multicolumn{2}{l}{Equiaxed FCC (Isotropic)} & 1.685 & 0.653 & \thirdplace{1.084} & \thirdplace{0.460} & 1.136 & 0.503 & 1.150 & 0.541 & 1.836 & 0.482 & 2.069 & 0.794 & \secondplace{0.973} & \secondplace{0.312} & \firstplace{\textbf{0.776}} & \firstplace{\textbf{0.270}} \\
\multicolumn{2}{l}{Hexagonal Titanium (HCP)} & 12.773 & 5.086 & \thirdplace{8.263} & \thirdplace{3.555} & 8.781 & 3.911 & 8.835 & 4.159 & 14.626 & 3.978 & 14.307 & 6.183 & \secondplace{7.814} & \secondplace{2.530} & \firstplace{\textbf{6.231}} & \firstplace{\textbf{2.157}} \\
\multicolumn{2}{l}{Rolled Sheet (Anisotropic)} & 1.461 & 0.670 & \thirdplace{1.105} & \thirdplace{0.473} & 1.172 & 0.521 & 1.172 & 0.556 & 1.471 & 0.487 & 1.478 & 0.743 & \secondplace{0.968} & \secondplace{0.363} & \firstplace{\textbf{0.820}} & \firstplace{\textbf{0.333}} \\
\midrule
\multicolumn{2}{l}{\textbf{Average}} & 3.917 & 1.616 & \thirdplace{2.672} & \thirdplace{1.219} & 2.815 & 1.249 & 2.863 & 1.335 & 4.313 & 1.241 & 4.242 & 1.876 & \secondplace{2.423} & \secondplace{0.821} & \firstplace{\textbf{1.988}} & \firstplace{\textbf{0.725}} \\
\bottomrule
\end{tabular}
\end{adjustbox}
\end{table*}

\subsection{Scalability}
To evaluate the scalability of our algorithm in practice, we run decimation on the MRBrain scene sampled uniformly using 10k, 100k and 1 million input sites at $\beta=0.02$. The results in Table \ref{tab:combined_runtime_quality} show that runtimes scale roughly linearly with input cell size due to the sparse cost-recomputation. We find that for larger scenes, our fast cost computation performs is roughly $10\times$ faster than exact computation, resulting in up to $5\times$ overall faster reduction times, with only a minor increase in error metrics. %Note that the exact method is parallelized across 24 CPU cores and single-thread performance for the cost computation is roughly $20\times$ slower. 
More optimized implementations for the Delaunay triangulation may increase this difference even further. Note also that $\beta$ was chosen conservatively for this test to match the evaluation. Higher values of $\beta$ accelerate both Ours (Exact) and Ours (Fast) by a factor of up to $5\times$ with minor quality impact (cf. Supplemental \ref*{app:B_ablation}), especially benefitting the fast variant.
The Spearman correlation on the initial edge ordering validates the effectiveness of the volume approximations described in Section~\ref{sec:approximation}, showing only minimal difference in ranking compared to exact computation. In comparison, the \emph{Ward} baseline achieves poor correlation and significantly higher error, since it does not account for changes to neighbour cells. 

\begin{table*}[htbp]
\centering
\caption{Runtime, ranking agreement, and output quality at 90\% reduction on MRBrain (uniform). Time (total) represents total runtime without evaluation, while Time (cost) only counts the cumulative runtime of cost computations. $\rho$ is Spearman's rank correlation with the exact initial edge ordering. Experiments were run on a rig with RTX 4090 and AMD Ryzen 9 3900X.}
\label{tab:combined_runtime_quality}
\resizebox{\textwidth}{!}{%
\begin{tabular}{lrrrrrrrrrrrrrr}
\toprule
& \multicolumn{5}{c}{Ward}
& \multicolumn{5}{c}{Ours (Fast)}
& \multicolumn{4}{c}{Ours (Exact)} \\
\cmidrule(lr){2-6}\cmidrule(lr){7-11}\cmidrule(lr){12-15}

Input & {Time (total)} & {Time (cost)} & {Final MSE} & {AUC} & {$\rho$} & {Time (total)} & {Time (cost)} & {Final MSE} & {AUC} & {$\rho$} & {Time (total)} & {Time (cost)} & {Final MSE} & {AUC} \\
 & {(s)} & {(s)} &  &  &  & {(s)} & {(s)} &  &  &  & {(s)} & {(s)} &  &  \\
\midrule
10k & 3.9 & 0.1 & 0.707539 & 0.242835 & 0.657 & 7.7 & 4.0 & 0.125878 & 0.014627 & 0.948 & 17.6 & 16.0 & 0.105642 & 0.011521 \\
100k & 20.3 & 0.3 & 0.235609 & 0.090463 & 0.735 & 50.2 & 33.8 & 0.032492 & 0.003449 & 0.967 & 284.3 & 268.8 & 0.027663 & 0.002944 \\
1M & 340.7 & 1.2 & 0.065323 & 0.022660 & 0.779 & 428.2 & 140.8 & 0.008805 & 0.001105 & 0.970 & 1586.8 & 1516.0 & 0.008279 & 0.001060 \\
\bottomrule
\end{tabular}%
}
\end{table*}
\section{Discussion and Outlook}
In this paper, we introduced VoroDuce, the first GPU-accelerated, volume-aware reduction algorithm designed specifically for 3D Voronoi-based geometry representations. By treating the decimation process as a spatially-aware optimization problem, our method preserves structures and details of models more faithfully than traditional mesh or point-cloud downsampling techniques.
\paragraph*{Algorithmic Trade-offs and Approximations}
A central point of discussion is the deliberate trade-off between computational scalability and exact geometric precision for the fast variant. To avoid the computational bottleneck of exactly computing thousands of convex cell intersections per point removal, we implemented a novel approximation strategy that models intersection volumes as truncated prismatoids. While this formulation introduces minor inaccuracies in specific edge cases, such as when faces have disparate vertex counts or when intersections must be approximated as pyramids, it also enables massive parallelization on the GPU. Our evaluations on synthetic data, volume scans and Material Science simulations demonstrate that these localized approximations do not cascade into macroscopic structural failures, allowing us to maintain high accuracy.  

\paragraph*{Limitations}
VoroDuce always places the merged site at the midpoint of the collapsed Delaunay edge: Unlike quadric-based surface simplification, we do not optimize the position of the new site. Experiments on optimizing the site position (cf. Supplemental \ref*{app:merge}) have failed to produce statistically significant improvements to quality.

The algorithm assumes sites in general position. Near-degenerate inputs, such as the jittered regular lattices in our benchmark, produce unstable triangulations and many flips. Further, the algorithm is greedy and provides no guarantee of optimality.

\paragraph*{Future Work} 
While VoroDuce provides a robust framework for Voronoi reduction, several exciting avenues remain for future exploration:
\begin{enumerate}
    \item \textbf{Other Distance Metrics}: An interesting future direction would be generalizing the reduction algorithm to other distance metrics such as power or Manhattan distance.

    \item \textbf{GPU Optimization}: Enabling massive parallelization of the edge cost computation on the GPU accelerates the computations enormously, but further low-level optimization like kernel fusing and more optimized use of shared memory and registers could push the performance even further. 
    
    \item \textbf{Finite volume simulations}: Many finite volume methods in Astronomy~\cite{AstronomySim} or Fluid Dynamics~\cite{TESS} are based on 3D Voronoi meshes. Efficient Voronoi reduction could be used to accelerate these methods.

   % \item \textbf{Dynamic Level-of-Detail}: Our reduction algorithm could be used as the backbone for a full dynamic level-of-detail rendering system for radiance fields using Radiant Foam \cite{radfoam}. 
\end{enumerate}

%-------------------------------------------------------------------------
% bibtex
\bibliographystyle{eg-alpha-doi}
\pagebreak
\bibliography{egbibsample}       

\newcommand{\etalchar}[1]{$^{#1}$}
\begin{thebibliography}{\uppercase{MWLB25}}

\bibitem[ABE{\etalchar{*}}20]{vorocrust}
\textsc{Abdelkader A., Bajaj C.~L., Ebeida M.~S., Mahmoud A.~H., Mitchell S.~A., Owens J.~D., Rushdi A.~A.}:
\newblock {VoroCrust}: {Voronoi} meshing without clipping.
\newblock \emph{ACM Transactions on Graphics 39}, 3 (2020), 1--16.

\bibitem[AP18]{GaussianMixture3}
\textsc{Assa A., Plataniotis K.~N.}:
\newblock {Wasserstein}-distance-based {Gaussian} mixture reduction.
\newblock \emph{IEEE Signal Processing Letters 25}, 10 (2018), 1465--1469.

\bibitem[CCM{\etalchar{*}}00]{tetmesh2}
\textsc{Cignoni P., Costanza D., Montani C., Rocchini C., Scopigno R.}:
\newblock Simplification of tetrahedral meshes with accurate error evaluation.
\newblock In \emph{IEEE Visualization} (2000), IEEE, pp.~85--92.

\bibitem[CM02]{tetmesh4}
\textsc{Chopra P., Meyer J.}:
\newblock {TetFusion}: an algorithm for rapid tetrahedral mesh simplification.
\newblock In \emph{IEEE Visualization} (2002), IEEE, pp.~133--140.

\bibitem[Coo86]{PoissonDisk}
\textsc{Cook R.~L.}:
\newblock Stochastic sampling in computer graphics.
\newblock \emph{ACM Transactions on Graphics 5}, 1 (1986), 51--72.

\bibitem[DDI06]{AdaptiveThinning2}
\textsc{Demaret L., Dyn N., Iske A.}:
\newblock Image compression by linear splines over adaptive triangulations.
\newblock \emph{Signal Processing 86}, 7 (2006), 1604--1616.

\bibitem[DFG99]{CVT}
\textsc{Du Q., Faber V., Gunzburger M.}:
\newblock Centroidal {Voronoi} tessellations: Applications and algorithms.
\newblock \emph{SIAM Review 41}, 4 (1999), 637--676.

\bibitem[DFI02]{AdaptiveThinning1}
\textsc{Dyn N., Floater M.~S., Iske A.}:
\newblock Adaptive thinning for bivariate scattered data.
\newblock \emph{Journal of Computational and Applied Mathematics 145}, 2 (2002), 505--517.

\bibitem[DLA19]{SNet}
\textsc{Dovrat O., Lang I., Avidan S.}:
\newblock Learning to sample.
\newblock In \emph{IEEE/CVF Conference on Computer Vision and Pattern Recognition} (2019), IEEE, pp.~2755--2764.

\bibitem[DM11]{TESS}
\textsc{Duffell P.~C., MacFadyen A.~I.}:
\newblock Tess: a relativistic hydrodynamics code on a moving voronoi mesh.
\newblock \emph{The Astrophysical Journal Supplement Series 197}, 2 (2011), 15.

\bibitem[FC89]{MonteCarlo2}
\textsc{Fok D.~S., Crevier D.}:
\newblock Volume estimation by {Monte Carlo} methods.
\newblock \emph{Journal of Statistical Computation and Simulation 31}, 4 (1989), 223--235.

\bibitem[FSBP14]{Material2}
\textsc{Falco S., Siegkas P., Barbieri E., Petrinic N.}:
\newblock A new method for the generation of arbitrarily shaped {3D} random polycrystalline domains.
\newblock \emph{Computational Mechanics 54}, 6 (2014), 1447--1460.

\bibitem[GH97]{QEM}
\textsc{Garland M., Heckbert P.~S.}:
\newblock Surface simplification using quadric error metrics.
\newblock In \emph{SIGGRAPH Conference Proceedings} (1997), pp.~209--216.

\bibitem[GM15]{MonteCarlo1}
\textsc{Ge C., Ma F.}:
\newblock A fast and practical method to estimate volumes of convex polytopes.
\newblock In \emph{International Workshop on Frontiers in Algorithmics} (2015), Springer, pp.~52--65.

\bibitem[Gon85]{FPS}
\textsc{Gonzalez T.~F.}:
\newblock Clustering to minimize the maximum intercluster distance.
\newblock \emph{Theoretical Computer Science 38} (1985), 293--306.

\bibitem[GRV{\etalchar{*}}26]{powerfoam}
\textsc{Govindarajan S., Rebain D., Verbin D., Yi K.~M., Prabhu A., Tagliasacchi A.}:
\newblock {Power Foam}: Unifying real-time differentiable ray tracing and rasterization.
\newblock \emph{arXiv preprint arXiv:2604.24994} (2026).

\bibitem[GRYT25]{radfoam}
\textsc{Govindarajan S., Rebain D., Yi K.~M., Tagliasacchi A.}:
\newblock {Radiant Foam}: Real-time differentiable ray tracing.
\newblock In \emph{IEEE/CVF International Conference on Computer Vision} (2025), pp.~4135--4145.

\bibitem[GZZ21]{PointCloud3}
\textsc{Gong M., Zhang Z., Zeng D.}:
\newblock A new simplification algorithm for scattered point clouds with feature preservation.
\newblock \emph{Symmetry 13}, 3 (2021), 399.

\bibitem[Hop96]{Hoppe1996}
\textsc{Hoppe H.}:
\newblock Progressive meshes.
\newblock In \emph{Proceedings of the 23rd Annual Conference on Computer Graphics and Interactive Techniques} (New York, NY, USA, 1996), SIGGRAPH '96, Association for Computing Machinery, pp.~99--108.

\bibitem[KK94]{Material1}
\textsc{Kumar S., Kurtz S.~K.}:
\newblock Simulation of material microstructure using a {3D} {Voronoi} tesselation: Calculation of effective thermal expansion coefficient of polycrystalline materials.
\newblock \emph{Acta Metallurgica et Materialia 42}, 12 (1994), 3917--3927.

\bibitem[Lev00]{stanford_volume}
\textsc{Levoy M.}:
\newblock The {Stanford} volume data archive.
\newblock \url{https://graphics.stanford.edu/data/voldata/}, 2000.
\newblock Stanford Computer Graphics Laboratory. Accessed 2026-09-30.

\bibitem[Llo82]{Lloyd}
\textsc{Lloyd S.}:
\newblock Least squares quantization in {PCM}.
\newblock \emph{IEEE Transactions on Information Theory 28}, 2 (1982), 129--137.

\bibitem[LMA20]{SampleNet}
\textsc{Lang I., Manor A., Avidan S.}:
\newblock {SampleNet}: Differentiable point cloud sampling.
\newblock In \emph{IEEE/CVF Conference on Computer Vision and Pattern Recognition} (2020), IEEE, pp.~7575--7585.

\bibitem[LT98]{Lindstrom1998}
\textsc{Lindstrom P., Turk G.}:
\newblock Fast and memory efficient polygonal simplification.
\newblock In \emph{Proceedings Visualization'98 (Cat. No. 98CB36276)} (1998), IEEE, pp.~279--286.

\bibitem[MOG11]{ShapeAnalysis}
\textsc{M{\'e}rigot Q., Ovsjanikov M., Guibas L.~J.}:
\newblock {Voronoi}-based curvature and feature estimation from point clouds.
\newblock \emph{IEEE Transactions on Visualization and Computer Graphics 17}, 6 (2011), 743--756.

\bibitem[MP78]{DualSpace}
\textsc{Muller D.~E., Preparata F.~P.}:
\newblock Finding the intersection of two convex polyhedra.
\newblock \emph{Theoretical Computer Science 7}, 2 (1978), 217--236.

\bibitem[MWLB25]{Blackbox}
\textsc{Maderbacher B., Windisch F., Larrauri A., Bloem R.}:
\newblock Synthesis of controllers for continuous blackbox systems.
\newblock In \emph{Verification, Model Checking, and Abstract Interpretation} (Cham, 2025), Shankaranarayanan K., Sankaranarayanan S., Trivedi A., (Eds.), Springer Nature Switzerland, pp.~137--159.

\bibitem[{\"O}Y19]{Architecture}
\textsc{{\"O}ncel A.~C., Yaman U.}:
\newblock Generation of optimized {Voronoi} based interior structures for improved mechanical properties.
\newblock \emph{Procedia Manufacturing 38} (2019), 42--51.

\bibitem[PGK02]{PointSimplification}
\textsc{Pauly M., Gross M., Kobbelt L.~P.}:
\newblock Efficient simplification of point-sampled surfaces.
\newblock In \emph{IEEE Visualization} (2002), IEEE, pp.~163--170.

\bibitem[QK22]{neper}
\textsc{Quey R., Kasemer M.}:
\newblock The {Neper}/{FEPX} project: free/open-source polycrystal generation, deformation simulation, and post-processing.
\newblock In \emph{IOP Conference Series: Materials Science and Engineering} (2022), vol.~1249, IOP Publishing, p.~012021.

\bibitem[QR18]{polycrystal}
\textsc{Quey R., Renversade L.}:
\newblock Optimal polyhedral description of {3D} polycrystals: Method and application to statistical and synchrotron {X-ray} diffraction data.
\newblock \emph{Computer Methods in Applied Mechanics and Engineering 330} (2018), 308--333.

\bibitem[Run07]{GaussianMixture2}
\textsc{Runnalls A.~R.}:
\newblock {Kullback-Leibler} approach to {Gaussian} mixture reduction.
\newblock \emph{IEEE Transactions on Aerospace and Electronic Systems 43}, 3 (2007), 989--999.

\bibitem[SG98]{tetmesh3}
\textsc{Staadt O.~G., Gross M.~H.}:
\newblock Progressive tetrahedralizations.
\newblock In \emph{IEEE Visualization} (1998), IEEE, pp.~397--402.

\bibitem[SH74]{SutherlandHodgman}
\textsc{Sutherland I.~E., Hodgman G.~W.}:
\newblock Reentrant polygon clipping.
\newblock \emph{Communications of the ACM 17}, 1 (1974), 32--42.

\bibitem[Spr10]{AstronomySim}
\textsc{Springel V.}:
\newblock E pur si muove: Galilean-invariant cosmological hydrodynamical simulations on a moving mesh.
\newblock \emph{Monthly Notices of the Royal Astronomical Society 401}, 2 (2010), 791--851.

\bibitem[SZW{\etalchar{*}}22]{PointCloud2}
\textsc{Sun Y., Zhang S., Wang T., Lou F., Ma J., Wang C., Gui C.}:
\newblock An improved spatial point cloud simplification algorithm.
\newblock \emph{Neural Computing and Applications 34}, 15 (2022).

\bibitem[TCC24]{BioMed}
\textsc{Todd L., Chin M.~H., Coppens M.-O.}:
\newblock Two conjectures on {3D} {Voronoi} structures: a toolkit with biomedical case studies.
\newblock \emph{Molecular Systems Design \& Engineering 9}, 9 (2024), 912--919.

\bibitem[THJ99]{tetmesh1}
\textsc{Trotts I.~J., Hamann B., Joy K.~I.}:
\newblock Simplification of tetrahedral meshes with error bounds.
\newblock \emph{IEEE Transactions on Visualization and Computer Graphics 5}, 3 (1999), 224--237.

\bibitem[vdWS09]{AstroPhysics}
\textsc{van~de Weygaert R., Schaap W.}:
\newblock The cosmic web: geometric analysis.
\newblock In \emph{Data Analysis in Cosmology}. Springer, 2009, pp.~291--413.

\bibitem[WJ63]{Ward}
\textsc{Ward~Jr J.~H.}:
\newblock Hierarchical grouping to optimize an objective function.
\newblock \emph{Journal of the American Statistical Association 58}, 301 (1963), 236--244.

\bibitem[WM03]{GaussianMixture1}
\textsc{Williams J.~L., Maybeck P.~S.}:
\newblock Cost-function-based {Gaussian} mixture reduction for target tracking.
\newblock In \emph{International Conference on Information Fusion} (2003), vol.~2, International Society of Information Fusion Cairns, Australia, pp.~1047--1054.

\bibitem[WZPB23]{APES}
\textsc{Wu C., Zheng J., Pfrommer J., Beyerer J.}:
\newblock Attention-based point cloud edge sampling.
\newblock In \emph{IEEE/CVF Conference on Computer Vision and Pattern Recognition} (2023), IEEE, pp.~5333--5343.

\bibitem[ZQW{\etalchar{*}}19]{PointCloud1}
\textsc{Zhang K., Qiao S., Wang X., Yang Y., Zhang Y.}:
\newblock Feature-preserved point cloud simplification based on natural quadric shape models.
\newblock \emph{Applied Sciences 9}, 10 (2019), 2130.

\bibitem[ZWBS17]{CollisionAvoidance}
\textsc{Zhou D., Wang Z., Bandyopadhyay S., Schwager M.}:
\newblock Fast, on-line collision avoidance for dynamic vehicles using buffered {Voronoi} cells.
\newblock \emph{IEEE Robotics and Automation Letters 2}, 2 (2017), 1047--1054.

\end{thebibliography}

% biblatex with biber
% \printbibliography                

%-------------------------------------------------------------------------
%Color tables are no longer required for purely electronic publications.
\newpage

\newcommand{\secondtitle}[1]{%
 \begin{center}
  \twocolumn[
    {\LARGE\bfseries #1\par}
    \vspace{1.5em}]
  \end{center}
}
%\setcounter{section}{0}
%\renewcommand{\thesection}{\Alph{section}}
%\renewcommand{\thesection}{\thesection.\arabic{section}}
%\clearpage
%\setcounter{page}{1}
\setcounter{section}{0}
\renewcommand{\thesection}{\Alph{section}}
\secondtitle{Supplemental Material for VoroDuce: Massively Parallel 3D Voronoi Simplification with Volume-Aware Edge Collapse}

\setcounter{theorem}{0}
\setcounter{prop}{0}

\section{Example Illustrations}
Figure~\ref{fig:edge_collapse_leads_to_flips} shows an example of Delaunay edge flips caused by a single Voronoi cell merge and Figure~\ref{fig:voronoi_cases_2D} gives a 2D example of how these edge flips lead to small errors in our approximation.
\label{app:example}
\begin{figure}[h]
    \centering
    \includegraphics[width=0.99\linewidth]{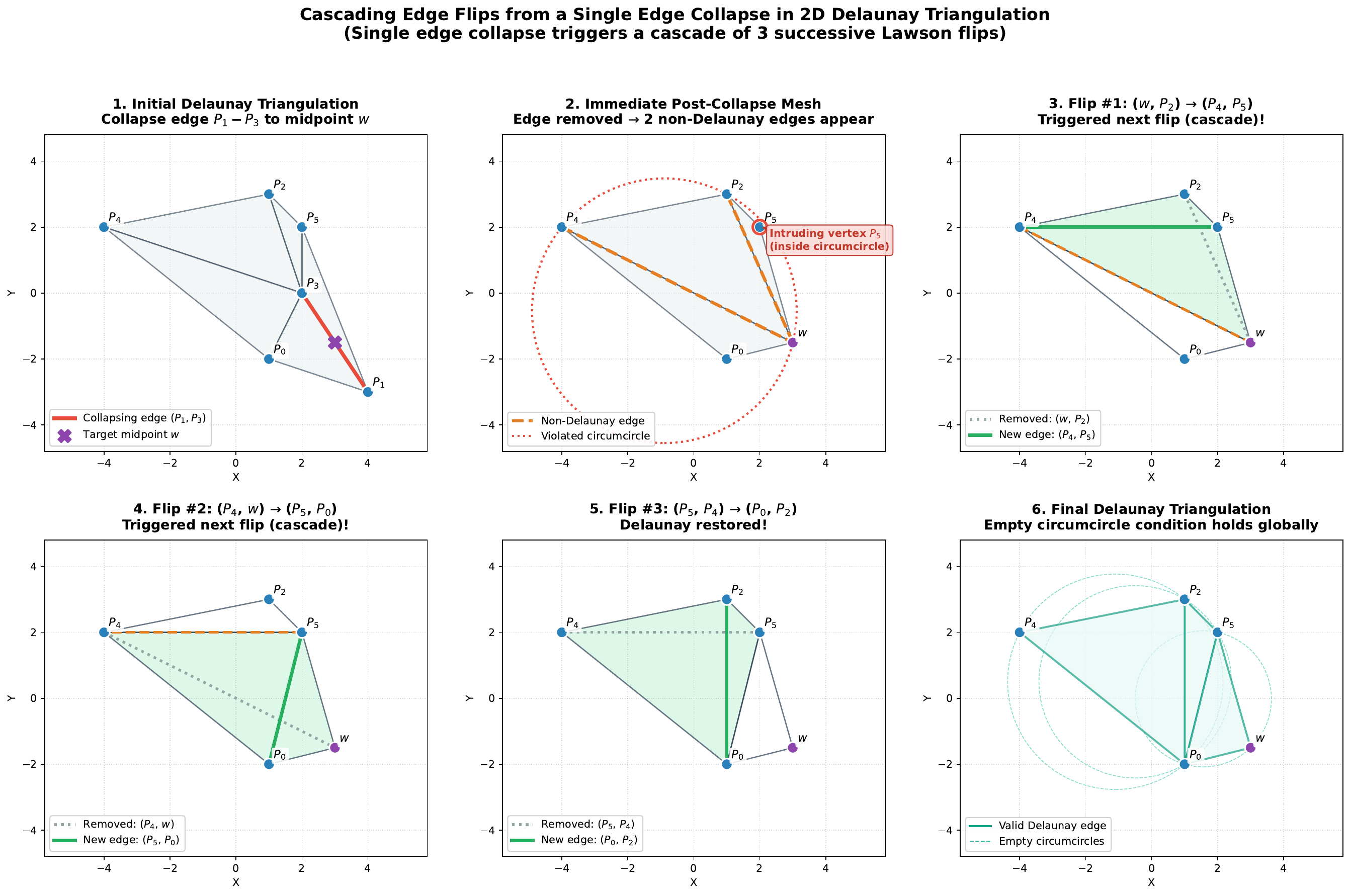}
    \caption{Example: A single edge collapse in a Delaunay Triangulation may cause several edge flips.}
    \label{fig:edge_collapse_leads_to_flips}
\end{figure}

\begin{figure*}[h]
    \centering
    \begin{subfigure}[b]{0.47\linewidth}
        \centering
        \includegraphics[width=0.8\linewidth]{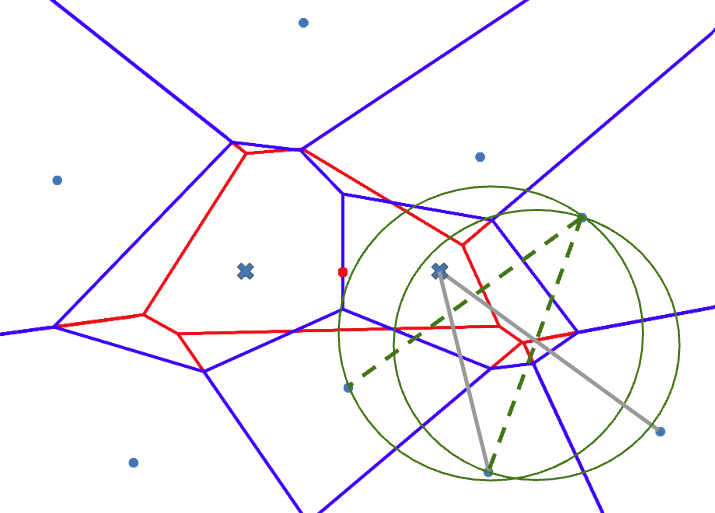}
        \caption{The collapse of the edge (red) makes supporting circles (green) for new edges (dotted, green) empty, resulting in two previous edges (gray) getting flipped. }
        \label{fig:edgeflips}
    \end{subfigure}
    \hspace{2mm}
    \begin{subfigure}[b]{0.47\linewidth}
        \centering
        \includegraphics[width=1\linewidth]{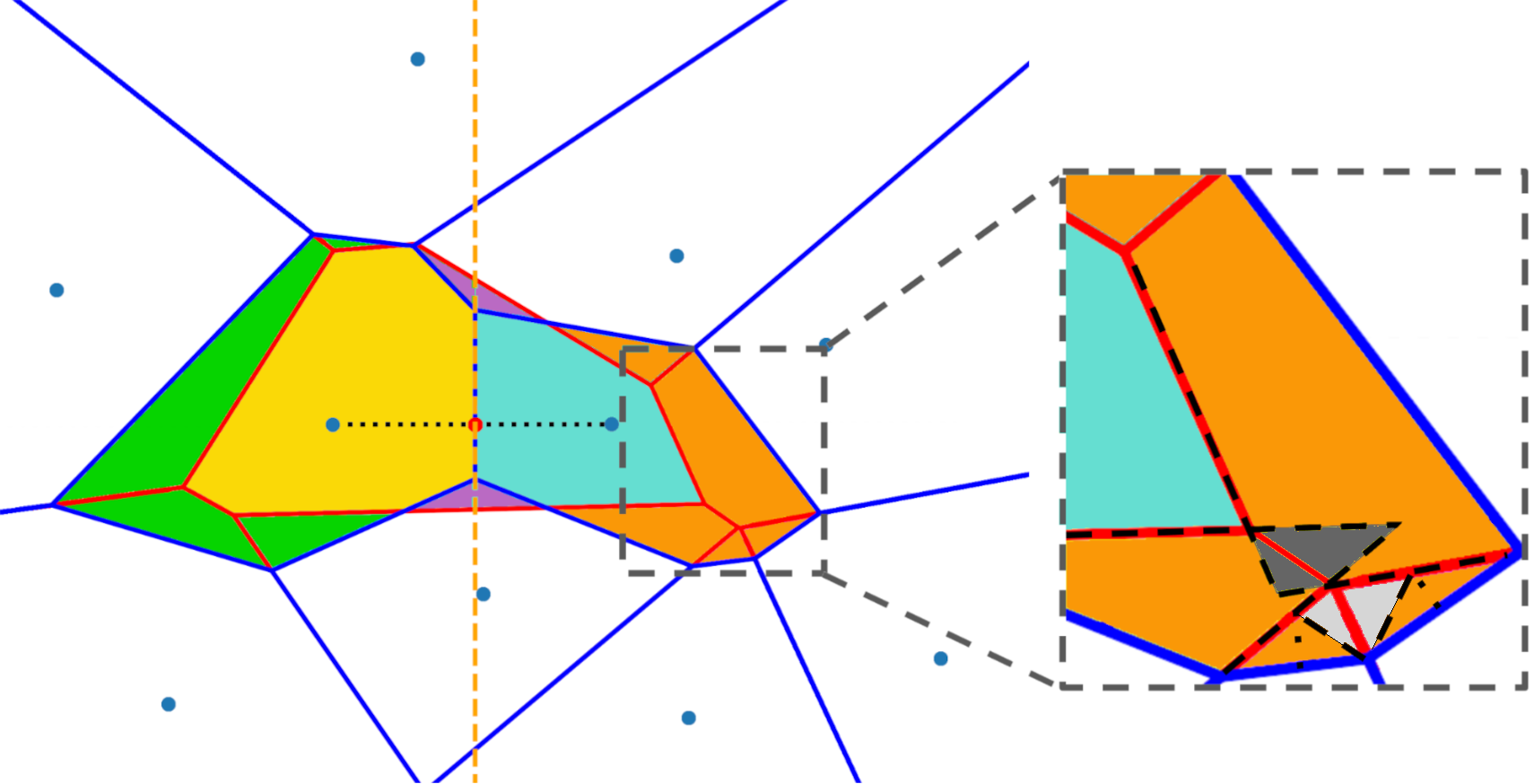}
        \caption{The neighbour edge flips result in non-standard cells. Still, the general convex cells can be approximated as quadrangles, double-counting the dark gray region and ignoring the light gray region. }
    \end{subfigure}

    \caption{Two cells (blue) are merged along a Delaunay edge (dotted) to a new cell (red). Areas where the nearest neighbour has changed are colored. Green intersection volumes form a quadrangle, cyan intersection volumes form two triangles and purple regions a single triangle. Yellow regions are convex cells with pinch vertices.}
    \label{fig:voronoi_cases_2D}
\end{figure*}

\section{Dataset Details}
In this section, we include additional details about the generation and collection for the fields used in the evaluation. 
\label{app:datasets}
\subsection{Synthetic Dataset}
Each dataset is generated via a decoupled two-stage pipeline: (1) sampling 3D seed coordinates $\mathbf{x}_i \in \Omega$ from a spatial density distribution, followed by (2) analytically mapping a volumetric scalar field to assign cell values $y_i = f(\mathbf{x}_i)$.
\paragraph*{Spatial Seed Distributions}
To benchmark algorithmic robustness against spatial anisotropy and non-uniform cell volume distributions, seed points $\mathbf{x}_i \in \Omega$ are sampled from five spatial patterns:
\begin{itemize}
    \item \textbf{Uniform:} Homogeneous spatial Poisson point process drawn uniformly from $\mathcal{U}([0, 10]^3)$.
    \item \textbf{Gaussian:} Centered radial point concentration sampled from $\mathcal{N}(\boldsymbol{\mu}, \sigma^2 \mathbf{I})$ with $\boldsymbol{\mu} = (5, 5, 5)^\top$ and $\sigma = 1.8$, clamped to $\Omega$.
    \item \textbf{Dual Clustered:} Bimodal Gaussian mixture with equal components centered at $\boldsymbol{\mu}_1 = (3, 3, 3)^\top$ and $\boldsymbol{\mu}_2 = (7, 7, 7)^\top$ ($\sigma = 1.2$), producing localized high-density clusters separated by sparse regions.
    \item \textbf{Exponential:} Severe unidirectional density gradient along the $x$-axis distributed as $p(x) \propto e^{-\lambda x}$ with decay rate $\lambda = 0.3$, combined with uniform transverse coordinates $y, z \sim \mathcal{U}([0, 10])$.
\end{itemize}
\paragraph*{Analytical Volumetric Fields}
Scalar values $y_i = f(\mathbf{x}_i)$ are evaluated directly at seed coordinates across eight representative functional categories:
\begin{itemize}
    \item \textbf{Harmonic Waves:} Low-frequency ($k = 2\pi/10$) and high-frequency ($k = 2\pi/1.5$) continuous sine fields:
    \begin{equation}
        f_{\text{sine}}(\mathbf{x}) = 5 \cdot \left(\sin(k x) \cos(k y) \sin(k z) + 1\right).
    \end{equation}
    \item \textbf{Multi-Scale Fractal:} A 4-octave harmonic field accumulating dyadically doubled spatial frequencies with halved amplitudes:
    \begin{equation}
       \tiny f_{\text{fractal}}(\mathbf{x}) = 5 \cdot \left( \frac{1}{\sum_{o=0}^3 2^{-o}} \sum_{o=0}^3 2^{-o} \sin(2^o \omega x + \phi_o) \cos(2^o \omega y + \psi_o) \sin(2^o \omega z + \xi_o) + 1 \right).
    \end{equation}
    \item \textbf{Minimal Surfaces (Gyroid):} Triply periodic minimal surfaces evaluated as either a smooth scalar field $f_{\text{gyroid}}(\mathbf{x})$ or a binary solid partition:
    \begin{equation}
        g(\mathbf{x}) = \sin(k x)\cos(k y) + \sin(k y)\cos(k z) + \sin(k z)\cos(k x), \quad k = \frac{\pi}{2},
    \end{equation}
    yielding $f_{\text{solid}}(\mathbf{x}) = 10 \cdot \mathbb{I}[g(\mathbf{x}) > 0]$ and $f_{\text{smooth}}(\mathbf{x}) = 5 \cdot (g(\mathbf{x})/1.5 + 1)$.
    \item \textbf{Sharp Step Geometries:} Piecewise-constant binary indicator functions $f(\mathbf{x}) \in \{0, 10\}$ modeling sharp material boundaries:
    \begin{itemize}
        \item \emph{Solid Sphere:} $f(\mathbf{x}) = 10 \cdot \mathbb{I}[\|\mathbf{x} - \mathbf{c}\|_2 \le 3.5]$ with center $\mathbf{c} = (5, 5, 5)^\top$.
        \item \emph{Solid Torus:} $f(\mathbf{x}) = 10 \cdot \mathbb{I}[(\sqrt{(x-c_x)^2 + (y-c_y)^2} - R)^2 + (z-c_z)^2 \le r^2]$ with major radius $R = 3.0$ and minor radius $r = 1.0$.
        \item \emph{3D Checkerboard:} $f(\mathbf{x}) = 10 \cdot \mathbb{I}[(\lfloor x \rfloor + \lfloor y \rfloor + \lfloor z \rfloor) \bmod 2 = 0]$.
    \end{itemize}
\end{itemize}

\subsection{Material Science Dataset}
Here we provide the exact neper commands used to generate the \emph{Material Science} dataset. %As we are no experts, we asked an LLM to construct commands to simulate realistic testing conditions. 
The value $\mathbf{y}$ assigned to each cell is the maximum resolved shear stress across active slip system.
\paragraph*{Datasets 1-3: Equiaxed FCC Polycrystal, Cold-Rolled Sheet Polycrystal and Additive Columnar Polycrystal}
Shear stress was evaluated across the 12 primary octahedral slip systems $\{111\}\langle 110\rangle$. Generation commands:
\begin{verbatim}
bash
<Path>/neper/src/build/neper -T -n 100000 \
  -statcell id,x,y,z,vol,diameq,sphericity,
  facenb,euler-bunge \
  -o /tmp/neper_work/neper_equiaxed_fcc_100k \
  -morpho voronoi \
  -crysym cubic \
  -ori random
\end{verbatim}
\begin{verbatim}
bash
<Path>/neper/src/build/neper -T -n 100000 \
  -statcell id,x,y,z,vol,diameq,sphericity,
  facenb,euler-bunge \
  -o /tmp/neper_work/neper_rolled_sheet_100k \
  -morpho voronoi \
  -transform "scale(2.0,1.0,0.5)" \
  -crysym cubic \
  -ori "fiber(1,1,0, 1,0,0, 15)"
\end{verbatim}

\begin{verbatim}
bash
<Path>/neper/src/build/neper -T -n 100000 \
  -statcell id,x,y,z,vol,diameq,sphericity,
  facenb,euler-bunge \
  -o /tmp/neper_work/neper_columnar_additive_100k \
  -morpho voronoi \
  -transform "scale(0.5,0.5,2.0)" \
  -crysym cubic \
  -ori "fiber(0,0,1, 0,0,1, 10)"
  \end{verbatim}
\paragraph*{Dataset 4: Dual-Phase Steel Polycrystal}
Shear stress was evaluated across the 3 Basal $\{0001\}\langle 11\bar{2}0\rangle$ and 3 Prismatic $\{10\bar{1}0\}\langle 11\bar{2}0\rangle$ slip systems. Generation command:
\begin{verbatim}
bash
<Path>/neper/src/build/neper -T -n 100000 \
  -statcell id,x,y,z,vol,diameq,sphericity,
  facenb,euler-bunge \
  -o /tmp/neper_work/neper_dualphase_steel_100k \
  -morpho voronoi \
  -crysym cubic \
  -ori "0.7*random+0.3*cube"
  \end{verbatim}
\paragraph*{Dataset 5: Hexagonal Close-Packed 
$\alpha$-Titanium}
Scaled with a bimodal Hall-Petch grain size hardening factor ($\times 1.35$ for fine grains below median volume, $\times 0.85$ for coarse grains). Generation command:
\begin{verbatim}
bash
<Path>/neper/src/build/neper -T -n 100000 \
  -statcell id,x,y,z,vol,diameq,sphericity,
  facenb,euler-bunge \
  -o /tmp/neper_work/neper_hexagonal_titanium_100k \
  -morpho voronoi \
  -crysym hexagonal \
  -ori "fiber(0,0,0,1, 0,0,1, 20)"
  \end{verbatim}

\section{Ablation of collapse batch size $\beta$}
\label{app:B_ablation}
 The merge loop (Algorithm 1) does not process one edge at a time, it groups them into blocks. After each Delaunay rebuild, all candidate edges are sorted by collapse cost and the $\lfloor \beta \cdot |\mathcal{S}|\rfloor$ of them are considered together, with a greedy independent-set matching picking the subset whose endpoints don't collide. As such, edge costs do not take into account merges that happened beforehand in the same block. Increasing $\beta$ also increases the number of these ``stale'' costs. Note that $\beta$ can be increased beyond $1$, as there are more Delaunay edges than vertices. To evaluate the sensitivity of this parameter, we run reductions for the volume scans dataset across a sweep of $\beta$ values. The resulting quality metrics and runtimes are visualized in Figure~\ref{fig:beta}. The results show that the reduction algorithm is generally insensitive to values up to $\beta=1$, which provides up to a $5\times$ speed-up over $\beta=0.01$ with minimal change in quality, while for $\beta> 1$, quality degrades quickly. Counterintuitively, the best quality is achieved at values of $\beta \in [0.1, 0.5]$. As such, $\beta$ cannot be interpreted simply as a trade-off between quality and performance. Instead it acts as a regularizer: Because each cell can only be merged once per block, a moderate value of $\beta$ prevents that a small set of cells is merged too aggressively. Excessive values of $\beta$ face the opposite problem, where good candidate cells are not merged often enough.

 Both our exact and fast variant gain significant speed-up with larger batch sizes, but the fast variant benefits more from the reduction in triangulation costs, achieving up to an $8\times$ overall speedup over exact costs.  
\begin{figure}
    \centering
    \includegraphics[width=0.9\linewidth]{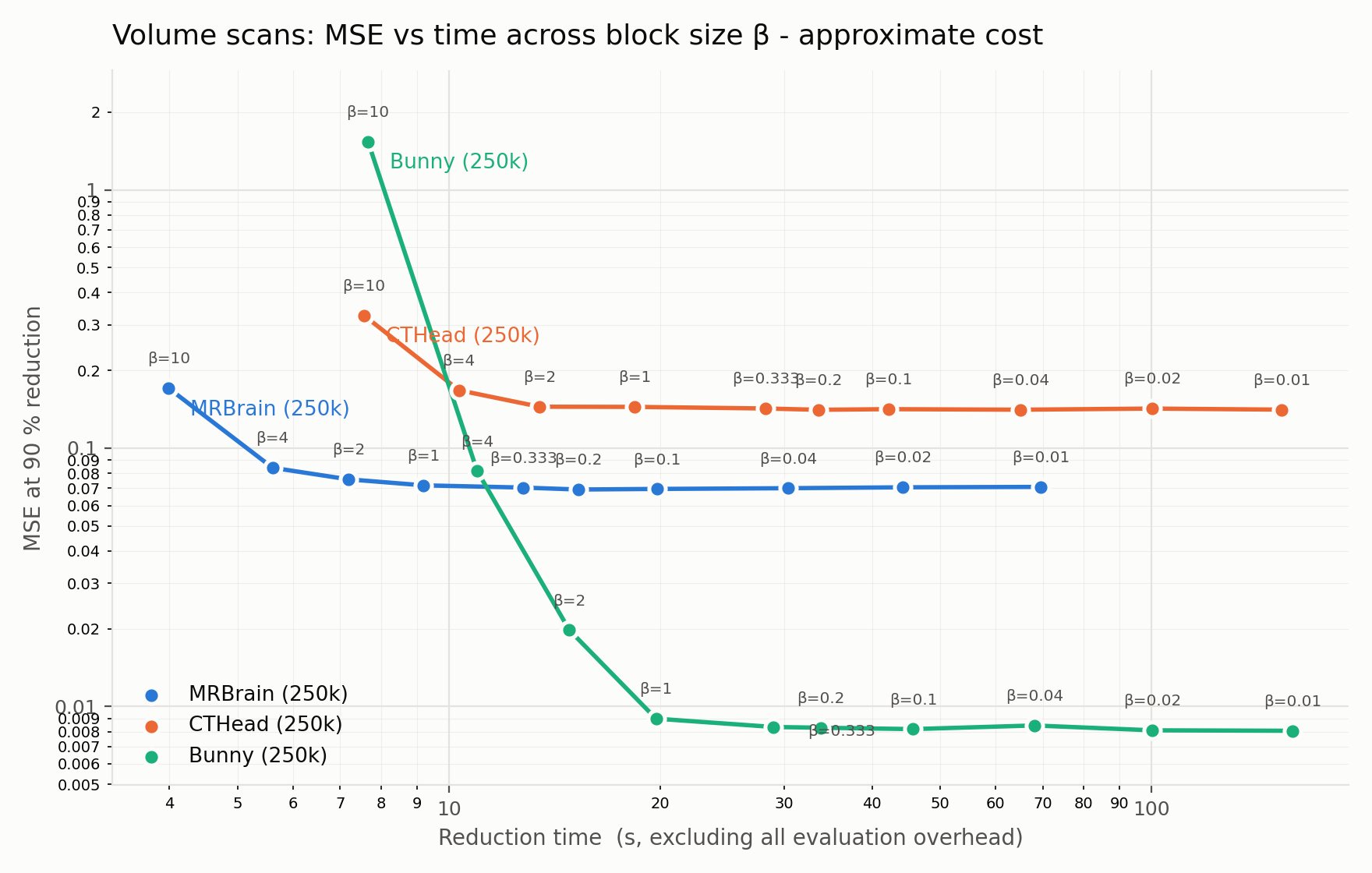}
    \includegraphics[width=0.9\linewidth]{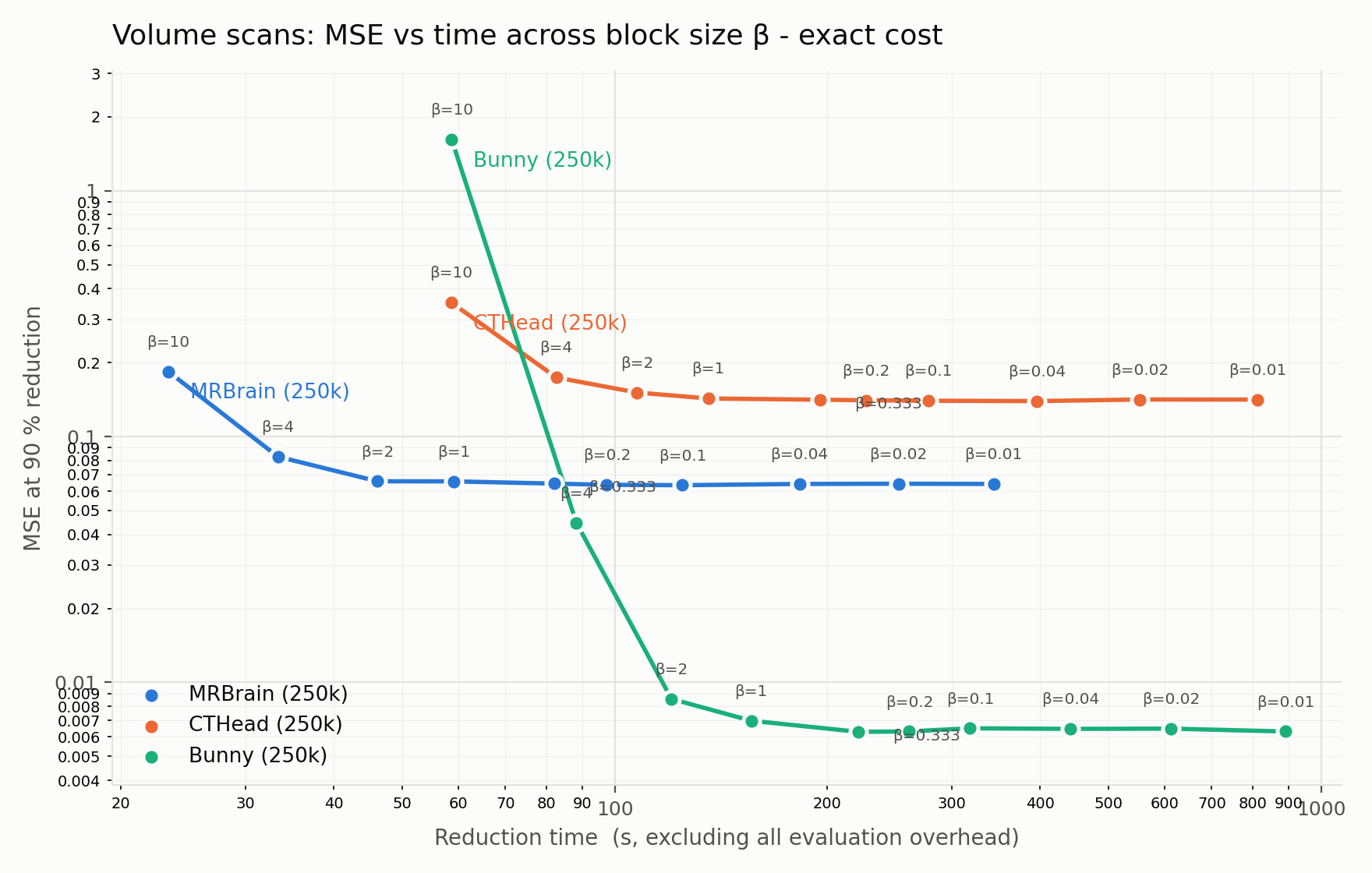}
    \caption{MSE and reduction time on the volume scan dataset with varying values of collapse batch size $\beta$.}
    \label{fig:beta}
\end{figure}
\section{Implementation Details}
We provide an optimized, hardware-accelerated implementation of our reduction algorithm using vectorized pytorch operations and CUDA kernels implemented in Warp. For computing the Delaunay triangulation, we use the GPU implementation by Govindarajan et al.~\cite{radfoam}, which attempts local re-triangulation with a fallback to full rebuild. 
\paragraph*{Infinite cells} The Voronoi cells that form the convex hull of the point set are infinite. Since merging infinite cells introduces infinite error, these points are excluded from the cost computation, except as neighbour cells. While it would not be difficult to implement clipping of all cells to a bounding box, we found that the number of hull vertices is too small in practice to have any effect on quality.% Otherwise no special handling is required for the algorithm to function.  

%\paragraph*{Exact volume intersection}
%Evaluating one Voronoi cell through scipy's QHull bindings costs 253 $µs$, of which 119 $µs$ is fixed per-call overhead; our kernel evaluates the same cell in 12.0 $µs$ on one core and 1.05 $µs$ across 24, while agreeing with QHull to 3·10⁻¹⁴ relative. For the cell-intersection volumes of the exact cost the gap widens to 66× single-threaded, since QHull additionally requires a linear program per intersection to obtain an interior point (940 µs of the 1067 $µs$ total).

\subsection{Merging Cells}
\label{app:merge}
When merging cells, we choose the midpoint of the collapsed Delaunay edge as the new site position, because it minimizes the error in our intersection volume approximations. From this, we can pick the new function value of the site (denoted as InterpolateValues($i, j$) in pseudocode) to minimize the resulting error. 
 Mathematically, we want to find $\mathbf{y}_{new}$ s.t. 
\begin{align}
\mathbf{y}_{new}^* &= \underset{\mathbf{y}_{new}}{\mathrm{argmin}} \ ISE(\mathcal{S}, \left(\mathcal{S} \setminus \{(P_i, \mathbf{y}_i), (P_j, \mathbf{y}_j)\} \right) \cup (\frac{P_i + P_j}{2}, \mathbf{y}_{new})) \\
   \underset{\mathbf{y}}{\mathrm{argmin}}& \ \sum_{k \in \mathcal{N}(i) \cup \mathcal{N}(j)} (\hat{\mathbf{y}}_M - \mathbf{y}_k)^2 \cdot \text{Vol}(\hat{V}_M \cap V_k)  \\
  &+ (\hat{\mathbf{y}}_M - \mathbf{y}_i)^2 \cdot \text{Vol}(\hat{V}_M \cap V_i)\\ 
  &+ (\hat{\mathbf{y}}_M - \mathbf{y}_j)^2 \cdot \text{Vol}(\hat{V}_M \cap V_j). 
\end{align} 
We set the derivative of the objective with regards to $\mathbf{y}$ to zero:
\begin{align}
   \frac{dISE}{d\mathbf{y}_{new}} =& \ \sum_{k \in (\mathcal{N}(i) \cup \mathcal{N}(j)) \setminus \{i, j\}) \setminus \{i, j\}} 2\cdot(\hat{\mathbf{y}}_M - \mathbf{y}_k) \cdot \text{Vol}(\hat{V}_M \cap V_k)  \\
  &+ 2\cdot(\hat{\mathbf{y}}_M - \mathbf{y}_i) \cdot \text{Vol}(\hat{V}_M \cap V_i)\\ 
  &+ 2\cdot(\hat{\mathbf{y}}_M - \mathbf{y}_j) \cdot \text{Vol}(\hat{V}_M \cap V_j) = 0. 
\end{align} 
Solving this for $\mathbf{y}_{new}$, we find that the optimal solution is exactly the mean of $f$ over the volume of $\hat{V}_M$:
\begin{align}
    \mathbf{y}_{new}^* &= \frac{1}{\text{Vol}(\hat{V}_M)} \cdot \Big( \sum_{k \in \mathcal{N}(i) \cup \mathcal{N}(j)}  \mathbf{y}_k \cdot \text{Vol}(\hat{V}_M \cap V_k)   \\
  &+  \mathbf{y}_i \cdot \text{Vol}(\hat{V}_M \cap V_i) + \mathbf{y}_j \cdot \text{Vol}(\hat{V}_M \cap V_j)\Big) \\&=\frac{1}{\text{Vol}(\hat{V}_M)} \cdot \int_{\mathbf{x} \in \hat{V}_m} f(\mathbf{x})dx
\end{align}

Note that this is also true for neighbour cells, where the mean value integrated over its new volume differs where it gains volume. In practice, we find no measurable benefit to this change and forego this optimization. \\

\paragraph*{Optimized site placement.} We ran the following experiment to evaluate if the site position can be improved from the midpoint placement: Instead of fixing the merged site at the edge midpoint, we place it at $\mathbf{x}_{new}(t) = \mathbf{x}_i + t(\mathbf{x}_j - \mathbf{x}_i)$ and optimize $t$ per edge. Since $\hat{\mathbf{y}}_{new}$ is the mean of $f$ over $\hat V_{new}$, the derivative of the cost $c(i,j)$ depends only on the motion of the faces $\mathcal{F}_{new}$ of $\hat V_{new}$:

$$\frac{\partial c(i,j)}{\partial \mathbf{x}_{new}} = \sum_{m \in \{i,j\}} \sum_{k \in \mathcal{N}(m)\setminus\{i,j\}} \sum_{s \in \{+,-\}} \Delta e_k^s \, \mathrm{Area}(\mathcal{F}_{new}^s) \, \frac{\mathbf{c}(\mathcal{F}_{new}^s) - \mathbf{x}_{new}}{\|\mathbf{x}_k - \mathbf{x}_{new}\|}$$

Here $\mathbf{c}(\cdot)$ is the area centroid, $\Delta e_k^+ = (\hat{\mathbf{y}}_{new} - \mathbf{y}_m)^2 - (\mathbf{y}_k - \mathbf{y}_m)^2$ and $\Delta e_k^- = (\hat{\mathbf{y}}_{new} - \mathbf{y}_k)^2$. The faces $\mathcal{F}_{new}^\pm$ are already computed in Algorithm 2, so the gradient adds little overhead.

Starting from $t = \tfrac12$, we take one secant step on $\partial c / \partial t$ and clamp $t \in [0.25, 0.75]$. The new position is accepted only if it lowers $c(i,j)$. The optimized cost is used both for edge ranking and for the collapse itself. In practice, we did not observe any significant quality improvement and even some regression in test scenes when using optimized site placement, compared to midpoint placement.

\subsection{Computing Cell Volumes From Delaunay Tetrahedra}
\label{App:Volume}
The previously introduced cost function requires knowing the original cell volumes of each Voronoi cell. Instead of computing the volumes of all cells individually, it is more efficient to instead compute the volumes of all Delaunay tetrahedra and then redistribute that volume to its four vertex cells.
This pass can then be re-run after each triangulation re-build and is fast, as the underlying Delaunay structure is already known and constant.

\subsection{Computing Prismatoid Volume}
\label{PismaVolume}
Given a large number of
truncated prismatoids, each specified by the ordered vertices of its two
faces, their volumes can be computed very efficiently in parallel on the GPU.
First, we ensure that both faces of each prismatoid have the same number of
vertices by padding the smaller face with duplicate vertices. With the information from the intersection with the neighbour tunnel, we can ensure that the vertex orderings of the two faces are consistently
aligned. We then
fan-triangulate each polygon and add edges between corresponding vertices of
the two faces, which decomposes the prismatoid into tetrahedra. The volume of
the prismatoid is obtained by summing the (signed) volumes
$V = \tfrac{1}{6}\,\bigl|(\mathbf{b}-\mathbf{a})\cdot\bigl((\mathbf{c}-\mathbf{a})
\times(\mathbf{d}-\mathbf{a})\bigr)\bigr|$ of all resulting tetrahedra. Every
step of this procedure is highly parallel. In particular, all polygon pairs
that share the same maximal vertex count can use exactly the same
tetrahedralization pattern, so every thread executes an identical instruction
sequence and warp divergence is eliminated entirely. An example of a tetrahedralized prismatoid is shown in Figure~\ref{fig:Tetrahedralization}.

\begin{figure}[h]
    \centering
        \centering
        \includegraphics[width=0.4\linewidth]{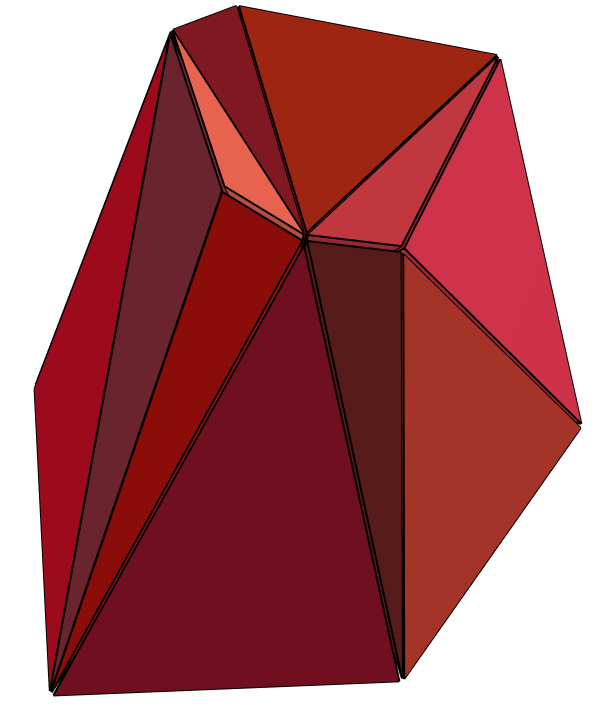}
    \caption{Tetrahedralization of a truncated prismatoid.}
    \label{fig:Tetrahedralization}
\end{figure}

%\todo{Describe the data structure for efficiently finding $\mathcal{H}$}

%\section{Radfoam Visual Evaluation}
%Scenes trained using Radfoam provide a perfect test bed for our algorithm, as they model a complex, continuous radiance field function using up to millions of 3D Voronoi cells. A reduction algorithm can be used to compress the representation and increase performance. 
%
%In the context of Radfoam, the underlying function $f$ represents the radiance field by mapping each 3D point to an RGB color $\mathbf{c}(\mathbf{x})$ and density $\sigma(\mathbf{x})$. For this application, it makes sense to slightly modify the objective function, since colour difference only matters if there is also density:
%\begin{equation}
%\int_{\mathbb{R}^{d_i}} I(\mathbf{x}) \Big( \sigma(\mathbf{x}) \cdot \big(\mathbf{c}(\mathbf{x}) -  \hat{\mathbf{c}}(\mathbf{x})\big)^2 + \lambda_{\sigma} \big( \sigma(\mathbf{x}) - \hat{\sigma}(\mathbf{x}) \big)^2\Big) \, d\mathbf{x}    
%\end{equation}
%with density weighting hyper-parameter $\lambda_{\sigma}$. We also include an importance term $I(\mathbf{x})$, which measures how important a cell is for the representation of the scene. This is computed as the sum of all blending weights to all of the camera views from which the scene was trained. 
%
%By default, Radfoam models view-dependent effects using per-cell spherical harmonic coefficients. For our test, we deactivate these view-dependent effects, as they would unnecessarily complicate the optimization objective.
\section{ISE Curves}
\label{app:ISE_curves}
Figure~\ref{fig:ISECurves} shows error curves for evaluated methods on a selection of datasets. \emph{Ours (Fast)} matches \emph{Ours (Exact)} very closely throughout the entire reduction. Where the relative performance of baselines varies significantly depending on the input data, VoroDuce achieves best accuracy throughout the entire reduction with few exceptions.
\begin{figure*}
    \centering
    \includegraphics[width=0.32\linewidth]{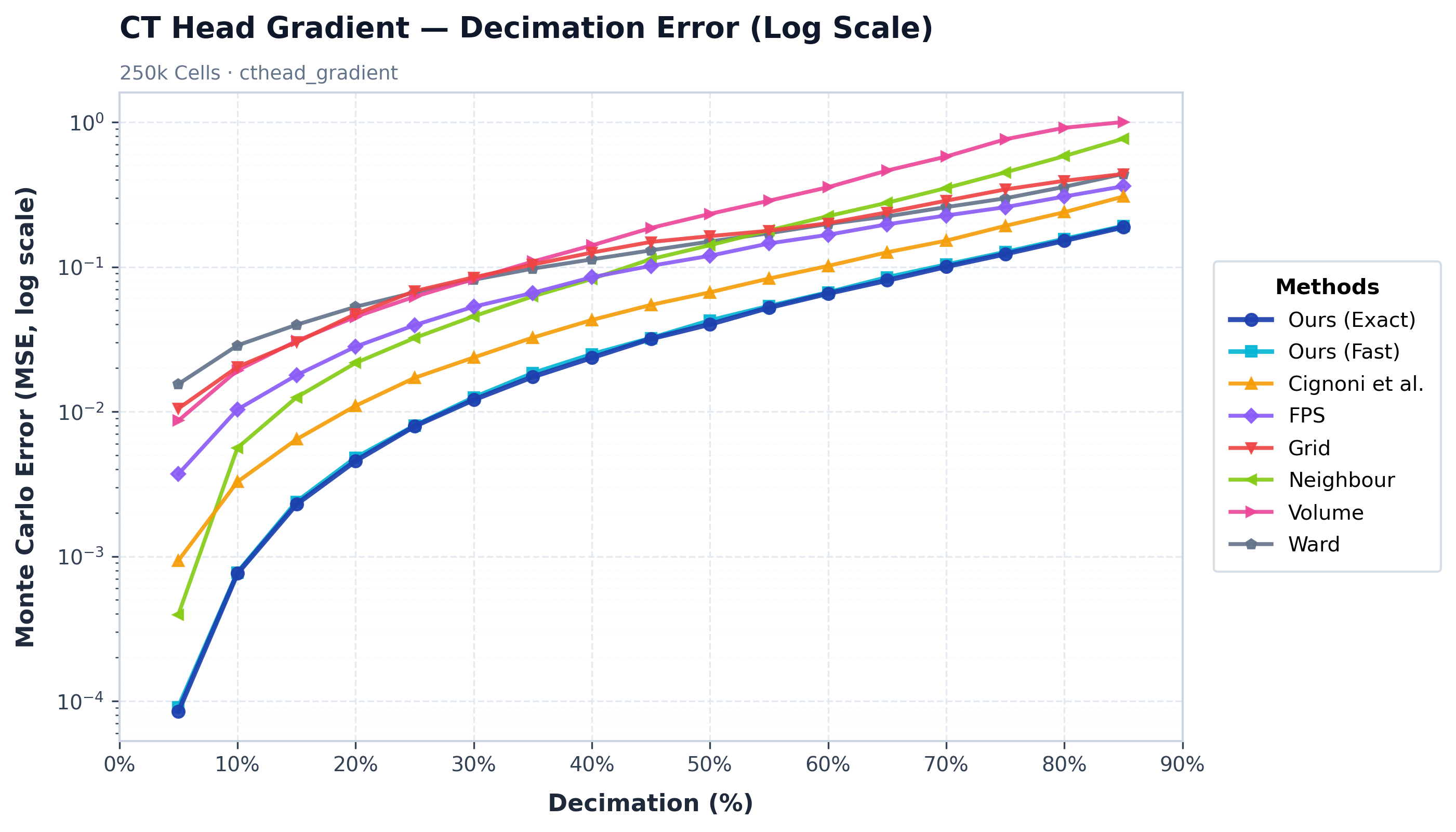}
    \includegraphics[width=0.32\linewidth]{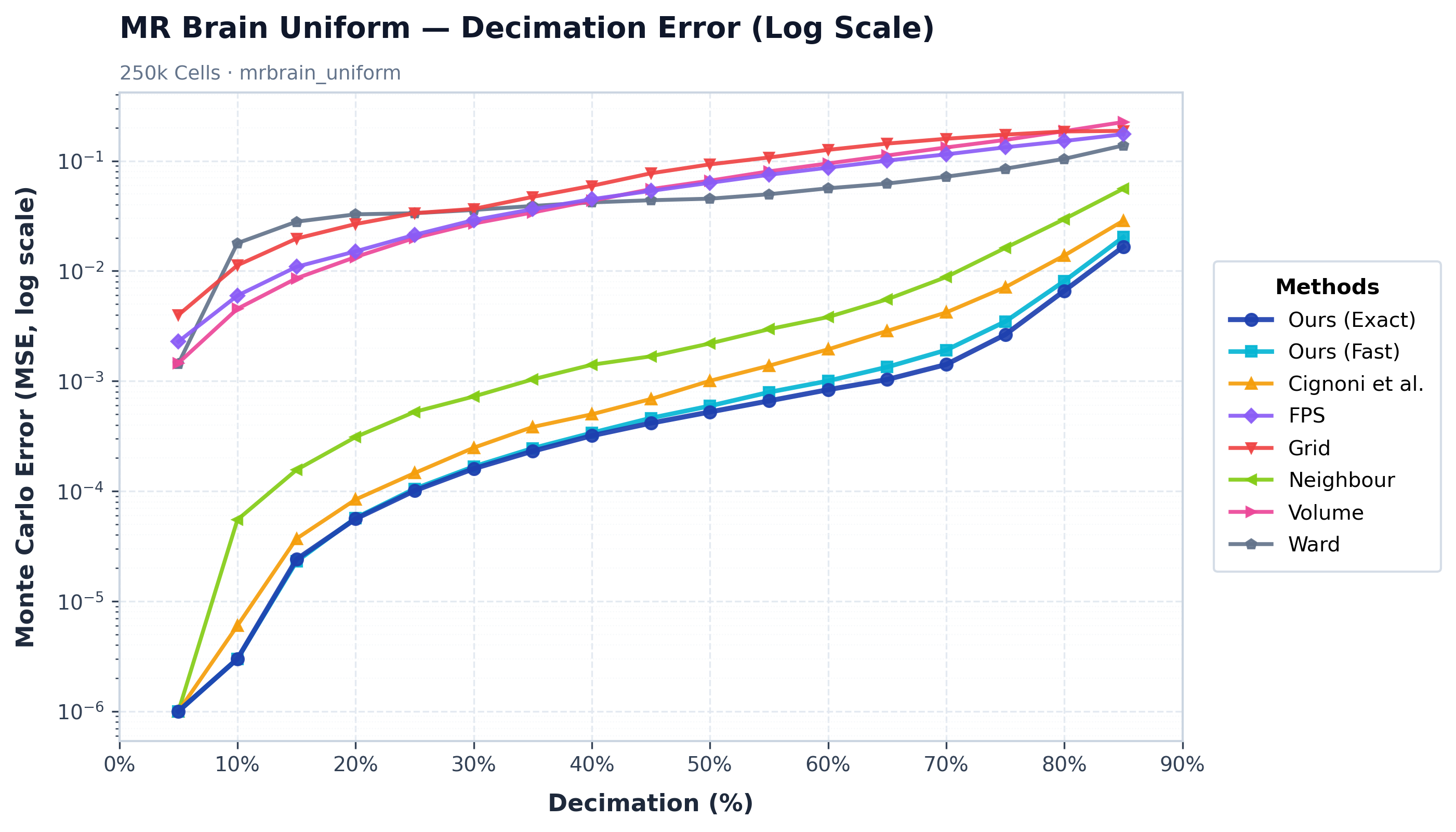}
    \includegraphics[width=0.32\linewidth]{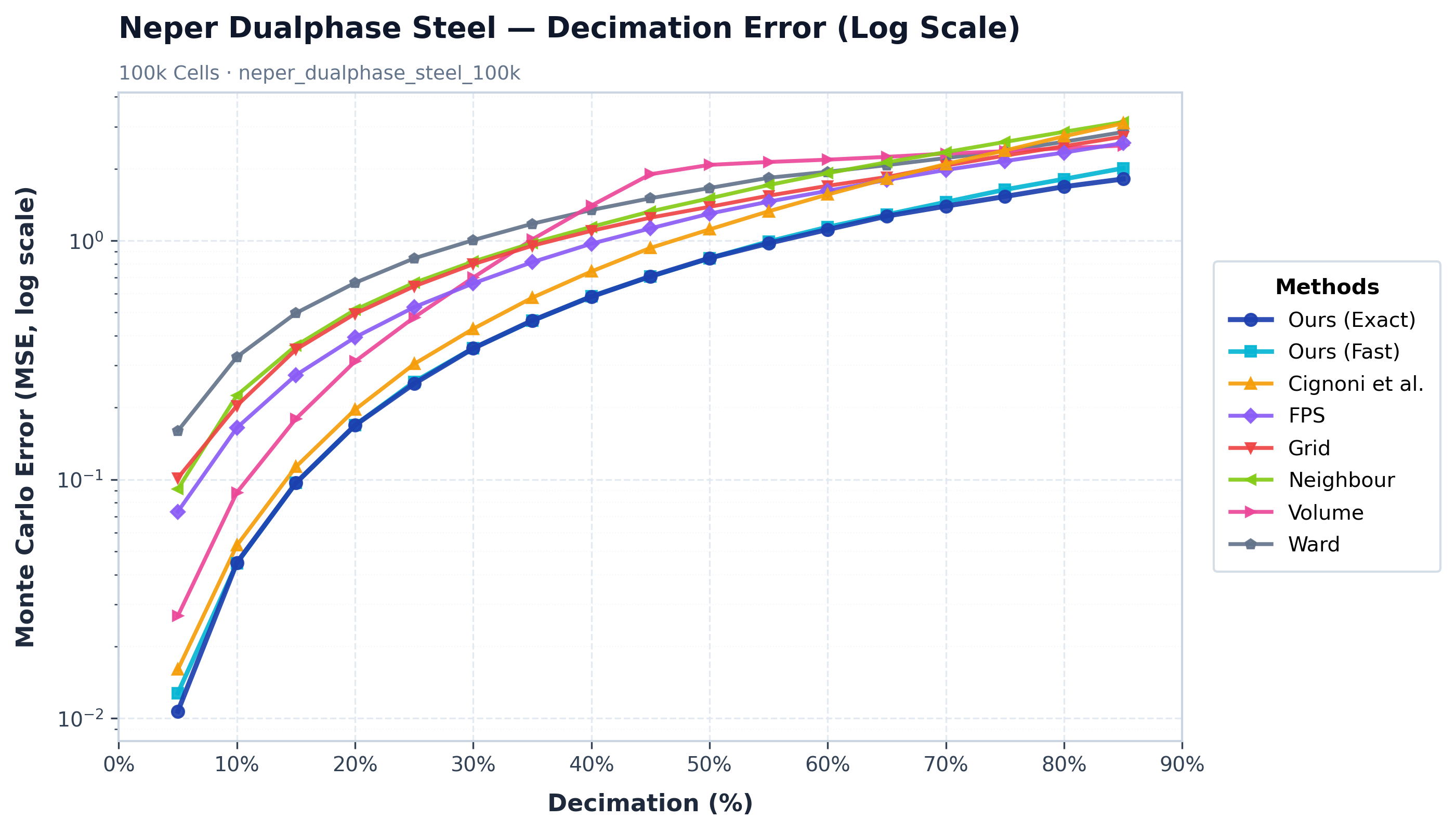}
    \includegraphics[width=0.32\linewidth]{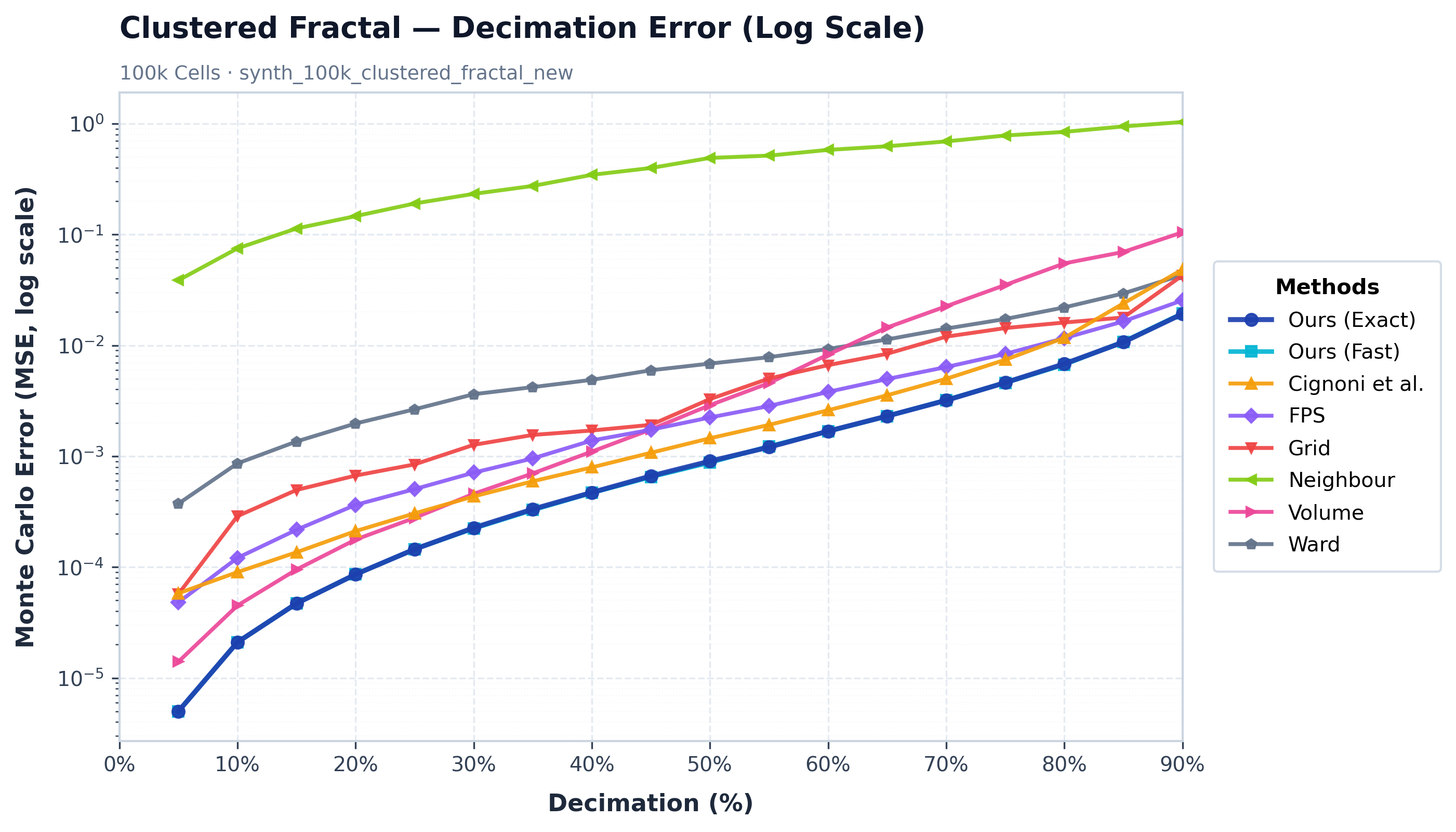}
    \includegraphics[width=0.32\linewidth]{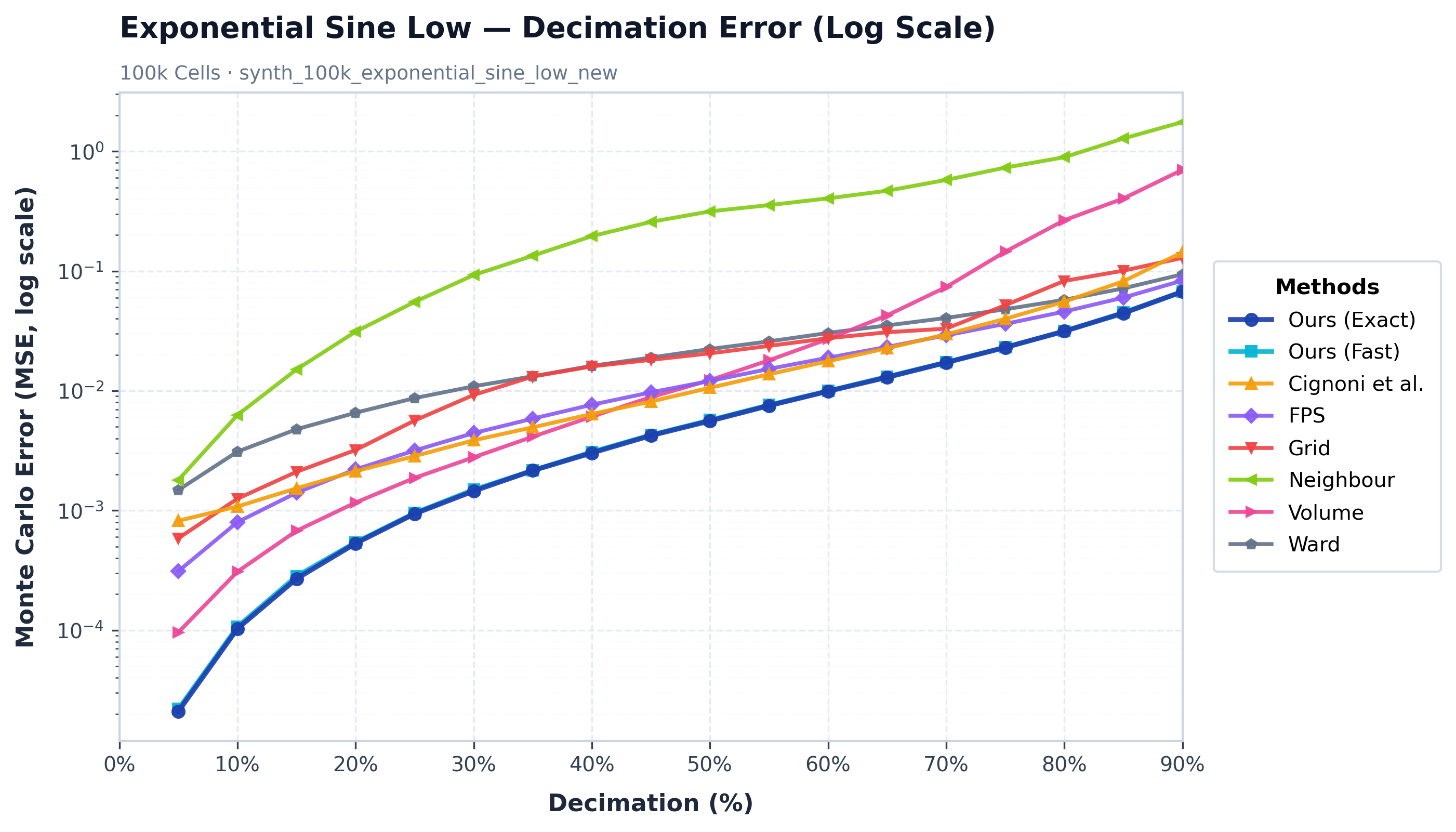}
    \includegraphics[width=0.32\linewidth]{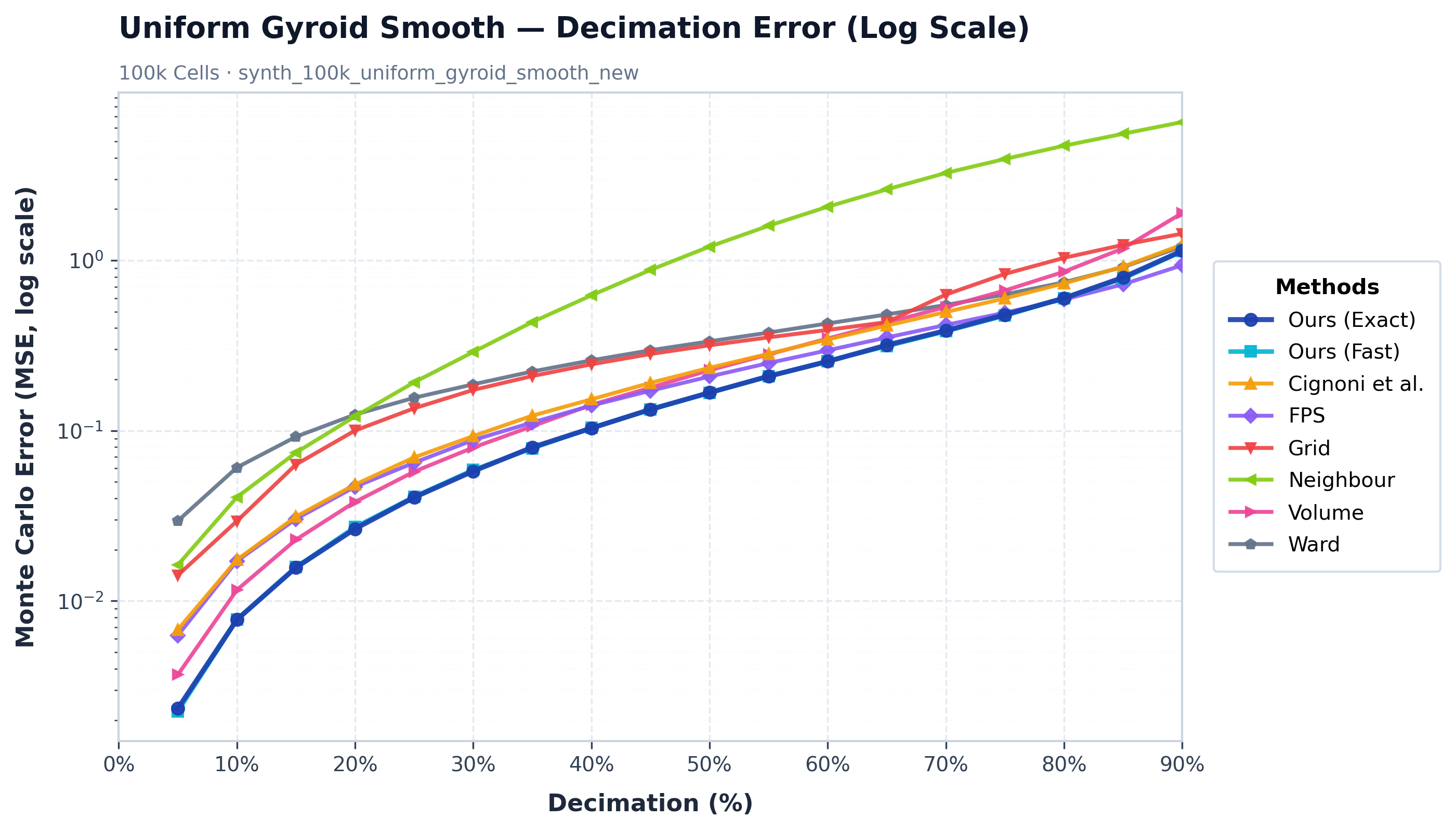}
    \caption{Error curves over evaluated methods on a selection of datasets. Log scale is used to better differentiate small differences. }
    \label{fig:ISECurves}
\end{figure*}

\section{Empirical Frequency of Delaunay flips}
\label{app:Flip_Frequency}
During cell merging, edge flips in the Delaunay triangulation cause special cases and inaccuracies in our approximations. In this section, we evaluate the frequency of the explored cases of edge flips.

The case in Equation \ref*{eq6}, where $\hat{V}_{new} \cap V_n \neq \emptyset$ for some cell $n \in \mathcal{S} \setminus (\mathcal{N}(i) \cup \mathcal{N}(j))$, i.e. the new cell gains a neighbour that was not connected to either of its parents, is not detected by our approximation. However, it is also infrequent, only accounting for $3.4\%$ of neighbours of $\hat{V}_{new}$ on average. Further, the intersection volume $Vol(\hat{V}_{new} \cap V_n)$ is small, representing less than 0.1\% of the $Vol(\hat{V}_{new})$ on average. We tested explicit handling of this case in the exact volume method, but found no statistically significant improvement in quality.

The only case we directly detect is case 3, where $(k, new) \not \in \mathcal{T}_\mathcal{E}$, which occurs for about 30\% of neighbouring cells. Note that these cells are also intrinsically small in volume and error, accounting for only $8.2\%$ of the neighbour volumes on average. 

To test the frequency of the undetected flip types, we run a simulation of $40,000$ merges on random neighbour cells with $20,000$ uniformly sampled points. We find that for any cell $k \in \mathcal{N}(i) \cup \mathcal{N}(j)$ adjacent to merging cells $i,j$, new neighbours caused by case 2 account for $4.27\%$ of cells in $\mathcal{N}(k) \cap (\mathcal{N}(i) \cup \mathcal{N}(j))$ and pinch vertices are generated  by $6.42\%$ of neighbours.

\subsection{Proofs}
\label{app:proofs}
In this section, we provide detailed proofs for the Propositions and Theorems in the paper.

\begin{theorem}[Exactness without flips (restated)]
\label{thm:noflip}
Let $k\in\mathcal{N}(i)\setminus\mathcal{N}(j)$ and $\hat{\mathcal{T}} = \mathcal{T}/(i,j)$.
Assume $\mathcal{S}\cup\{\mathbf{x}_{new}\}$ is in general position, and $V_i$ and
$\hat V_{new}$ are bounded. Then
\begin{enumerate}
  \item $\Pi_{\mathbf{x}_k,\mathbf{x}_{new}}\cap \textstyle\bigcap \mathcal{H}_{k,i} = \hat{\mathcal{F}}_{k,new}$;
  \item if $\mathcal{F}_{new} \cap \mathcal{F}_{orig} = \emptyset$, then
        $G_k = \conv\!\big(\mathcal{F}_{orig}\cup\mathcal{F}_{new}\big)$ and
        $L_k = \emptyset$ (or vice versa, if $\sigma_i<0$ on $\mathcal{F}_{orig}$);
  \item otherwise, $G_k = \conv\!\big(\mathcal{F}^+_{orig}\cup\mathcal{F}^+_{new}\big)$ and
        $L_k = \conv\!\big(\mathcal{F}^-_{orig}\cup\mathcal{F}^-_{new}\big)$.
        Both are truncated prismatoids with positive volume.
\end{enumerate}
\end{theorem}
 
\begin{proof}
We proceed in six steps.
 
\emph{Step 1: Common walls.}
Since $k\notin\mathcal{N}(j)$, no simplex of $\mathcal{T}$ contains both $k$
and $j$. The contraction therefore maps the edge $(k,i)$ to $(k,new)$ and
leaves all other edges at $k$ unchanged, so
$\hat{\mathcal{N}}(k) = (\mathcal{N}(k)\setminus i)\cup new$, with $\mathcal{N}$ denoting the neighbourhood function after the merge. Let
\begin{equation*}
  W := \textstyle\bigcap_{l\in\mathcal{N}(k)\setminus\{i\}}\Pi^+_{\mathbf{x}_k,\mathbf{x}_l}
\end{equation*}
denote the walls that $V_k$ and $\hat V_k$ have in common. Then
$V_k = W\cap\Pi^+_{\mathbf{x}_k,\mathbf{x}_i}$ and
$\hat V_k = W\cap\Pi^+_{\mathbf{x}_k,\mathbf{x}_{new}}$. Hence
$\mathcal{F}_{k,i} = W\cap\Pi_{\mathbf{x}_k,\mathbf{x}_i}$ and
$\hat{\mathcal{F}}_{k,new} = W\cap\Pi_{\mathbf{x}_k,\mathbf{x}_{new}}$.
Recall the definition of lost and gained neighbour volumes:
\begin{align}
  G_k &:= V_i \cap \hat V_k      && \text{(volume $V_k$ gains from $V_i$)}, \label{eq:gain_region2}\\
  L_k &:= V_k \cap \hat V_{new}  && \text{(volume $V_k$ loses to $\hat V_{new}$)}. \label{eq:loss_region2}
\end{align}
All combined, this results in:
\begin{align}
  G_k &= W\cap\Pi^+_{\mathbf{x}_i,\mathbf{x}_k}\cap\Pi^+_{\mathbf{x}_k,\mathbf{x}_{new}}, \label{eq:gain_W}\\
  L_k &= W\cap\Pi^+_{\mathbf{x}_k,\mathbf{x}_i}\cap\Pi^+_{\mathbf{x}_{new},\mathbf{x}_k}. \label{eq:loss_W}
\end{align}
Because $\mathcal{N}(k)\cap\mathcal{N}(i)\subseteq\mathcal{N}(k)\setminus\{i\}$,
we also have $W\subseteq\bigcap\mathcal{H}_{k,i}$.
 
\emph{Step 2: Sign identities.}
Define the affine function
\begin{equation}
  \sigma_i(\mathbf{p}) := \|\mathbf{p}-\mathbf{x}_{new}\|^2 - \|\mathbf{p}-\mathbf{x}_i\|^2 .
\end{equation}
Its zero set is the bisector $\Pi_{\mathbf{x}_i,\mathbf{x}_{new}}$.

On $\Pi_{\mathbf{x}_k,\mathbf{x}_i}$ we have
$\|\mathbf{p}-\mathbf{x}_i\| = \|\mathbf{p}-\mathbf{x}_k\|$, so
$\sigma_i(\mathbf{p}) = \|\mathbf{p}-\mathbf{x}_{new}\|^2-\|\mathbf{p}-\mathbf{x}_k\|^2$. Hence
\begin{equation}
  \begin{aligned}
  \mathbf{p}\in\Pi^+_{\mathbf{x}_k,\mathbf{x}_{new}}&\iff\sigma_i(\mathbf{p})\ge 0,\\
  \mathbf{p}\in\Pi^+_{\mathbf{x}_{new},\mathbf{x}_k}&\iff\sigma_i(\mathbf{p})\le 0.
  \end{aligned}
  \label{eq:sign_orig}
\end{equation}
On $\Pi_{\mathbf{x}_k,\mathbf{x}_{new}}$ we have
$\|\mathbf{p}-\mathbf{x}_{new}\| = \|\mathbf{p}-\mathbf{x}_k\|$, so
$\sigma_i(\mathbf{p}) = \|\mathbf{p}-\mathbf{x}_k\|^2-\|\mathbf{p}-\mathbf{x}_i\|^2$. Hence
\begin{equation}
  \begin{aligned}
  \mathbf{p}\in\Pi^+_{\mathbf{x}_i,\mathbf{x}_k}&\iff\sigma_i(\mathbf{p})\ge 0,\\
  \mathbf{p}\in\Pi^+_{\mathbf{x}_k,\mathbf{x}_i}&\iff\sigma_i(\mathbf{p})\le 0.
  \end{aligned}
  \label{eq:sign_new}
\end{equation}
 
\emph{Step 3: Part (1).}
Let $\mathcal{C}_{k,i} = \bigcap\mathcal{H}_{k,i}$.
Since $W\subseteq\mathcal{C}_{k,i}$, Step~1 gives
$\hat{\mathcal{F}}_{k,new}\subseteq\Pi_{\mathbf{x}_k,\mathbf{x}_{new}}\cap\mathcal{C}_{k,i}=\mathcal{F}_{new}$.
 
For the converse, note that the edges of a Voronoi face are dual to the
Delaunay triangles containing its dual edge. The contraction maps the triangles
$\{k,i,l\}$ of $\mathcal{T}$ bijectively onto the triangles $\{k,new,l\}$ of
$\hat{\mathcal{T}}$; since $k\notin\mathcal{N}(j)$, no triangle $\{k,j,l\}$
exists. Hence
$\mathcal{E}(\hat{\mathcal{F}}_{k,new}) = \mathcal{E}(\mathcal{F}_{k,i})\subseteq\mathcal{N}(k)\cap\mathcal{N}(i)$.
The face $\hat{\mathcal{F}}_{k,new}$ is a non-empty convex polygon, since
$(k,new)\in\hat{\mathcal{T}}_\mathcal{E}$. Within its plane, it is the
intersection of the half-planes supporting its edges. Each of these is the
restriction of a half-space in $\mathcal{H}_{k,i}$, so
$\mathcal{F}_{new}\subseteq\hat{\mathcal{F}}_{k,new}$.
 
The same argument with $\Pi_{\mathbf{x}_k,\mathbf{x}_i}$ in place of
$\Pi_{\mathbf{x}_k,\mathbf{x}_{new}}$ gives $\mathcal{F}_{orig}=\mathcal{F}_{k,i}$.
Therefore
\begin{equation}
  \mathcal{F}_{orig} = W\cap\Pi_{\mathbf{x}_k,\mathbf{x}_i},
  \qquad
  \mathcal{F}_{new} = W\cap\Pi_{\mathbf{x}_k,\mathbf{x}_{new}}.
  \label{eq:faces_W}
\end{equation}
 
\emph{Step 4: The faces meet exactly where $\sigma_i$ vanishes.}
By definition,
$\mathcal{F}_{orig}\cap\mathcal{F}_{new} = \Pi_{\mathbf{x}_k,\mathbf{x}_i}\cap\Pi_{\mathbf{x}_k,\mathbf{x}_{new}}\cap\mathcal{C}_{k,i}$.
A point of $\Pi_{\mathbf{x}_k,\mathbf{x}_i}$ lies on $\Pi_{\mathbf{x}_k,\mathbf{x}_{new}}$
if and only if $\sigma_i=0$, by Step~2. The same holds with the roles of the two
planes exchanged. Hence
\begin{equation}
  \mathcal{F}_{orig}\cap\mathcal{F}_{new}
  = \mathcal{F}_{orig}\cap\{\sigma_i=0\}
  = \mathcal{F}_{new}\cap\{\sigma_i=0\}.
  \label{eq:meet_zero}
\end{equation}
 
\emph{Step 5: Vertices of $G_k$ and $L_k$.}
Both sets are bounded, since $G_k\subseteq V_i$ and $L_k\subseteq\hat V_{new}$.
A bounded convex polyhedron is the convex hull of its vertices. We show that
every vertex lies on $\Pi_{\mathbf{x}_k,\mathbf{x}_i}$ or on
$\Pi_{\mathbf{x}_k,\mathbf{x}_{new}}$.
 
Let $\mathbf{v}$ be a vertex on neither plane. By Equations~\ref{eq:gain_W}
and~\ref{eq:loss_W}, it lies on at least three walls
$\Pi_{\mathbf{x}_k,\mathbf{x}_{l_1}},\Pi_{\mathbf{x}_k,\mathbf{x}_{l_2}},\Pi_{\mathbf{x}_k,\mathbf{x}_{l_3}}$
with $l_m\in\mathcal{N}(k)\setminus\{i\}$. In particular, $l_m\notin\{i,j,new\}$.
So $\mathbf{v}$ is equidistant from $\mathbf{x}_k$ and
$\mathbf{x}_{l_1},\mathbf{x}_{l_2},\mathbf{x}_{l_3}$, and the tetrahedron
$\tau=\{k,l_1,l_2,l_3\}$ is incident to none of $i,j,new$. Hence
$\tau\in\mathcal{T}/(i,j)$ if and only if $\tau\in\mathcal{T}$.
\begin{itemize}
  \item \emph{Vertex of $G_k$.} Since $\mathbf{v}\in\hat V_k$, the four sites
        of $\tau$ are nearest to $\mathbf{v}$ among the new sites. Hence
        $\tau\in\hat{\mathcal{T}}$. Since $\mathbf{v}\notin\Pi_{\mathbf{x}_k,\mathbf{x}_i}$,
        it lies strictly on $\mathbf{x}_i$'s side, so $\mathbf{x}_i$ is
        strictly inside the circumsphere of $\tau$. Hence $\tau\notin\mathcal{T}$.
  \item \emph{Vertex of $L_k$.} Since $\mathbf{v}\in V_k$, the four sites of
        $\tau$ are nearest to $\mathbf{v}$ among the old sites. Hence
        $\tau\in\mathcal{T}$. Since $\mathbf{v}\notin\Pi_{\mathbf{x}_k,\mathbf{x}_{new}}$,
        $\mathbf{x}_{new}$ is strictly inside the circumsphere of $\tau$. Hence
        $\tau\notin\hat{\mathcal{T}}$.
\end{itemize}
Both contradict $\hat{\mathcal{T}}=\mathcal{T}/(i,j)$.
 
\emph{Step 6: Parts (2) and (3).}
By Step~5, $G_k$ and $L_k$ are the convex hulls of their intersections with the
two planes. We identify these intersections using Equations~\ref{eq:gain_W},
\ref{eq:loss_W}, \ref{eq:faces_W}, \ref{eq:sign_orig} and~\ref{eq:sign_new}:
\begin{align*}
  G_k\cap\Pi_{\mathbf{x}_k,\mathbf{x}_i}      &= \mathcal{F}_{orig}\cap\Pi^+_{\mathbf{x}_k,\mathbf{x}_{new}} = \mathcal{F}^+_{orig},\\
  G_k\cap\Pi_{\mathbf{x}_k,\mathbf{x}_{new}}  &= \mathcal{F}_{new}\cap\Pi^+_{\mathbf{x}_i,\mathbf{x}_k} = \mathcal{F}^+_{new},\\
  L_k\cap\Pi_{\mathbf{x}_k,\mathbf{x}_i}      &= \mathcal{F}_{orig}\cap\Pi^+_{\mathbf{x}_{new},\mathbf{x}_k} = \mathcal{F}^-_{orig},\\
  L_k\cap\Pi_{\mathbf{x}_k,\mathbf{x}_{new}}  &= \mathcal{F}_{new}\cap\Pi^+_{\mathbf{x}_k,\mathbf{x}_i} = \mathcal{F}^-_{new}.
\end{align*}
Hence, without further assumptions,
\begin{equation}
  \begin{aligned}
  G_k &= \conv\!\big(\mathcal{F}^+_{orig}\cup\mathcal{F}^+_{new}\big),\\
  L_k &= \conv\!\big(\mathcal{F}^-_{orig}\cup\mathcal{F}^-_{new}\big).
  \end{aligned}
  \label{eq:general_hulls}
\end{equation}
 
We next show that $G_k$ has positive volume whenever $\sigma_i>0$ at some point
$\mathbf{p}$ in the relative interior of $\mathcal{F}_{orig}$. By general
position, $\mathbf{p}$ lies strictly inside every wall of $W$, and it lies
strictly inside $\Pi^+_{\mathbf{x}_k,\mathbf{x}_{new}}$ by
Equation~\ref{eq:sign_orig}. Hence a small ball $B$ around $\mathbf{p}$
satisfies all constraints of Equation~\ref{eq:gain_W} except
$\Pi^+_{\mathbf{x}_i,\mathbf{x}_k}$, and the half-ball
$B\cap\Pi^+_{\mathbf{x}_i,\mathbf{x}_k}\subseteq G_k$ has positive volume.
Symmetrically, $L_k$ has positive volume whenever $\sigma_i<0$ at a point in
the relative interior of $\mathcal{F}_{orig}$.
 
\emph{Part (2).} Suppose $\mathcal{F}_{orig}\cap\mathcal{F}_{new}=\emptyset$.
By Equation~\ref{eq:meet_zero}, $\sigma_i$ has no zero on $\mathcal{F}_{orig}$
or on $\mathcal{F}_{new}$. Both faces are convex, hence connected, and $\sigma_i$
is continuous. So $\sigma_i$ has a constant strict sign on each face.
 
Assume $\sigma_i>0$ on $\mathcal{F}_{orig}$. Then
$\mathcal{F}^+_{orig}=\mathcal{F}_{orig}$ and $\mathcal{F}^-_{orig}=\emptyset$,
and $G_k$ has positive volume. If $\mathcal{F}^+_{new}$ were empty,
Equation~\ref{eq:general_hulls} would give
$G_k = \mathcal{F}_{orig}\subset\Pi_{\mathbf{x}_k,\mathbf{x}_i}$, which has zero
volume. Hence $\sigma_i>0$ somewhere on $\mathcal{F}_{new}$, and therefore
everywhere on it. Thus $\mathcal{F}^+_{new}=\mathcal{F}_{new}$ and
$\mathcal{F}^-_{new}=\emptyset$. Equation~\ref{eq:general_hulls} yields
$G_k=\conv(\mathcal{F}_{orig}\cup\mathcal{F}_{new})$ and $L_k=\emptyset$.
 
The case $\sigma_i<0$ follows by exchanging the roles of $G_k$ and $L_k$.
 
\emph{Part (3).} Suppose $\mathcal{F}_{orig}\cap\mathcal{F}_{new}\neq\emptyset$.
By Equation~\ref{eq:meet_zero}, $\sigma_i$ vanishes somewhere on
$\mathcal{F}_{orig}$. By general position of $\mathcal{S}\cup\{\mathbf{x}_{new}\}$,
it does not vanish at a vertex of $\mathcal{F}_{orig}$. A vertex of
$\mathcal{F}_{orig}$ is the circumcentre of a tetrahedron $\{k,i,l,l'\}$, and
$\sigma_i=0$ there would place $\mathbf{x}_{new}$ on that circumsphere. Hence
the zero line of the affine function $\sigma_i$ crosses the relative interior of
$\mathcal{F}_{orig}$, and $\sigma_i$ takes both signs there. By the argument
above, both $G_k$ and $L_k$ have positive volume.
Equation~\ref{eq:general_hulls} is the claimed identity. Each region has all
its vertices on the two planes, with convex bases, and is therefore a truncated
prismatoid.
 
The final statement follows from parts (2) and (3), since $\sigma_i$ has a
constant strict sign on $\mathcal{F}_{orig}$ if and only if the faces are
disjoint.
\end{proof}

\begin{prop}[Rest Volumes (restated)]
Let $\mathcal{S}$ be a set of samples, $i, j \in \mathcal{T}_\mathcal{E} (\mathcal{S})$ and $\hat{\mathcal{S}} = \left(\mathcal{S} \setminus \{(\mathbf{x}_i, \mathbf{y}_i), (\mathbf{x}_j, \mathbf{y}_j)\} \right) \cup (\mathbf{x}_{new}, \mathbf{y}_{new})$. Then,

\begin{equation} Vol(V_i) - \sum_{k \in \mathcal{N}(i)} \text{Vol}(V_i \cap \hat{V}_k)  = \text{Vol}(\hat{V}_{new} \cap V_i).\end{equation}    
\end{prop}

\begin{proof}
Because all new Voronoi cells $\hat{V}_k$, $\hat{V}_{new}$ have empty intersection,
\begin{equation}
    \sum_{k \in \mathcal{N}(i)} \text{Vol}(V_i \cap \hat{V}_k) + \text{Vol}(\hat{V}_{new} \cap V_i) \leq Vol(V_i).
    \label{less_than}
\end{equation} 
Edge collapse only changes the nearest neighbour of points in $V_i$, $V_j$ or their neighbours, which implies 
\begin{equation}
    V_i \subset \Big( \hat{V}_{new} \cup\bigcup_{k \in  \mathcal{N}(i)} \hat{V}k \Big)
    \label{union}
.\end{equation} From \ref{union}, it follows that

\begin{equation}
    \sum_{k \in \mathcal{N}(i)} \text{Vol}(V_i \cap \hat{V}_k) + \text{Vol}(\hat{V}_{new} \cap V_i) \geq Vol(V_i).
    \label{greater_than}
\end{equation}
Equations \ref{less_than} and \ref{greater_than} combine to:
\begin{equation}
     \sum_{k \in \mathcal{N}(i)} \text{Vol}(V_i \cap \hat{V}_k) + \text{Vol}(\hat{V}_{new} \cap V_i) = Vol(V_i)
\end{equation}
and re-arranging proves the result.
\end{proof}

\end{document}